\documentclass[reqno]{amsart}
\usepackage{fullpage,amsfonts,amssymb,amsthm,mathrsfs,graphicx,color,framed,esint}
\usepackage{cancel,bm}
\usepackage{tikz}
\usetikzlibrary{decorations.markings}
\usepackage{booktabs,longtable}

\usepackage[initials]{amsrefs}

\usepackage{multicol}

\usepackage[colorlinks=true, pdfstartview=FitV, linkcolor=blue, citecolor=blue, urlcolor=blue]{hyperref}
\definecolor{shadecolor}{rgb}{0.95, 0.95, 0.86}

\definecolor{ao}{rgb}{0.0, 0.4, 0.0}

\renewcommand{\d}{{\mathrm d}}

\newcommand{\im}{\mathrm{i}}
\newcommand{\e}{\mathrm{e}}

\def\res{\mathop{\mathrm{res}}\limits}
\def\tr{\mathop{\mathrm{tr}}\limits}

\numberwithin{equation}{section}

\newtheorem{theo}{Theorem}[section]

\newtheorem{lem}[theo]{Lemma}
\newtheorem{rem}[theo]{Remark}
\newtheorem{problem}[theo]{Riemann-Hilbert Problem}
\newtheorem{prob}[theo]{Hilbert Boundary Value Problem}
\newtheorem{remark}[theo]{Remark}
\newtheorem{prop}[theo]{Proposition} 
\newtheorem{cor}[theo]{Corollary} 
\newtheorem{definition}[theo]{Definition}

\newtheorem{assu}[theo]{Assumption}

\begin{document}

\title[What is $\ldots$ a RHP?]{On the origins of Riemann-Hilbert problems in mathematics}

\author{Thomas Bothner}
\address{School of Mathematics, University of Bristol, Fry Building, Woodland Road, Bristol, BS8 1UG, United Kingdom}
\email{thomas.bothner@bristol.ac.uk}



\thanks{The author is greatly indebted to A. Its, P. Bleher, P. Deift, A. Kuijlaars, M. Bertola, J. Baik and P. Miller for countless discussions on Riemann-Hilbert problems over the past 10 years. This article is dedicated to Boris Dubrovin (1950-2019) and his outstanding legacy in the field of mathematical physics. The author would also like to thank the referees for their constructive suggestions which improved the paper in a variety of ways.}

\begin{abstract}
This article is firstly a historic review of the theory of Riemann-Hilbert problems with particular emphasis placed on their original appearance in the context of Hilbert's 21st problem and Plemelj's work associated with it. The secondary purpose of this note is to invite a new generation of mathematicians to the fascinating world of Riemann-Hilbert techniques and their modern appearances in nonlinear mathematical physics. We set out to achieve this goal with six examples, including a new proof of the integro-differential Painlev\'e-II formula of Amir, Corwin, Quastel \cite{ACQ} that enters in the description of the KPZ crossover distribution. Parts of this text are based on the author's Szeg\H{o} prize lecture at the $15$th International Symposium on Orthogonal Polynomials, Special Functions and Applications (OPSFA) in Hagenberg, Austria.
\end{abstract}

\date{\today}

\dedicatory{Dedicated to the memory of Boris Anatolievich Dubrovin}
\maketitle

\section{Introduction}
Riemann-Hilbert problems, in short RHPs, are useful in the analysis of orthogonal polynomials (OP), special functions (SF) and thus several applications (A) in mathematics or physics. Still, these modern appearances of RHPs are quite detached from the original meaning of the term RHP and it is our primary objective to pin down the occurrence of the first problem in mathematics that was coined a RHP. Once done, we will follow Plemelj's traces whose work on the original RHP gave rise to an analytic apparatus, the \textit{Riemann-Hilbert techniques} or \textit{Riemann-Hilbert methods}, which found various applications in the OPSFA realm. We will showcase and advertise the latter modern aspects of the theory with examples from mathematical physics, in particular examples from nonlinear wave and integrable systems theory, statistical mechanics, random matrix theory and integrable probability.
\section{The original Riemann-Hilbert problem}
In order to answer the question ``What is $\ldots$ a Riemann-Hilbert problem" one should go back to 1900, the year of the second international congress of mathematics.  Held at the turn of the century in Paris, the ICM 1900 proved to be highly influential for the development of 20th century mathematics, in particular since Hilbert delivered his famous address in the section \textit{Bibliographie et Histoire. Enseignement et methodes}. During his lecture he presented (because of time constraints) ten mathematical problems ``from the discussion of which an advancement of science may be expected", see \cite[page $445$]{H}. The full, twenty-three long, list of Hilbert's problems appeared in German print even before the French proceedings of the congress. Scrolling through this list we eventually encounter problem $\#21$, the main focus point of this section, cf. \cite[page 289-290]{H0}:\bigskip
\begin{quote}
\begin{center} 
\small{\textsl{21. Beweis der Existenz linearer Differentialgleichungen mit vorgeschriebener Monodromiegruppe}}\bigskip\\
\end{center}
\noindent\footnotesize{\textsl{Aus der Theorie der linearen Differentialgleichungen mit einer unabh\"angigen Ver\"anderlichen $z$ m\"ochte ich auf ein wichtiges Problem
hinweisen, welches wohl bereits Riemann im Sinne gehabt hat, und welches darin besteht, zu zeigen, da\ss\,es stets eine lineare Differentialgleichung der Fuchsschen Klasse
mit gegebenen singul\"aren Stellen und einer gegebenen Monodromiegruppe giebt. Die Aufgabe verlangt also die Auffinding von $n$ Functionen der Variabeln $z$, die sich \"uberall in der complexen $z$-Ebene regul\"ar verhalten, au\ss er etwa in den gegebenen singul\"aren Stellen: in diesen d\"urfen sie nur von endlich hoher Ordnung unendlich werden und beim Umlauf der Variabeln $z$ um dieselben erfahren sie die gegebenen linearen Substitutionen. Die Existenz solcher Differentialgleichungen ist durch Constantenz\"ahlung wahrscheinlich gemacht worden, doch gelang der strenge Beweis bisher nur in dem besonderen Falle, wo die Wurzeln der Fundamentalgleichungen der
gegebenen Substitutionen s\"amtlich vom absoluten Betrage 1 sind. Diesen Beweis hat L. Schlesinger auf Grund der Poincar\'eschen Theorie der Fuchsschen $\zeta$-Functionen erbracht. Es w\"urde offenbar die Theorie der linearen Differentialgleichungen ein wesentlich abgeschlosseneres Bild zeigen, wenn die allgemeine
Erledigung des bezeichneten Problems gel\"ange.}}\smallskip
\end{quote}\bigskip
In this problem Hilbert asks the reader to \textit{show that there always exists a linear differential equation of the Fuchsian class, with given singular points and monodromic group}. Before reviewing some of the necessary background material in Section \ref{sec:3} below it will be important to point out the following peculiar wording in Hilbert's original formulation: \textit{The problem requires the production of $n$ functions of the variable $z$, regular throughout the complex $z$ plane except at the given singular points; at these points the functions may become infinite of only finite order}. Thus, Hilbert does not use the word pole singularity - he refers to what is nowadays called a regular singularity, cf. \cite[Definition $16.2$]{IY} - and he also never speaks of Fuchsian linear systems. These subtleties in the formulation of the problem played an important role as we will soon learn.\smallskip

Right now we follow Anosov and Bolibrukh \cite[page $7$]{AB} and coin Hilbert's 21st problem the original RHP: Hilbert formulated it and he explicitly mentions Riemann in his wording as this problem is extremely close to Riemann's ideas\footnote{Although \cite[page $7$]{AB} also argues that ``Riemann never spoke exactly of something like it."} (namely the global construction of a function from given local analytic data, i.e. the singularity locations and associated monodromy). In short,
\begin{problem}[The original RHP (\cite{H0}, 1900)]\label{RHP:1} Proof of the existence of linear differential equations having a prescriped monodromic group.
\end{problem}
This problem seems quite different from the modern RHPs discussed in Section \ref{ex:sec} below, yet it is this problem on which the modern theory is based upon. Thus, being so influential, it is only natural to ask if RHP \ref{RHP:1} has been solved $120$ years after Hilbert's address? Well, according to the Encyclopedia of Mathematics \cite{e}, the ``solution is negative or positive depending on how the problem is understood". This answer points at a curious mathematical misunderstanding which persisted from 1908 till 1983 and which relates to the aforementioned atypicalities in Hilbert's wording. Indeed, we highlight two possible sources of confusion:\bigskip

$\diamond$ \textit{equation} vs. \textit{system}: The monodromy of a $p\textnormal{th}$ order scalar linear ODE with solution $y$ matches the monodromy of the $p\times p$ linear ODE system which the vector $(y,y',\ldots,y^{p-1})$ solves. Still, at the time of the congress, cf. \cite{P}, it was known that Hilbert's question for scalar Fuchsian equations has a negative solution unless additional singularities are included. So by a $p\textnormal{th}$ order scalar Fuchsian \textit{equation}, Hilbert likely meant a generic $p\times p$ Fuchsian \textit{system}.\smallskip
	
$\diamond$ \textit{regular} singularities vs. \textit{Fuchsian} singularities: Both notions exist for $p\textnormal{th}$ order scalar linear ODEs in the complex plane. Although up front different in appearance, these notions are in fact equivalent (this is a result by Fuchs from 1868, cf. \cite[Theorem $19.20$]{IY}). However, this is no longer the case when one considers the same singularities for $p\times p$ linear ODE systems, see Section \ref{Plejcon}.\bigskip

In turn we could think of at least three different interpretations of RHP \ref{RHP:1}: Given a monodromy group with encoded singularity locations, are we asked to realize it by\smallskip
\begin{enumerate}
	\item[(1)] a Fuchisan linear $p\textnormal{th}$ order differential equation?
	\item[(2)] a linear $p\times p$ system having only regular singularities?
	\item[(3)] a Fuchsian system on the whole Riemann sphere $\mathbb{CP}^1$?\smallskip
\end{enumerate}
The answer to (1) is \textit{negative} as shown by Poincar\'e \cite{P}: the number of parameters in a $p\textnormal{th}$ order linear Fuchsian equation with singularities $a_1,\ldots,a_n$ is less than the dimension of the space $\mathcal{M}$ of monodromy representations\footnote{Thus the need for additional, so called \textit{apparent}, singular points in the construction of Fuchsian scalar equations with prescribed monodromy. At these (from $a_1,\ldots,a_n$ different) singular points the coefficients in the equation are singular but the solutions single-valued analytic.}. Number (2) was \textit{positively} solved by Plemelj in 1908 \cite{P0} and believed to settle (3) as well. This belief persisted for 75 years until Kohn \cite{K} and Arnold, Il'yashenko \cite{AI} discovered a serious gap in Plemelj's argument. As it turned out, interpretation (3) is much more subtle and a \textit{negative} answer to Hilbert's 21st problem in said context was eventually given by Bolibrukh in 1989, cf. \cite{BO1,BO2}. We will now discuss in detail interpretations (2), (3) and highlight the pioneering work of Plemelj in this context. His contributions laid the foundation for several aspects of modern Riemann-Hilbert theory which we showcase through a series of OPSFA-type examples in Sections \ref{ex:sec}, \ref{cool} and \ref{ex:6} - however only Section \ref{ex:6} contains original content. For an in-depth discussion of interpretation (1) we refer the interested reader to \cite[Chapter $7$]{AB} or more recently \cite[$\S 3$]{GP}. 

\section{Background terminology}\label{sec:3}
We are considering $p\times p$ linear ODE systems of the form
\begin{equation}\label{e:1}
	\frac{\d\Psi}{\d z}=A(z)\Psi,\ \ \ \ \ \Psi=\Psi(z)\in\mathbb{C}^{p\times p},
\end{equation}
where $A(z)\in\mathbb{C}^{p\times p}$ is analytic in a disk $\mathbb{D}_r(z_0)\subset\mathbb{C}$ of radius $r>0$ centered at some $z_0\in\mathbb{C}$. Then, by the Picard-Lindel\"of theorem \cite[Theorem $1.1$]{IY}, any fundamental solution of \eqref{e:1} is analytic in $\mathbb{D}_r(z_0)$. Now suppose $A(z)$ is analytic on the punctured Riemann sphere $S:=\mathbb{CP}^1\setminus\{a_1,\ldots,a_n\}$ where $a_1,\ldots,a_n$ are $n$ distinct points on $\mathbb{CP}^1$. Through the use of overlapping disks and the previous local existence theorem, any fundamental solution of \eqref{e:1} can be continued along any path in $S$. This statement is part of the monodromy theorem \cite[Theorem $67$]{Bal} which furthermore asserts that the continuation only depends on the homotopy class of the path. With this in mind, we now define the \textit{monodromy} of \eqref{e:1} as follows: fix $a_0\in S$ distinct from $a_1,\ldots,a_n$ and continue any fundamental solution $\Psi(z)$ of \eqref{e:1} along some $\gamma\in\pi_1(S,a_0)$ (the fundamental group of the surface $S$ with base point $a_0$). The continuation yields $\widetilde{\Psi}(z)$, another fundamental solution of \eqref{e:1}, which links to $\Psi(z)$ via
\begin{equation*}
	\Psi(z)=\widetilde{\Psi}(z)G_{\gamma},
\end{equation*}
for some invertible $p\times p$ matrix $G_{\gamma}$, i.e. $G_{\gamma}\in\textnormal{GL}(p,\mathbb{C})$. Strictly speaking, $G_{\gamma}$ only depends on the homotopy class of the loop $\gamma$. In turn, the map $[\gamma]\mapsto G_{\gamma}$ defines a representation
\begin{equation}\label{e:2}
	\chi:\pi_1(S,a_0)\rightarrow\textnormal{GL}(p,\mathbb{C}),
\end{equation}
of $\pi_1(S,a_0)$, called the \textit{monodromy} or \textit{monodromy representation} of system \eqref{e:1}. Next, fix $n$ fundamental loops $\gamma_1,\ldots,\gamma_n$ as indicated in Figure \ref{fig1} below. Each loop $\gamma_i$ wraps once around its corresponding blue dashed cut, i.e. $\gamma_i$ wraps once around the singularity $a_i$, and all loops compose to the identity, $\gamma_1\circ\ldots\circ\gamma_n=e\in\pi_1(S,a_0)$. 
\begin{figure}[tbh]
\includegraphics[width=0.65\textwidth]{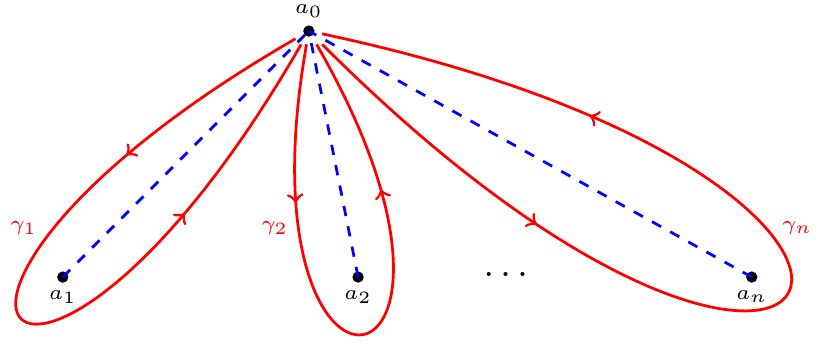}
\caption{A possible visualization of the \textcolor{red}{red} fundamental loops $\gamma_i$.}
\label{fig1}
\end{figure}

The \textit{monodromy matrix} at the singular point $a_i$ is defined as
\begin{equation*}
	G_i:=\chi([\gamma_i]),\ \ \ i=1,\ldots,n,
\end{equation*}
and we record the cyclic constraint $G_1\cdot\ldots\cdot G_n=\mathbb{I}\in\textnormal{GL}(p,\mathbb{C})$. Observe that the image of $\pi_1(S,a_0)$ under \eqref{e:2}, the \textit{monodromy group} of system \eqref{e:1}, is generated by the matrices $\{G_1,\ldots,G_n\}$.\bigskip

We are left with two obvious ambiguities in our definitions: first, if we were to start from a different fundamental solution $\Psi^{\ast}(z)$ instead of $\Psi(z)$, then $\Psi^{\ast}(z)=\Psi(z)H$ for some $H\in\textnormal{GL}(p,\mathbb{C})$ and the associated monodromy matrices change accordingly,
\begin{equation}\label{e:3}
	G_i^{\ast}=H^{-1}G_iH,\ \ i=1,\ldots,n.
\end{equation}
Second, the dependence of $\pi_1(S,a_0)$ and thus $G_i=\chi([\gamma_i])$ on the base point $a_0$: given that $S$ is path connected we know that $\pi_1(S,a_0)$ and $\pi_1(S,a_0^{\ast})$
are isomorphic for any two base points $a_0,a_0^{\ast}$ different from $a_1,\ldots,a_n$, thus under the representation \eqref{e:2} the associated monodromy matrices change also in the style of \eqref{e:3}. Hence, summarizing the two ambiguities, the monodromy of \eqref{e:1} is defined up to conjugation equivalence and so an element of the space
\begin{equation*}
	\mathcal{M}:=\textnormal{Hom}\big(\pi_1(S,a_0),\textnormal{GL}(p,\mathbb{C})\big)\big/\textnormal{GL}(p,\mathbb{C}),
\end{equation*}
of conjugacy classes of representations of $\pi_1(S,a_0)$. Alternatively, in the language of monodromy matrices, we have
\begin{equation}\label{e:4}
	\mathcal{M}\cong\big\{(G_1,\ldots,G_n):\,G_1\cdot\ldots\cdot G_n=\mathbb{I}\big\}\big/\textnormal{GL}(p,\mathbb{C}).
\end{equation}
One notion remains: we say system \eqref{e:1} is \textit{Fuchsian} if all its singular points $\{a_1,\ldots,a_n\}$ are first order poles of $A(z)$, i.e. without loss of generality (thanks to a fractional linear map) we have
\begin{equation*}
	A(z)=\sum_{i=1}^n\frac{B_i}{z-a_i},\ \ \ \ \ \ \ \ \ \ \ \ \res_{z=\infty}A(z)=-\sum_{i=1}^nB_i=0.
\end{equation*}
Using the above terminology we are now able to rephrase RHP \ref{RHP:1} in interpretation (3) as follows:
\begin{problem}[The original RHP, case (3)]\label{RHP:2} Is the monodromy map $\mu:\mathcal{M}^{\ast}\rightarrow\mathcal{M}$ from the space
\begin{equation*}
	\mathcal{M}^{\ast}:=\left\{(B_1,\ldots,B_n):\ \ B_i\in\mathbb{C}^{p\times p},\ \ \sum_{i=1}^nB_i=0\right\}\bigg/\textnormal{GL}(p,\mathbb{C})
\end{equation*}
of Fuchsian systems with fixed singularities $a_1,\ldots,a_n$ into $\mathcal{M}$ surjective?
\end{problem}
This problem concludes our content on the necessary background terminology of Hilbert's $21$st problem. We will now review some of the mathematical works devoted to its solution and which were published between $1908$ and $1989$. These works introduced several key ideas and techniques that are nowadays used in the Riemann-Hilbert analysis of OPSFA-type problems.
\section{Plemelj's contributions}\label{Plejcon}
In 1908, Plemelj \cite{P0} (upgraded to book form 55 years later in \cite{Pl}) published a solution of RHP \ref{RHP:2} that was widely accepted until the early 1980s. In his work, he reduced RHP \ref{RHP:2} to a Hilbert boundary value problem\footnote{In standard textbooks on singular integral equations, cf. \cite[$\S 39$]{M}, a \textit{Riemann-Hilbert problem}, named after the original works \cite{R} and \cite{H1,H2}, generally refers to the problem of constructing a function which is analytic in a domain $\Omega\subset\mathbb{C}$, continuous on the closure $\overline{\Omega}$ and with prescribed boundary values on $\partial\Omega$. We will not follow this tradition but instead use the term \textit{Hilbert boundary value problem} like in \cite[$\S 4$]{V} as to not confuse ourselves with the original RHPs \ref{RHP:1} and \ref{RHP:2}.} in the theory of singular integral equations and the idea goes as follows: Join all singularities $a_1,\ldots,a_n\in\mathbb{C}$ by a simple closed oriented contour $\Gamma$ as indicated in Figure \ref{fig2} below.
\begin{figure}[tbh]
\includegraphics[width=0.475\textwidth]{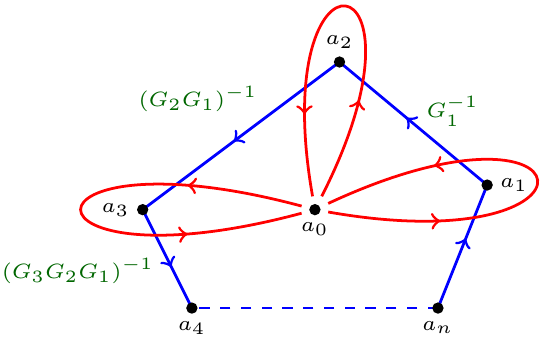}
\caption{The oriented contour $\Gamma$ in \textcolor{blue}{blue} together with the aformentioned \textcolor{red}{red} fundamental loops $\gamma_i$. Some values of $G(z)$ are indicated in \textcolor{ao}{green}.}
\label{fig2}
\end{figure}

Now define the piecewise constant, matrix-valued and invertible function
\begin{equation*}
	G(z):=(G_iG_{i-1}\cdot\ldots\cdot G_1)^{-1},\ \ \ z\in[a_i,a_{i+1}),\ \ \ \ i=1,\ldots,n;\ \ \ \ a_{n+1}:=a_1,
\end{equation*}
which, thanks to the cyclic constraint \eqref{e:4}, satisfies $G(z)=\mathbb{I}$ for $z\in[a_n,a_1)$. With $\Omega^+$ denoting the region in $\mathbb{C}$ bounded by $\Gamma$ and $\Omega^-$ the complement of the closure $\overline{\Omega^+}$ in $\mathbb{CP}^1$, we then consider the following boundary value problem:
\begin{prob}[{\cite[page $213,214,229$]{P0}}]\label{HP1} Find all vector-valued functions $X=X(z)\in\mathbb{C}^{1\times p}$ such that
\begin{enumerate}
	\item[(1)] $X(z)$ is analytic in $\Omega^{\pm}\setminus\{\infty\}$ and extends continuously from either side to the punctured contour $\Gamma\setminus\{a_1,\ldots,a_n\}$.
	\item[(2)] On the open segments $(a_i,a_{i+1}),i=1,\ldots,n$ the pointwise limits
	\begin{equation*}
		X_{\pm}(z):=\lim_{\substack{w\rightarrow z\\ w\in\Omega^{\pm}}}X(w)
	\end{equation*}
	satisfy the boundary condition $X_+(z)=X_-(z)G(z)$, compare Figure \ref{fig3} below.
	\begin{figure}[tbh]
\includegraphics[width=0.4\textwidth]{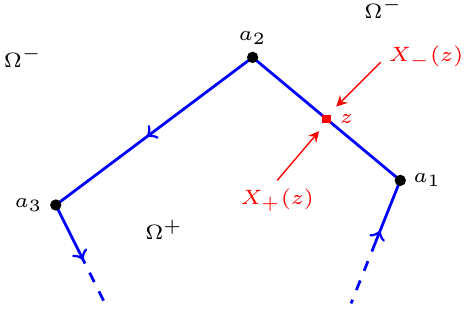}
\caption{The pointwise limits $X_{\pm}(z)$ at some $z\in(a_i,a_{i+1}),i=1,\ldots,n$ in \textcolor{red}{red}.}
\label{fig3}
\end{figure}
	\item[(3)] $X(z)$ is of finite degree at $z=\infty$, that is
	\begin{equation}\label{inf}
		X(z)=\gamma(z)+\mathcal{O}\big(z^{-1}\big),\ \ \ \ \ |z|\rightarrow\infty,
	\end{equation}
	with some vector-valued polynomial $\gamma\in\mathbb{C}^{1\times p}[z]$. Moreover, in a small neighborhood of $a_i$,
	\begin{equation*}
		\big\|X(z)\big\|\leq\frac{C}{|z-a_i|^{\alpha}},\ \ \ \ \ 0\leq\alpha<1,\ \ \ C>0.
	\end{equation*}
\end{enumerate}
\end{prob}
With Problem \ref{HP1} as starting point, Plemelj first proved its solvability through an application of Fredholm's theory of integral equations, cf. \cite{Fre,P1}, at the time a novel analytic tool. We will lay out some of his steps below subject to the following two temporary assumptions:
\begin{assu}\label{assu:0} We seek solutions of Problem \ref{HP1} which extend H\"older continuously from either side up to the punctured contour $\Gamma\setminus\{a_1,\ldots,a_n\}$.
\end{assu}
\begin{assu}\label{assu:1} The invertible jump matrix $G(z),z\in\Gamma$ in condition $(2)$ of Problem \ref{HP1} is H\"older continuous on all of $\Gamma$ and not just on each open segment $(a_i,a_{i+1})$ separately. In turn, the blow up constraint near $z=a_i$ in Problem \ref{HP1} $(3)$ becomes unnecessary.
\end{assu}
Contingent on these assumptions, Problem \ref{HP1} is equivalent to the problem of finding all H\"older continuous functions $X_-(z)\in\mathbb{C}^{1\times p}$ defined for $z\in\Gamma$ such that
%
%
\begin{equation}\label{e:5}
	\forall\,z\in\Omega^-:\ 0=\frac{1}{2\pi\im}
	\int_{\Gamma}\frac{X_-(\lambda)G(\lambda)}{\lambda-z}\,\d\lambda;\ \ \ \ \ \textnormal{and} \ \ \ \ \ 
	\forall\,z\in\Omega^+:\ 0=-\frac{1}{2\pi\im}
	\int_{\Gamma}\frac{X_-(\lambda)}{\lambda-z}\,\d\lambda+\gamma(z).
\end{equation}
Indeed, by Cauchy's theorem\footnote{In the slightly more general form $\oint_{\,\partial\Omega}f(z)\,\d z=0$ where $f:\overline{\Omega}\rightarrow\mathbb{C}$ is analytic in the region $\Omega\subset\mathbb{C}$, continuous on the closure $\overline{\Omega}$ and the boundary $\partial\Omega$ of $\Omega$ is a rectifiable Jordan curve, cf. \cite{W}.} and Plemelj's own $1908$ formul\ae\,\cite{P2} (the Plemelj-Sokhotski formul\ae\, which Plemelj rediscovered) the first constraint in \eqref{e:5} holds precisely when $X_-(z)G(z),z\in\Gamma$ is the boundary value of a function $X(z)\in\mathbb{C}^{1\times p}$ which is analytic in $\Omega^+$ and H\"older continuous on $\Omega^+\cup\Gamma$. Likewise the second constraint in \eqref{e:5} holds if and only if $X_-(z),z\in \Gamma$ is the boundary value of some $X(z)\in\mathbb{C}^{1\times p}$, analytic in $\Omega^-\setminus\{\infty\}$, H\"older continuous on $(\Omega^-\setminus\{\infty\})\cup\Gamma$ and of finite degree at $z=\infty$, in the sense of \eqref{inf}. Next, Plemelj noticed that both constraints in \eqref{e:5} are equivalent (again by Cauchy's theorem and the Plemelj-Sokhotski formul\ae) to the principal value integral equations,
\begin{equation}\label{e:6}
	X_-(z)=-\frac{1}{\pi\im}\,\textnormal{pv}
	\int_{\Gamma}\frac{X_-(\lambda)}{\lambda-z}\,\d\lambda+2\gamma(z),\ \ \  X_-(z)G(z)=\frac{1}{\pi\im}\,\textnormal{pv}
	\int_{\Gamma}\frac{X_-(\lambda)G(\lambda)}{\lambda-z}\,\d\lambda,\ \ \ \ z\in\Gamma,
\end{equation}
and from these one obtains in turn the quasi-regular integral equation 
\begin{equation}\label{e:7}
	X_-(z)-\frac{1}{\pi\im}\,\textnormal{pv}
	\int_{\Gamma}X_-(\lambda)\frac{K(z,\lambda)}{\lambda-z}\,\d\lambda
	=\gamma(z),\ \ \ z\in\Gamma,
\end{equation}
with the H\"older continuous matrix-valued kernel function
\begin{equation*}
	K(z,\lambda):=\frac{1}{2}\big(G(\lambda)-G(z)\big)(G(z))^{-1},\ \ \ (z,\lambda)\in\Gamma\times\Gamma.
\end{equation*}
\begin{remark}\label{tech} In more abstract and general terms, the left hand side in \eqref{e:7} defines a singular Fredholm operator
\begin{equation*}
	(K\phi)(z):=\phi(z)-\frac{1}{\pi\im}\,\textnormal{pv}
	\int_{\Gamma}\phi(\lambda)\frac{K(z,\lambda)}{\lambda-z}\,\d\lambda,\ \ \ z\in\Gamma,
\end{equation*}
of index zero which acts on H\"older continuous functions defined on $\Gamma$, cf. \cite[$\S 45$]{M}, and we assume that $K(z,\lambda)$ satisfies a H\"older condition on $\Gamma$ in both variables. Plemelj did not work with such singular Fredholm equations in \cite{P0} since he managed to transform the piecewise constant jumps in Problem \ref{HP1} to differentiable ones. In turn, his \cite[$(4)$]{P0} is an ordinary Fredholm equation of the second kind which can be studied by Fredholm's $1903$ theory. We will discuss a more generally applicable reduction of piecewise H\"older continuous jumps to H\"older continuous ones and en route lift Assumption \ref{assu:1} on $G(z)$. In fact, the general theory of singular Fredholm integral equations of the type \eqref{e:7} was systematically developed some $30$ years after Plemelj's paper with prominent contributions by Giraud \cite{Gi}, Muskhelishvili-Vekua \cite{MV} and Gakhov \cite{Ga2}.
\end{remark}
We now make \eqref{e:7} the starting point for the solvability analysis of Problem \ref{HP1}. Precisely, given that any continuous solution $X_-(z),z\in\Gamma$ of \eqref{e:7} will be automatically H\"older continuous as a consequence of Assumption \ref{assu:1}, we answer the following two questions:\smallskip
\begin{enumerate}
	\item[Q1.] Under which conditions is \eqref{e:7} solvable in the space of continuous functions on $\Gamma$?\smallskip
	\item[Q2.] Does each continuous solution of \eqref{e:7} yield a solution of Problem \ref{HP1}?\smallskip
\end{enumerate}
Regarding Q2, we recall that every continuous solution $X_-(z),z\in\Gamma$ of \eqref{e:7} will produce a solution of Problem \ref{HP1} with Assumption \ref{assu:1} in place if and only if $X_-(z),z\in\Gamma$ solves both equations in \eqref{e:6}. Now rephrase \eqref{e:6} in terms of the function
\begin{equation}\label{e:8}
	\widehat{X}(z):=\begin{cases}\displaystyle\frac{1}{2\pi\im}
	\int_{\Gamma}\frac{X_-(\lambda)}{\lambda-z}\,\d\lambda-\gamma(z),&z\in\Omega^+\bigskip\\
	\displaystyle\frac{1}{2\pi\im}
	\int_{\Gamma}\frac{X_-(\lambda)G(\lambda)}{\lambda-z}\,\d\lambda,&z\in\Omega^-\end{cases},
\end{equation}
which is analytic on $\Omega^{\pm}$, H\"older continuous on $\Omega^{\pm}\cup\Gamma$ and vanishes at $z=\infty$. In fact,
\begin{equation}\label{e:9}
	\eqref{e:6}\ \textnormal{holds}\ \ \ \Leftrightarrow\ \ \ \ \widehat{X}_+(z)=0=\widehat{X}_-(z),\ \ z\in\Gamma\ \ \ \Leftrightarrow\ \ \ \ \widehat{X}(z)\equiv 0\ \ \textnormal{in}\ \mathbb{C}
\end{equation}
Furthermore, our central integral equation \eqref{e:7} is equivalent to the jump condition 
\begin{equation*}
	\widehat{X}_+(z)=\widehat{X}_-(z)(G(z))^{-1},\ \ z\in\Gamma,
\end{equation*}
which motivates the introduction of the below, after Plemelj the \textit{accompanying}, Hilbert boundary value problem:
\begin{prob}[{\cite[page $215$]{P0}}]\label{HP2} Find all vector-valued functions $\widehat{X}=\widehat{X}(z)\in\mathbb{C}^p$ such that
\begin{enumerate}
	\item[(1)] $\widehat{X}(z)$ is analytic in $\Omega^{\pm}$ and extends H\"older continuously from either side up to $\Gamma$.
	\item[(2)] The pointwise limits
	\begin{equation*}
		\widehat{X}_{\pm}(z):=\lim_{\substack{w\rightarrow z\\ w\in\Omega^{\pm}}}\widehat{X}(w),\ \ \ z\in\Gamma,
	\end{equation*}
	satisfy the boundary condition $\widehat{X}_+(z)=\widehat{X}_-(z)(G(z))^{-1},z\in\Gamma$.
	\item[(3)] $\widehat{X}(z)$ vanishes at $z=\infty$, that is
	\begin{equation*}
		\widehat{X}(z)=\mathcal{O}\big(z^{-1}\big),\ \ \ \ \ |z|\rightarrow\infty.
	\end{equation*}
\end{enumerate}
\end{prob}
The special homogeneous problem \ref{HP2} is useful since it allows us to summarize our previous chain \eqref{e:9} as the following exclusive alternative: let $X_-(\lambda),\lambda\in\Gamma$ be a continuous solution of \eqref{e:7},\bigskip 
	
	$\diamond$ If $\widehat{X}(z)$ defined in terms of said $X_-(\lambda),\lambda\in\Gamma$ in \eqref{e:8} is identically zero, then $X_-(\lambda)$ solves \eqref{e:6} and thus produces a solution of the initial Problem \ref{HP1}.\smallskip
	
	$\diamond$ If $\widehat{X}(z)$ defined in terms of said $X_-(\lambda),\lambda\in\Gamma$ in \eqref{e:8} is not identically zero, then $\widehat{X}(z)$ is a non-trivial solution of the accompanying Problem \ref{HP2}.\bigskip

Most importantly we can now formulate Plemelj's first criterion as response to Q2 above:
\begin{lem}[{\cite[page $215$]{P0}}]\label{Plem1} If the accompanying Problem \ref{HP2} has no non-trivial solutions, then every continuous solution $X_-(z),z\in\Gamma$ of \eqref{e:7} yields a solution of Problem \ref{HP1}.
\end{lem}
Evidently, Lemma \ref{Plem1} does not guarantee solvability of \eqref{e:7}. Indeed, in order to answer Q1, Plemelj used Fredholm's theory of integral equations\footnote{Fredholm's theorems for singular integral equations of the form \eqref{e:7} were proven in \cite{N} and are called Noether's theorems. If however the index of the underlying Fredholm operator is zero, compare Remark \ref{tech}, then those theorems are exactly the same as for standard Fredholm integral equations of the second kind, cf. \cite{Fre,P1}.}. First, consider two additional homogeneous Hilbert boundary value problems: the \textit{associated} Hilbert boundary value problem which consists in finding $Y=Y(z)\in\mathbb{C}^{1\times p}$ that is analytic in $\Omega^{\pm}$, extends H\"older continuously up to $\Gamma$, vanishes at infinity and instead of the jump condition in Problem \ref{HP2} satisfies
\begin{equation}\label{e:10}
	Y_+(z)=Y_-(z)\big(G^{\top}(z)\big)^{-1},\ \ z\in\Gamma,
\end{equation}
with the matrix transpose $G^{\top}(z)$ of $G(z)$. Moreover, the problem which accompanies the associated problem \eqref{e:10} and which seeks $\widehat{Y}=\widehat{Y}(z)\in\mathbb{C}^{1\times p}$ as in Problem \ref{HP2} but with jump constraint
\begin{equation}\label{e:11}
	\widehat{Y}_+(z)=\widehat{Y}_-(z)G^{\top}(z),\ \ z\in\Gamma.
\end{equation}
Repeating the logic that took us from Problem \ref{HP1} (with Assumptions \ref{assu:0} and \ref{assu:1}) to equation \eqref{e:7} we easily see that \eqref{e:11} leads to the integral equation
\begin{equation}\label{e:12}
	\widehat{Y}_+(z)+\frac{1}{\pi\im}\,\textnormal{pv}
	\int_{\Gamma}\widehat{Y}_+(\lambda)\frac{K^{\top}(\lambda,z)}{\lambda-z}\,\d\lambda=0,\ \ z\in\Gamma,
\end{equation} 
and which is the adjoint equation of \eqref{e:7}. We can now derive Plemelj's second criterion:
\begin{lem}[{\cite[page $217,218$]{P0}}]\label{Plem2} If the associated problem \eqref{e:10} has no non-trivial solutions, then \eqref{e:7} is solvable in the space of continuous functions on $\Gamma$ for any given $\gamma\in\mathbb{C}^{1\times p}[z]$.
\end{lem}
\begin{proof} Problem \eqref{e:11} accompanies the associated problem \eqref{e:10}, or equivalently, \eqref{e:10} is the accompanying problem of \eqref{e:11} (as we invert the jump matrix and homogenize at infinity in the accompanying problems). But by assumption, \eqref{e:10} has no non-trivial solutions, so by Lemma \ref{Plem1} every continuous solution $\widehat{Y}_+(z),z\in\Gamma$ of the adjoint equation \eqref{e:12} yields a solution of the Hilbert boundary value problem \eqref{e:11}. Let $\widehat{Y}_+(z),z\in\Gamma$ denote such a non-trivial continuous solution of \eqref{e:12}\footnote{If none exists, i.e. the adjoint equation \eqref{e:12} has only the trivial solution, then by \cite[$(3''),(5')$]{N}, equation \eqref{e:7} is continuously solvable for any right-hand side $\gamma(z)$.  This is the Fredholm Alternative.}, so $\widehat{Y}_+(z),z\in\Gamma$ is the boundary value of some analytic function $\widehat{Y}(z),z\in\Omega^+$. However, by \cite[$(7)$]{N}, the inhomogeneous equation \eqref{e:7} is continuously solvable if and only if 
\begin{equation*}
	\int_{\Gamma}\gamma(z)\widehat{Y}_+^{\top}(z)\,\d z=0,
\end{equation*}
and which is guaranteed by our last conclusion and the fact that $\gamma\in\mathbb{C}^{1\times p}[z]$ (any entire $\gamma(z)$ would do). This concludes the proof of the Lemma.
\end{proof}
To summarize, Lemma \ref{Plem1} and \ref{Plem2} guarantee solvability of our initial Problem \ref{HP1} (with Assumption \ref{assu:1} in place, but no longer Assumption \ref{assu:0}), provided the accompanying Problem \ref{HP2} and the associated problem \eqref{e:10} are only trivially solvable. Although these two conditions seem peculiar they are in fact easily verifiable from our assumption that $\gamma(z)$ is polynomial and thus our solutions of Problem \ref{HP1} meromorphic at $z=\infty$. Indeed, we first recall, cf. \cite[$\S 2$]{N}, that the homogeneous version of \eqref{e:7} (or any of the other homogeneous equations in this section) has finitely many, say $s\in\mathbb{Z}_{\geq 1}$, linearly independent continuous solutions\footnote{In the abstract setting $K\phi=0$ of Remark \ref{tech}, this statement is part of the Riesz-Schauder theorem, asserting that all eigenspaces of the associated Fredholm integral operator are finite-dimensional.}. Then 
\begin{prop}[{\cite[page $219$]{P0}}] Any solution $X(z)$ of Problem \ref{HP1} (with Assumption \ref{assu:1}) has a zero of order at most $s\in\mathbb{Z}_{\geq 1}$ at $z=\infty$.
\end{prop}
\begin{proof} If $X(z)$ has a zero of order $k\in\mathbb{Z}_{\geq 1}$ at $z=\infty$, then 
\begin{equation*}
	X(z),\,zX(z),\,z^2X(z),\,\ldots,\,z^{k-1}X(z)
\end{equation*}
all solve the homogeneous version of Problem \ref{HP2} (with Assumption \ref{assu:1}) and thus $X_-(z),zX_-(z)$, $z^2X_-(z),\ldots$, $z^{k-1}X_-(z),z\in\Gamma$ all solve \eqref{e:7} with $\gamma\equiv 0$. Since these solutions are linearly independent we must have $k\leq s$ by our previous discussion.
\end{proof}
So we can always find an integer $r\in\mathbb{Z}_{\geq 0}$ such that neither the accompanying problem \ref{HP1} nor the associated problem \eqref{e:10} admit solutions with a zero of order greater than $r$ at $z=\infty$. Hence, if $X(z)$ solves Problem \ref{HP1} where $\gamma\in\mathbb{C}^{1\times p}[z]$ satisfies $\textnormal{deg}(\gamma)\leq r$, then
\begin{equation}\label{e:13}
	X^{\circ}(z):=\begin{cases}\displaystyle X(z),&z\in\Omega^+\smallskip\\
	\displaystyle\frac{X(z)}{(z-z_0)^r},&z\in\Omega^-\end{cases}\ \ \ \ \ \ \ \ \textnormal{for some arbitrary}\ \ z_0\in\Omega^+,
\end{equation}
solves Problem \ref{HP1} with jump matrix $G^{\circ}(z):=(z-z_0)^rG(z)$ and is bounded at $z=\infty$. Most importantly, the central assumptions in Lemma \ref{Plem1} and \ref{Plem2} for the accompanying and associated problem of the $X^{\circ}$-Hilbert boundary value (which is Problem \ref{HP1} with $X^{\circ}(z)=\mathcal{O}(1)$ at $z=\infty$ and H\"older continuous jumps on all of $\Gamma$) are satisfied: 
\begin{enumerate}
	\item[(A)] The accompanying problem has jump $\widehat{X}_+^{\circ}(z)=\widehat{X}^{\circ}_-(z)(z-z_0)^{-r}(G(z))^{-1}$, $z\in\Gamma$, so if this problem has a non-trivial solution (that vanishes at infinity) then
	\begin{equation*}
		\widehat{X}(z):=\begin{cases}\widehat{X}^{\circ}(z),&z\in\Omega^+\smallskip\\
		(z-z_0)^{-r}\widehat{X}^{\circ}(z),&z\in\Omega^-\end{cases}
	\end{equation*}
	solves Problem \ref{HP2} and has a zero of order greater than $r$ at $z=\infty$, contradicting our initial choice for $r$.
	\item[(B)] The associated problem has jump $Y^{\circ}_+(z)=Y^{\circ}_-(z)(z-z_0)^{-r}(G^{\top}(z))^{-1}$, $z\in\Gamma$, so if this problem has a non-trivial solution (that vanishes at infinity) then
	\begin{equation*}
		Y(z):=\begin{cases} Y^{\circ}(z),&z\in\Omega^+\\
		(z-z_0)^{-r}\,Y^{\circ}(z),&z\in\Omega^-\end{cases}
	\end{equation*}
	solves the associated problem \eqref{e:10} and vanishes to order greater than $r$ at $z=\infty$. Again a contradiction to our choice of $r$.\bigskip
\end{enumerate}
The $X^{\circ}$-Hilbert boundary value problem is thus solvable by Lemma \ref{Plem1} and \ref{Plem2} and by reversing \eqref{e:13} we arrive at
\begin{theo}[{\cite[page $222,223$]{P0}}]\label{theo:1} There exists $r\in\mathbb{Z}_{\geq 1}$ such that Problem \ref{HP1} (with Assumption \ref{assu:1}) is solvable for all $\gamma\in\mathbb{C}^{1\times p}[z]$ with $\textnormal{deg}(\gamma)\leq r$. Moreover, every solution of said problem is of the form
\begin{equation}\label{e:14}
	X(z)=\sum_{i=1}^p\gamma_iX_i(z)+\sum_{i=p+1}^m\gamma_iX_i(z),\ \ z\in\mathbb{C}\setminus\Gamma,
\end{equation}
with some constants $\gamma_1,\ldots,\gamma_m\in\mathbb{C}$ and where $\{X_i(z)\}_{j=1}^m$ are linearly independent particular solutions of Problem \ref{HP1} with
\begin{equation*}
	\lim_{|z|\rightarrow\infty}z^{-r}X^{ij}(z)=\delta_{ij},\ \ \ \ i,j=1,\ldots,p;\ \ \ \ \ \ \ \ \ \ \ \ X_i(z)=\big[X^{ij}(z)\big]_{j=1}^p\in\mathbb{C}^{1\times p},
\end{equation*}
and $\{X_i(z)\}_{i=p+1}^m$ (if there are any) have degree strictly less than $r$ at $z=\infty$.
\end{theo}
\begin{proof} Clearly, if $\{X_i\}_{j=1}^m$ are any particular solutions of the jump condition in Problem \ref{HP1}, then any polynomial linear combination of them will also solve the same jump constraint. Thus, choosing $r\in\mathbb{Z}_{\geq 1}$ sufficiently large as indicated in our previous discussion we can always ensure that Problem \ref{HP1} is solvable for the collection of constant polynomials $\gamma_i(z)=[\delta_{ij}]_{j=1}^p$ with $i=1,\ldots,p$. Reversing subsequently \eqref{e:13} and taking an arbitrary combination of the so obtained solutions yields the first half in \eqref{e:14} with the indicated asymptotic behavior at $z=\infty$. The second sum in \eqref{e:14} stems from the complete system of linearly independent solutions of the homogeneous version of Problem \ref{HP1}, modulo the inversion of \eqref{e:13}.
\end{proof}
Although Theorem \ref{theo:1} solves Problem \ref{HP1}, the implicit dependence of \eqref{e:14} on $r\in\mathbb{Z}_{\geq 1}$ is a drawback, in particular since $r$ affects $X^{\circ}(z)$ through \eqref{e:13} and it affects the parameter $m$. In order to bypass these deficiencies Plemelj  managed then to construct a \textit{fundamental system} $\{F_i\}_{i=1}^p\subset\mathbb{C}^{1\times p}$ of solutions to Problem \ref{HP1} such that every solution of the same problem is a polynomial linear combination of the form
\begin{equation*}
	X(z)=\sum_{i=1}^pq_i(z)F_i(z),\ \ z\in\mathbb{C}\setminus\Gamma;\ \ \ \ \ \ q_i\in\mathbb{C}[z].
\end{equation*}
In order to present his argument, we first note that the system $\{X_i\}_{i=1}^p$ in \eqref{e:14} is linearly independent over $\mathbb{C}[z]$ because of its asymptotic normalization. With $r\in\mathbb{Z}_{\geq 1}$ as stated in Theorem \ref{theo:1}, we then observe that among all solutions of the form \eqref{e:14} some have the lowest degree at $z=\infty$ since a solution of Problem \ref{HP1} cannot vanish to all orders at $z=\infty$. Let $F_1(z)$ denote one of the solutions of lowest degree $-\nu_1\in\mathbb{Z}$ at $z=\infty$ and $F_2(z)$ one of the solutions of Problem \ref{HP1} of lowest degree $-\nu_2\geq-\nu_1$ at $z=\infty$ such that $F_2\notin\textnormal{span}\{F_1\}$\footnote{All linear combinations are taken over $\mathbb{C}[z]$}. Next, we let $F_3(z)$ denote one of the solutions of Problem \ref{HP1} of lowest degree $-\nu_3\geq-\nu_2$ at $z=\infty$ such that $F_3\notin\textnormal{span}\{F_1,F_2\}$ and note that this process may be continued until we reach $F_p\notin\textnormal{span}\{F_1,\ldots,F_{p-1}\}$ of lowest degree $-\nu_p\geq-\nu_{p-1}$ at $z=\infty$. Indeed, if $F_1(z),\ldots,F_k(z)$ with $k<p$ have been constructed in this fashion then we can find yet another solution of Problem \ref{HP1} not contained in $\textnormal{span}\{F_1,\ldots,F_k\}$. For otherwise, the particular solutions $\{X_i\}_{i=1}^p$ in \eqref{e:14} would be contained in $\textnormal{span}\{F_1,\ldots,F_k\}$ and thus be linearly dependent over $\mathbb{C}[z]$, contradicting our previous discussion. In summary, the above procedure leads to $p$ fundamental solutions
\begin{equation}\label{e:15}
	F_1(z),F_2(z),\ldots,F_p(z),
\end{equation}
of Problem \ref{HP1} with lowest degrees $-\nu_1,-\nu_2,\ldots,-\nu_p$ at $z=\infty$ where $\nu_1\geq\nu_2\geq\ldots\geq\nu_p$ and $F_k\notin\textnormal{span}\{F_1,\ldots,F_{k-1}\}$ for $k=2,\ldots,p$. The central properties of the so constructed fundamental system \eqref{e:15} are summarized below.
\begin{theo}[{\cite[page $227$]{P0}}]\label{theo:2} There exists $r\in\mathbb{Z}_{\geq 1}$ such that Problem \ref{HP1} (with Assumption \ref{assu:1}) is solvable for all $\gamma\in\mathbb{C}^{1\times p}[z]$ with $\textnormal{deg}(\gamma)\leq r$. Moreover, every solution of said problem is of the form
\begin{equation}\label{e:16}
	X(z)=\sum_{i=1}^pq_i(z)F_i(z),\ \ \ z\in\mathbb{C}\setminus\Gamma,
\end{equation}
where $\{F_i\}_{j=1}^p$ is a fundamental system of solutions, $q_i\in\mathbb{C}[z]$ and the corresponding fundamental matrix
\begin{equation*}
	\Phi(z):=\big[F^{ij}(z)\big]_{i,j=1}^p,\ \ \ \ \ F_i(z)=\big[F^{ij}(z)\big]_{j=1}^p\in\mathbb{C}^{1\times p},\ \ \ z\in\mathbb{C}\setminus\Gamma
\end{equation*}
satisfies
\begin{itemize}
	\item[(i)] $\det\Phi(z)\neq 0$ for all $z\in\mathbb{C}$ including $z\in\Gamma$ with the appropriate limiting values $\Phi_{\pm}(z)$.\smallskip
	\item[(ii)] There exists a diagonal matrix $\Lambda\in\mathbb{Z}^{p\times p}$ such that $z^{\Lambda}\Phi(z)$ is invertible at $z=\infty$.
\end{itemize}
\end{theo}
\begin{proof} The first statement of the theorem was already established in Theorem \ref{theo:1} and \eqref{e:16} follows from (i) since by construction of the fundamental matrix, $G(z)=(\Phi_-(z))^{-1}\Phi_+(z),z\in\Gamma$ and thus for any solution $X(z)$ of Problem \ref{HP1} (with Assumption \ref{assu:1}),
\begin{equation*}
	X_+(z)(\Phi_+(z))^{-1}=X_-(z)(\Phi_-(z))^{-1},\ \ z\in\Gamma,
\end{equation*}
i.e. $X(z)(\Phi(z))^{-1}$ is an entire function. Hence, by condition (3) in Problem \ref{HP1} we find that $X(z)(\Phi(z))^{-1}\in\mathbb{C}^{1\times p}[z]$, so \eqref{e:16} follows and conversely it is clear that the right hand side in \eqref{e:16} solves the Hilbert boundary value problem \ref{HP1} for any $q_i\in\mathbb{C}[z]$. In order to establish (i) we first note that if $X(z)$ is any solution of Problem \ref{HP1} of degree strictly less than $-\nu_k\in\mathbb{Z}$ at $z=\infty$ where $k\in\{2,\ldots,p\}$, then necessarily
\begin{equation}\label{e:17}
	X(z)=\sum_{i=1}^{k-1}r_i(z)F_i(z)\ \ \ \textnormal{for some}\ r_i\in\mathbb{C}[z].
\end{equation}
For if such $X$, i.e. of degree strictly less than $-\nu_k$ at $z=\infty$, were not contained in $\textnormal{span}\{F_1,\ldots,F_{k-1}\}$, then our previous selection process of the subsequent $F_k\notin\textnormal{span}\{F_1,\ldots,F_{k-1}\}$ with lowest degree $-\nu_k\geq-\nu_{k-1}$ would fail. Returning to (i) we now show that the expression
\begin{equation}\label{e:18}
	F(z):=\sum_{i=1}^pa_iF_i(z),\ \ \ \ z\in\mathbb{C}\setminus\Gamma,
\end{equation}
with $a_i\in\mathbb{C}$ not all zero, does not vanish at any point $z\in\mathbb{C}$: if $F(z_0)=0$ for some $z_0\notin\Gamma$, then
\begin{equation*}
	(z-z_0)X(z)=\sum_{i=1}^pa_iF_i(z),\ \ \ \ z\in\mathbb{C}\setminus\Gamma
\end{equation*}
where $X(z)$ solves Problem \ref{HP1}. If $a_k$ denotes the last of the coefficients $a_1,\ldots,a_p$ which is non-zero, then the degree of $X(z)$ at $z=\infty$ is strictly less than $-\nu_k$, so \eqref{e:17} must hold for the same $X(z)$, i.e. we would have
\begin{equation*}
	(z-z_0)\sum_{i=1}^{k-1}r_i(z)F_i(z)=\sum_{i=1}^ka_iF_i(z).
\end{equation*}
But this is a contradiction to the requirement $F_k\notin\textnormal{span}\{F_1,\ldots, F_{k-1}\}$ and thus $F(z)\neq 0$ for all $z\in\mathbb{C}\setminus\Gamma$. If on the other hand $F(z_0)=0$ for some $z_0\in\Gamma$, i.e. either $F_+(z_0)=0$ or equivalently $F_-(z_0)=F_+(z_0)(G(z_0))^{-1}=0$, then define
\begin{equation*}
	X(z):=\frac{F(z)}{z-z_0},\ \ \ z\in\mathbb{C}\setminus\Gamma.
\end{equation*}
This $\mathbb{C}^{1\times p}$-valued function is analytic in $\Omega^{\pm}\setminus\{\infty\}$, of finite degree at $z=\infty$ and extends H\"older continuously from either side to the punctured contour $\Gamma\setminus\{z_0\}$. Thus the above $X(z)$ solves the singular integral equation \eqref{e:7} for all $z\in\Gamma\setminus\{z_0\}$ and since $F_{\pm}(z),z\in\Gamma$ are H\"older continuous with $F_{\pm}(z_0)=0$ we conclude that $X_-(z),z\in\Gamma$ as well as $X_+(z)=X_-(z)G(z),z\in\Gamma$ are H\"older continuous on all of $\Gamma$. At this point we simply repeat the previous logic for $z_0\notin\Gamma$ and conclude $F(z)\neq 0$ for all $z\in\Gamma$. All together, \eqref{e:18} does not vanish for any $z\in\mathbb{C}$ (choosing $F_{\pm}(z)$ for $z\in\Gamma$) and thus $\det\Phi(z)\neq 0$ for all $z\in\mathbb{C}$, as claimed. We are left with (ii) and thus the  point $z=\infty$ at which $\det\Phi(z)$ might have a pole or may remain finite. However, if $-\nu_i\in\mathbb{Z}$ denotes the degree of $F_i(z)$ at $z=\infty$, then
\begin{equation*}
	\det\big[z^{\nu_i}F^{ij}(z)\big]_{i,j=1}^p
\end{equation*}
is certainly finite at $z=\infty$. If it were zero, then we could find $a_1,\ldots,a_p\in\mathbb{C}$ not all zero such that
\begin{equation*}
	a_1z^{\nu_1}F_1(z)+a_2z^{\nu_2}F_2(z)+\ldots+a_pz^{\nu_p}F_p(z)
	=\mathcal{O}\big(z^{-1}\big),\ \ \ \ \textnormal{as}\ \ \ z\rightarrow\infty.
\end{equation*}
With $a_k$ denoting the last of the coefficients $a_1,\ldots,a_p$ which is non-zero we then have
\begin{equation*}
	X(z):=a_1z^{\nu_1-\nu_k}F_1(z)+a_2z^{\nu_2-\nu_k}F_2(z)+\ldots+a_kF_k(z)=\mathcal{O}\big(z^{-1-\nu_k}\big),\ \ \ z\rightarrow\infty,
\end{equation*}
i.e. the degree of $X(z)$ at $z=\infty$ is strictly less than $-\nu_k$. But then necessarily (compare above)
\begin{equation*}
	X(z)=\sum_{i=1}^{k-1}r_i(z)F_i(z)\ \ \ \ \ \ \textnormal{for some}\ \ r_i\in\mathbb{C}[z],
\end{equation*}
and in turn $F_k\in\textnormal{span}\{F_1,\ldots,F_{k-1}\}$, a contradiction. All together,
\begin{equation*}
	 z^{\Lambda}\Phi(z)\ \ \ \ \ \textnormal{with}\ \ \Lambda:=\textnormal{diag}[\nu_1,\ldots,\nu_p]\in\mathbb{Z}^{p\times p}
\end{equation*}
is invertible at $z=\infty$. This concludes the proof of the theorem.
\end{proof}
Plemelj's central Theorem \ref{theo:2} will be the key to solve RHP \ref{RHP:1} in interpretation (2), provided we can dispose of Assumption \ref{assu:1}. In \cite[page $229-236$]{P0} this was achieved by relying on the piecewise constant nature of $G(z)$ in Problem \ref{HP1}. Here, we shall outline a procedure of Vekua \cite[Chapter $2$]{V} which is more widely applicable and which employs multi-valued analytic functions defined in cut planes, equivalently discontinuous single-valued functions. First, we denote with 
\begin{equation*}
	G(a_j\mp 0),\ \ j=1,\ldots,n
\end{equation*}
the left, resp. right limits at the point $a_j\in\Gamma$ as we approach $a_j$ along $(a_{j-1},a_j)$, resp. $(a_j,a_{j+1})$ on $\Gamma$, compare Figure \ref{fig2}. These limits exist by our definition of $G(z)$ and we call $a_{\sigma}\in\{a_1,\ldots,a_n\}$ with $G(a_{\sigma}-0)\neq G(a_{\sigma}+0)$ a \textit{point of discontinuity}. Second, we use the eigenvalues $\lambda_1^{\sigma},\ldots,\lambda_p^{\sigma}$ of the non-singular matrix $G(a_{\sigma}-0)(G(a_{\sigma}+0))^{-1}$ at a point of discontinuity $a_{\sigma}$ written as
\begin{equation*}
	j=1,\ldots,p:\ \ \ \ \lambda_j^{\sigma}=\e^{2\pi\im\gamma_j^{\sigma}},\ \ \ \ \ \ \ \ \ \ \ \ \gamma_j^{\sigma}\in\mathbb{C}:\ -\frac{1}{2}<\textnormal{Re}\,\gamma_j^{\sigma}\leq\frac{1}{2}.
\end{equation*}
Third, we fix $z_0\in\Omega^+$ and let $L^{\sigma}$ denote the straight line connecting $z_0$ with a point of discontinuity $a_{\sigma}$ and extending to $z=\infty$. With this convention any branch of
\begin{equation*}
	d_j^{\,\sigma}(z):=(z-z_0)^{\gamma_j^{\sigma}},\ \ \ z\in\mathbb{C}\setminus L^{\sigma},\ \ \ \ j=1,\ldots,p,
\end{equation*} 
is single-valued and we have
\begin{equation*}
	\forall\,w\in\Gamma\setminus\{a_{\sigma}\}:\ \ \lim_{\substack{z\rightarrow w\\ z\in\Omega^+}}d_j^{\,\sigma}(z)=\lim_{\substack{z\rightarrow w\\ z\in\Omega^-}}d_j^{\,\sigma}(z),\ \ \ \ \ \ \ \ \textnormal{as well as}\ \ \ \ \ \ \ \ 
	\frac{d_j^{\,\sigma}(a_{\sigma}-0)}{d_j^{\,\sigma}(a_{\sigma}+0)}=\e^{2\pi\im\gamma_j^{\sigma}}.
\end{equation*}
We also need an arbitrary branch of the function $(z-a_{\sigma})^{\gamma_j^{\sigma}}$ defined in $\mathbb{C}$ with a cut along $[a_{\sigma},\infty)\subset L^{\sigma}$ and the function $(\frac{z-a_{\sigma}}{z-z_0})^{\gamma_j^{\sigma}}$ uniquely defined in $\mathbb{C}$ with a cut on $[z_0,a_{\sigma}]\subset L^{\sigma}$ such that $(\frac{z-a_{\sigma}}{z-z_0})^{\gamma_j^{\sigma}}\rightarrow 1$ as $z\rightarrow\infty$. We now come to the heart of Vekua's method: define the diagonal matrices
\begin{equation*}
	A_{\sigma}(z):=\bigoplus_{j=1}^p\left(\frac{z-a_{\sigma}}{z-z_0}\right)^{\gamma_j^{\sigma}},\ \ \ \ B_{\sigma}(z):=\bigoplus_{j=1}^p(z-a_{\sigma})^{\gamma_j^{\sigma}},\ \ \ \ C_{\sigma}(z):=\bigoplus_{j=1}^pd_j^{\,\sigma}(z),
\end{equation*}
and note that $A_{\sigma}(z)$ is analytic in $\Omega^-$ whereas any branch of $B_{\sigma}(z)$ is analytic in $\Omega^+$. But if $X(z)$ solves Problem \ref{HP1} then
\begin{equation}\label{e:19}
	T(z):=X(z)\begin{cases}\displaystyle M_{\sigma}^{-1}(B_{\sigma}(z))^{-1},&z\in\Omega^+\smallskip\\
	\displaystyle N_{\sigma}^{-1}(A_{\sigma}(z))^{-1},&z\in\Omega^-\end{cases},
\end{equation}
for properly chosen $M_{\sigma},N_{\sigma}\in\textnormal{GL}(p,\mathbb{C})$, will solve a modified Problem \ref{HP1} with jump matrix which is H\"older continuous on the whole segment $(a_{\sigma-1},a_{\sigma+1})$ including the point $z=a_{\sigma}$. Once $a_{\sigma}$ has been removed in this fashion all remaining points of discontinuity can be treated similarly and we eventually arrive at a Hilbert boundary value problem that satisfies Assumption \ref{assu:1}. So how do we choose $M_{\sigma}$ and $N_{\sigma}$? With \eqref{e:19} the jump condition on $\Gamma\setminus\{a_1,\ldots,a_n\}$ in Problem \ref{HP1} gets replaced by
\begin{align*}
	T_+(z)=&\,T_-(z)A_{\sigma}(z)N_{\sigma}\,G(z)M_{\sigma}^{-1}(B_{\sigma}(z))^{-1}\\
	=&\,
	T_-(z)\Big(\mathbb{I}+B_{\sigma}(z)\Big[(C_{\sigma}(z))^{-1}N_{\sigma}\,G(z)-M_{\sigma}\Big]M_{\sigma}^{-1}(B_{\sigma}(z))^{-1}\Big),
\end{align*}
and this motivates the requirements
\begin{equation}\label{e:20}
	(C_{\sigma}(a_{\sigma}+0))^{-1}N_{\sigma}\,G(a_{\sigma}+0)=M_{\sigma}=(C_{\sigma}(a_{\sigma}-0))^{-1}N_{\sigma}\,G(a_{\sigma}-0)
\end{equation}
to ensure that the jump for $T(z)$ is H\"older continuous on $(a_{\sigma-1},a_{\sigma+1})$. In fact, with \eqref{e:20} in place, the matrix $(C_{\sigma}(z))^{-1}N_{\sigma}\,G(z)-M_{\sigma}$ will be Lipschitz continuous on $(a_{\sigma-1},a_{\sigma+1})$, see \cite[$\S 6$]{M}. In order to make \eqref{e:20} more transparent, we use $C_{\sigma}(a_{\sigma}-0)(C_{\sigma}(a_{\sigma}+0))^{-1}=\bigoplus_{j=1}^p\lambda_j^{\sigma}$ and rewrite \eqref{e:20} as
\begin{equation*}
	N_{\sigma}\,G(a_{\sigma}-0)(G(a_{\sigma}+0))^{-1}N_{\sigma}^{-1}=\bigoplus_{j=1}^p\lambda_j^{\sigma},\ \ \ \ \ \ \ \ \  M_{\sigma}=\big(C_{\sigma}(a_{\sigma}+0)\big)^{-1}N_{\sigma}\,G(a_{\sigma}+0),
\end{equation*}
so $N_{\sigma}$ and $M_{\sigma}$ in \eqref{e:20} exist if $G(a_{\sigma}-0)(G(a_{\sigma}+0))^{-1}$ is diagonalizable. If this is not the case, then one has to properly modify $M_{\sigma},N_{\sigma}$ in \eqref{e:19} and work with the Jordan normal form of $G(a_{\sigma}-0)(G(a_{\sigma}+0))^{-1}$, see \cite[page $87-93$]{V} for details. Either way, a finite number of transformations of the type \eqref{e:19} lead to a Hilbert boundary value problem that is amenable to our previous analysis, i.e. a problem whose jump matrix is H\"older continuous along all of $\Gamma$\footnote{The choice of $\gamma_j^{\sigma}$ intimately ties to the blow up constraint near $z=a_i$ in Problem \ref{HP1} (3), cf. \cite[page $95$]{V}.} and which satisfies a normalization of the form \eqref{inf}.\bigskip

At this point we can come full circle and return to the original RHP \ref{RHP:1}: each function of the fundamental system \eqref{e:15} can be analytically continued to any point in $\mathbb{C}\setminus\{a_1,\ldots,a_n\}$ via any path that does not pass through any of the singularities. Concretely, if $\tau_{\gamma_j}$ denotes the operator of analytic continuation along the fundamental loop $\gamma_j$, compare Figure \ref{fig2}, then by Problem \ref{HP1} condition (2),  any fundamental matrix $\Phi(z)$ satisfies
\begin{equation*}
	\tau_{\gamma_j}(\Phi)(z)=\Phi(z)G_j,\ \ \ \ j=1,\ldots,n.
\end{equation*}
But the same is also true for the matrix-valued function
\begin{equation}\label{e:21}
	\Psi(z):=(z-a_1)^{\Lambda}\Phi(z),\ \ \ z\in\mathbb{CP}^1\setminus\Gamma
\end{equation}
which, according to Theorem \ref{theo:2}, is invertible and bounded at $z=\infty$. Using Liouville's theorem and the piecewise constant, that is $z$-independent, form of $G(z)$ it now follows that the matrix
\begin{equation}\label{e:22}
	A(z):=\frac{\d\Psi}{\d z}(z)(\Psi(z))^{-1},\ \ \ \ z\in\mathbb{CP}^1\setminus\{a_1,\ldots,a_n\}
\end{equation}
is single-valued on $\mathbb{CP}^1$ and analytic everywhere expect at the singularities $a_1,\ldots,a_n\in\mathbb{C}$. Moreover, using Problem \ref{HP1}, condition (3), \eqref{e:21} and our discussion on Vekua's regularization, it is clear that $a_1,\ldots,a_n$ are \textit{regular singular points} of the system \eqref{e:22} in the following sense:
\begin{definition}\label{regdef} Given a system \eqref{e:1} where $A(z)\in\mathbb{C}^{p\times p}$ is analytic in a punctured disk $\mathbb{D}_r(z_0)\setminus\{z_0\}$ for some $r>0$, we say that $z_0$ is a regular singular point of \eqref{e:1} if any solution of the system has at most polynomial growth in a vicinity of $z_0$.
\end{definition}
Since the system \eqref{e:22} also has the given monodromy $(G_1,\ldots,G_n)$, the above analysis solves RHP \ref{RHP:1} in interpretation (2). Precisely
\begin{theo}[Plemelj, 1908]\label{Plecent}
Any matrix group with $n\in\mathbb{Z}_{\geq 1}$ generators $G_1,\ldots,G_n\in\textnormal{GL}(p,\mathbb{C})$ satisfying the constraint $G_1\cdot\ldots\cdot G_n=\mathbb{I}$ can be realized as the monodromy group of a $p\times p$ linear system \eqref{e:1} on $\mathbb{CP}^1$ having only regular singularities.
\end{theo}
Plemelj's 1908 paper does not stop with Theorem \ref{Plecent}: indeed, while all singularities in a Fuchsian system are always regular singular points (this is Sauvage's 1886 theorem, cf. \cite[Theorem $2.1$]{CL} or \cite[Theorem $16.10$]{IY}), the converse statement is in general false. For instance, the non-Fuchsian system
\begin{equation}\label{e:23}
	\frac{\d\Psi}{\d z}=\left\{-\frac{3}{16z^2}\begin{bmatrix}0 & 0\\ 1 & 0\end{bmatrix} +\begin{bmatrix}0 & 1\\ 0 & 0\end{bmatrix}\right\}\Psi
\end{equation}
has a fundamental solution of the form
\begin{equation*}
	\Psi(z)=\frac{1}{4}\begin{bmatrix}4z& 4z\\ 1 & 3\end{bmatrix}z^{-\frac{1}{4}\bigl[\begin{smallmatrix}3 & 0\\ 0 & 1\end{smallmatrix}\bigr]},\ \ \ \ z\in\mathbb{C}\setminus(-\infty,0],
\end{equation*}
and thus $z=0$ is a regular singular point of system \eqref{e:23}. Plemelj was aware of the distinction between regular singular and Fuchsian singular points, thus he applied a gauge procedure that took him from \eqref{e:22} to another system with equal monodromy and same singular points. The transformed system is Fuchsian for all except perhaps one singular point. This part in Plemelj's work \cite[page $243-244$]{P0} is rigorous, unfortunately his subsequent claim that also the remaining singular point can be reduced to a Fuchsian one, and so his claim of having solved RHP \ref{RHP:1} in interpretation (3), is not. This gap went unnoticed for more than 70 years (Plemelj passed away in 1967, likely unaware of it) until it became clear that his argument for Fuchsian systems is valid only with further constraints in place:
\begin{theo}[Kohn, 1983] If at least one of the generators $G_1,\ldots,G_n\in\textnormal{GL}(p,\mathbb{C})$ is diagonalizable, then RHP \ref{RHP:1} for Fuchsian systems on $\mathbb{CP}^1$, i.e. RHP \ref{RHP:2}, has a positive solution.
\end{theo}
Still, Plemelj's \textit{almost solution} of RHP \ref{RHP:2} introduced the idea of rephrasing a dynamical system, here \eqref{e:1} having only regular singularities, as a Hilbert boundary value problem. In turn, the analysis of the problem's underlying singular integral equations led to valuable information about the dynamical system itself, here Theorem \ref{Plecent}. We will encounter both themes throughout our discussions of OPSFA-type problems after the upcoming section.

\section{Hilbert's 21st problem after 1908}
Plemelj's 1908 result became widely accepted as providing a positive answer to RHP \ref{RHP:2}. Thus, in the following 75 years, the field connected to Hilbert's 21st problem shifted its focus towards the effective construction of Fuchsian systems with prescribed monodromy groups. Here are the major results of this period:\bigskip

$\diamond$ In 1913, Birkhoff \cite{Bir} revisited Plemelj's paper \cite{P0} and set out to simplify his argument based on successive approximations while also extending RHP \ref{RHP:1} to certain difference equations.\smallskip

$\diamond$ In the late 1920s, Lappo-Danilevski\u{\i} \cite{LD1,LD2} solved RHP \ref{RHP:2} constructively, provided all generators $G_i$ are sufficiently close to the identity matrix. His method expressed solutions of a Fuchsian system and their associated monodromy via convergent series of the system's matrix coefficients. The solution of Hilbert's 21st problem subsequently boiled down to the problem of inverting the series and studying its convergence.\smallskip

$\diamond$ In 1956, Krylov \cite{Kr} explicitly solved RHP \ref{RHP:2} for all $2\times 2$ systems with $n=3$ singular points. His work made crucial use of Gauss hypergeometric functions.\smallskip

$\diamond$ In 1957, R\"ohrl \cite{Ro} introduced a novel set of algebro-geometric ideas in the analysis of Hilbert's 21st problem. This allowed for a reformulation and generalization of RHP \ref{RHP:1} to holomorphic vector bundles over Riemann surfaces, nowadays summarized under the umbrella of \textit{Riemann-Hilbert correspondences}. Important early contributions to this field were achieved by Deligne \cite{De}, Kashiwara \cite{Ka1,Ka2} and Mebkhout \cite{Meb1,Meb2}.\smallskip

$\diamond$ In 1979, Dekkers \cite{Dek} proved the solvability of RHP \ref{RHP:2} in case $p=2$ for an arbitrary number of singularities. His work did not directly address Hilbert's 21st problem (which was thought to have been solved by Plemelj), but contains its positive solution for $p=2$ as special case.\smallskip

$\diamond$ In 1982, Erugin \cite{Eru} considered RHP \ref{RHP:2} for all $2\times 2$ systems with $n=4$ singular points. His work established a remarkable connection of this specialized problem and the Painlev\'e-VI equation. Compare \cite[Chapter $2$, section $1.3$]{FIKN} for more details and also \cite{J82}.\smallskip

$\diamond$ In 1999, Deift, Its, Kapaev, Zhou \cite{DIKZ} ($p=2$) and in 2004, Korotin \cite{Ko} and Enolski, Grava \cite{EG} (general $p$) solved RHP \ref{RHP:2} for quasi-permutation monodromy matrices by means of algebro-geometric techniques. These works were published after Bolibrukh's negative solution, to be discussed in the following.\bigskip

Notice how all of the above works either yielded a positive solution to some special case of RHP \ref{RHP:2}  or were concerned with generalizations of the same problem to more general Riemann surfaces. However, as we already mentioned, RHP \ref{RHP:2} is in general unsolvable as there are monodromy representations which cannot be representations of any Fuchsian systems. This surprising fact came to light in the important works of Bolibrukh in 1989, cf. \cite{BO2,BO3}. In more detail, Bolibrukh's first counterexample concerns the following $3\times 3$ system with $n=4$ singular points:
\begin{equation}\label{e:24}
	\frac{\d\Psi}{\d z}=A(z)\Psi,\ \ \ A(z)=\left[\begin{array}{ccc}
	0 &\!\!\!\!\!\!\!a_{12}(z) &\!\!\!\!\!\!\!\!a_{13}(z)\\ 
	0 &\,\,\,\,\,\,\,\,\begin{array}{cc}
	& \\
	&B(z)
	\end{array} & \\
	0 & & \end{array}\right],
\end{equation}
where
\begin{equation*}
	a_{12}(z):=\frac{1}{z^2}+\frac{1}{z+1}-\frac{1}{z-\frac{1}{2}},\ \ \ \ \ \ \ \ a_{13}(z):=\frac{1}{z-1}-\frac{1}{z-\frac{1}{2}},
\end{equation*}
and we use the $2\times 2$ block
\begin{equation*}
	B(z):=\frac{1}{z}\begin{bmatrix}1 & 0\\ 0 & -1\end{bmatrix}+\frac{1}{6(z+1)}\begin{bmatrix}-1 & 1\\ -1 & 1\end{bmatrix}+\frac{1}{2(z-1)}\begin{bmatrix}-1 & -1\\ \,\,\,\,1 & \,\,\,\,1\end{bmatrix}+\frac{1}{3(z-\frac{1}{2})}\begin{bmatrix}-1 & 1\\ -1 & 1\end{bmatrix}.
\end{equation*}
This system has Fuchsian singularities at $a_1=-1,a_2=\frac{1}{2},a_3=1$ and a non-Fuchsian singularity at $a_4=0$. Although non-Fuchsian, Bolibrukh then showed that those singularities are regular singular points of \eqref{e:24} in the sense of Definition \ref{regdef}. Moreover he proved that system \eqref{e:24} has non-trivial monodromy, but \textit{there exists no Fuchsian system with the same mondromy group and encoded singularity locations}. Thus RHP \ref{RHP:2} is in general unsolvable. The details of Bolibrukh's argument can be found in the monograph \cite[Chapter $2$]{AB} and they are very different from Plemelj's analysis of the boundary value problem \ref{HP1}. Rather than discussing those techniques we only mention two features of Bolibrukh's counterexample \eqref{e:24}. One, the sensitive dependence on the location of the singular points: once slightly perturbed, the answer to RHP \ref{RHP:2} with the same monodromy can become positive, see \cite{BO2}. Two, the mondromy representation of system \eqref{e:24} is reducible, a fact which intimately links to the following general result by Kostov \cite{Kos} and Bolibrukh \cite{BO22,BO3}:
\begin{theo}[Kostov, Bolibrukh, 1992]\label{KoBo} For every irreducible representation \eqref{e:2}, the Riemann-Hilbert problem \ref{RHP:2} has a positive solution.
\end{theo}
In the years following 1992 and Theorem \ref{KoBo}, Bolibrukh sharpened his results and derived a further series of sufficient conditions for the positive solvability of RHP \ref{RHP:2}, compare \cite{BO4}, these are conditions formulated in the language of holomorphic vector bundles and we refer the reader to the excellent monograph \cite{AB} for details.\bigskip

In closing of this short section, and our content on the original RHP \ref{RHP:1}, we emphasize that Bolibrukh's application of methods from complex analytic geometry provided a definite, albeit negative, solution to Hilbert's 21st problem for Fuchsian systems on $\mathbb{CP}^1$. This problem was formulated back in 1900 and Bolibrukh lectured on his breakthroughs precisely 94 years after Hilbert at the ICM 1994 in Z\"urich, sadly a few years later he passed away.

\section{Developments tangential to Plemelj's work - five examples}\label{ex:sec}
At this point we readjust our focus and turn away from Hilbert's 21st problem, the original RHP. Instead we highlight certain developments parallel to Plemelj's 1908 work on Hilbert boundary value problems and associated singular integral equations which found multiple applications in the OPSFA realm. Those appearances gave rise to an analytic apparatus, the \textit{Riemann-Hilbert techniques}, which we are about to highlight through five examples - Section \ref{ex:6} contains another example, but it is of different technical nature and more advanced. Overall, it is important to trace several of these techniques back to their origin \cite{P0}, namely to Plemelj's work on RHP \ref{RHP:1}.
\subsection{The Wiener-Hopf method} The first member of our analytical toolbox, a.k.a. the Riemann-Hilbert techniques, is the \textit{Wiener-Hopf} method (pioneered by Wiener and Hopf in 1931 \cite{WH}, but applied implicitly by Carleman before 1931). This method is commonly used in hydrodynamics, diffraction or linear elasticity theory and was originally developed for the investigation of the Milne-Schwarzschild integral equation
\begin{equation}\label{wh:1}
	x>0:\ \ \ \phi(x)=\frac{1}{2}\int_0^{\infty}E(|x-y|)\phi(y)\,\d y,\ \ \ \ E(x):=\int_x^{\infty}\frac{\e^{-t}}{t}\,\d t,	
\end{equation}
which is used in models for radiative processes in astrophysics. Rather than discussing the Wiener-Hopf method in the more general setup
\begin{equation}\label{WHint}
	\lambda\phi(x)+\int_0^{\infty}k(x-y)\phi(y)\,\d y=f(x),\ \ \ x>0,
\end{equation}
we will simply use the original problem \eqref{wh:1}. How do we solve \eqref{wh:1} with Riemann-Hilbert techniques? According to the Wiener-Hopf recipe one first extends \eqref{wh:1} to the full real line via $\phi(x):=0$ for $x\leq 0$ and obtains
\begin{equation}\label{wh:2}
	x\in\mathbb{R}:\ \ \ \phi(x)=\int_{-\infty}^{\infty}K(x-y)\phi(y)\,\d y+\psi(x),\ \ \ \ \ \ \ \ \ \ K(x):=\frac{1}{2}E(|x|),
\end{equation}
where $\psi(x):=0$ for $x>0$ and otherwise appropriately chosen for $x\leq 0$ to achieve equality in \eqref{wh:2}. Then one takes a formal Fourier transform of \eqref{wh:2} and finds
\begin{equation}\label{wh:3}
	\Phi_+(z)=\Phi_+(z)L(z)+\Psi_-(z),\ \ \ z\in\mathbb{R},
\end{equation}
with
\begin{equation*}
	\Phi_+(z):=\int_0^{\infty}\e^{\im z x}\phi(x)\,\d x,\ \ \ \Psi_-(z):=\int_{-\infty}^0\e^{\im z x}\psi(x)\,\d x\ \ \ \ \ \textnormal{and}\ \ \ \ L(z):=\int_{-\infty}^{\infty}\e^{\im z x}K(x)\,\d x.
\end{equation*}
The subscripts $\pm$ in \eqref{wh:3} indicate our desire that $\Phi_+(z)$, resp. $\Psi_-(z)$ admit analytic continuation to some part of the upper, resp. lower $z$-plane - in fact we will seek solutions of \eqref{wh:1} which do \textit{not} exponentially grow at $x=+\infty$, and this implies that $\Phi_+(z)$ is the limiting value of a function $\Phi(z)$ analytic for $\Im z>0$ and likewise $\Psi_-(z)$ the limiting value of a function $\Psi(z)$ analytic for $\Im z<0$.
%
Next, by elementary calculus, we find
\begin{equation}\label{cute}
	L(z)=\frac{1}{z}\arctan(z)=\frac{\im}{2z}\ln\left(\frac{\im+z}{\im-z}\right),\ \ z\in\mathbb{C}\setminus\big((-\im\infty,-\im]\cup[\im,\im\infty)\big),
\end{equation}
and now Riemann-Hilbert techniques enter our calculation.
\begin{prob}\label{WHRHP} Determine $F(z)\in\mathbb{C}$ such that
\begin{enumerate}
	\item[(1)] $F(z)$ is analytic in $\mathbb{C}\setminus\mathbb{R}$ and extends continuously from either side to the real axis.
	\item[(2)] On the real axis, the pointwise limits
	\begin{equation*}
		F_{\pm}(z):=\lim_{\epsilon\downarrow 0}F(z\pm\im\epsilon)
	\end{equation*}
	satisfy the boundary condition $F_+(z)=F_-(z)G(z)$ with
	\begin{equation*}
		z\in\mathbb{R}:\ \ \ \ G(z)=(1-L(z))\left(1+\frac{1}{z^2}\right)\in\mathbb{C}.
	\end{equation*}.
	\item[(3)] Near $z=\infty$ we enforce the asymptotic behavior
	\begin{equation*}
		F(z)=1+\mathcal{O}\big(z^{-1}\big),\ \ \ \ |z|\rightarrow\infty.
	\end{equation*}
\end{enumerate}
\end{prob}
Since $G(z)$ is H\"older continuous on $\mathbb{R}$, $G(z)\neq 0$ for all $z\in\mathbb{R}$ and $G(z)=1+\mathcal{O}(z^{-1}\big)$ as $z\rightarrow\pm\infty$, the unique solution of Problem \ref{WHRHP} is readily seen to be
\begin{equation}\label{wh:4}
	F(z)=\exp\left[\frac{1}{2\pi\im}\int_{-\infty}^{\infty}\ln G(w)\frac{\d w}{w-z}\right],\ \ z\in\mathbb{R}\setminus\mathbb{C},
\end{equation}
by the Plemelj-Sokhotski formul\ae. The point of \eqref{wh:4} is that we can use the underlying \textit{Wiener-Hopf} factorization
\begin{equation*}
	\frac{F_+(z)}{F_-(z)}=\big(1-L(z)\big)\left(1+\frac{1}{z^2}\right),\ \ z\in\mathbb{R}
\end{equation*}
back in \eqref{wh:3} and obtain
\begin{equation*}
	\frac{z^2}{z+\im}\Phi_+(z)F_+(z)=(z-\im)\Psi_-(z)F_-(z),\ \ z\in\mathbb{R}.
\end{equation*}
Here, the equation's left hand side defines a function which admits analytic continuation to the upper half-plane and the right hand side a function which admits analytic continuation to the lower half-plane. Hence, there is an entire function $E(z),z\in\mathbb{C}$ such that
\begin{equation*}
	\Phi(z)=\frac{(z+\im)E(z)}{z^2F(z)},\ \ \Im z>0;\ \ \ \ \ \ \Psi(z)=\frac{E(z)}{(z-\im)F(z)},\ \ \Im z<0.
\end{equation*}
But $\Phi_+(z)$ is a Fourier integral, so we find $\Phi(z)=\mathcal{O}(z^{-1})$ as $z\rightarrow\infty, \Im z>0$ and thus with Problem \ref{WHRHP} and Liouville's theorem, $E(z)\equiv \im c\in\mathbb{C}$. All together,
\begin{equation*}
	\Phi(z)=\frac{(z+\im)\im c}{z^2F(z)},\ \ \Im z>0.
\end{equation*}
This result at hand, we would now like to return to $\phi(x)$ via an inverse Fourier transform
\begin{equation}\label{wh:5}
	\phi(x)=\frac{1}{2\pi}\int_{\Gamma}\Phi(z)\e^{-\im zx}\,\d z,
\end{equation}
but how do we choose the contour $\Gamma\subset\mathbb{C}$ in \eqref{wh:5} to achieve equality? Notice that by the Plemelj-Sokhotski formula, $F_+(0)=\frac{1}{\sqrt{3}}\neq 0$, so the continuation of $\Phi(z)$ to the lower half-plane will have a double pole at $z=0$ and a branch cut extending from $-\im$ to $-\im\infty$ along the imaginary axis, compare \eqref{cute}. Hence, if $x<0$, we choose $\Gamma$ as straight line in the upper half-plane and find that \eqref{wh:5} vanishes by residue theorem. If $x>0$, then we wrap $\Gamma$ around the negative imaginary axis and encircle the double pole at $z=0$ as well as the branch cut, compare Figure \ref{figWH1}. 
\begin{figure}[tbh]
\includegraphics[width=0.38\textwidth]{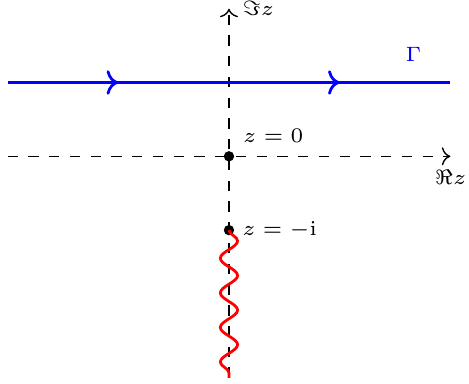}\,\,\,\,\,\,\,\,\,
\includegraphics[width=0.38\textwidth]{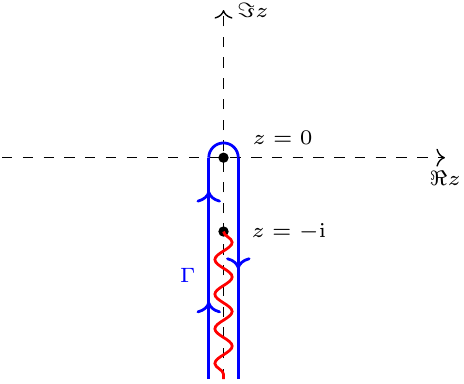}
\caption{The integration contour $\Gamma$ for \eqref{wh:5} shown in \textcolor{blue}{blue} with branch cut in \textcolor{red}{red}. On the left our choice for $x<0$ and on the right for $x>0$.}
\label{figWH1}
\end{figure}

Computing the residue and parametrizing the integral along the branch cut, one then obtains the well-known solution formula
\begin{align*}
	\phi(x)=c\Bigg(\sqrt{3}\bigg[1+x&\,-\frac{1}{\pi}\int_0^{\infty}\left\{\frac{1}{1-L(w)}-1-\frac{3}{w^2}\right\}\frac{\d w}{1+w^2}\bigg]\\
	&\,+\frac{\pi}{2}\e^{-x}\int_0^{\infty}\frac{\e^{-xw}}{F(-\im-\im w)}\frac{\d w}{(2+w-\frac{1}{2}\ln(\frac{2+w}{w}))^2+\frac{\pi^2}{4}}\Bigg),\ \ x>0.
\end{align*}
In summary of this short subsection, using scalar Hilbert boundary value problems we can solve
integral equations of the form \eqref{WHint} by the Wiener-Hopf method. This method has found countless applications and is particularly well-suited for PDE problems with semi-infinite boundaries, e.g. the Sommerfeld diffraction problem, see for instance \cite{Nob} and \cite[Chapter $16$]{Dav} for further details.

\subsection{The integrable systems revolution of the late 1960s} The classical (think of Euler, Hamilton, Jacobi, Liouville, Neumann and Kowaleskaya) field of integrable systems found renewed interest in the late 1960s when it became clear that several nonlinear PDEs in $1+1$ dimensions can be \textit{integrated} by the inverse scattering method. Examples are the Kortweg-de Vries equation, the nonlinear Schr\"odinger equation (NLS), and the sine-Gordon equation, among many others. Important early contributions to this remarkable technique are due to Gardner, Green, Kruskal, and Miura \cite{GGKM}, due to Lax \cite{Lax}, due to Faddeev, Zakharov \cite{ZF}, and due to Shabat, Zakharov \cite{ZS}. For a solid first introduction to the inverse scattering method we refer the interested reader to the monographs \cite{NMPZ,AC,BDT,FT}, here we shall discuss only one example and showcase the appearance of Riemann-Hilbert techniques: Consider the defocusing NLS
\begin{equation}\label{e:25}
	\mathrm{i}y_t+y_{xx}-2|y|^2y=0,\ \ \ \ y=y(x,t):\mathbb{R}^2\rightarrow\mathbb{C}.\smallskip
\end{equation}
How do we solve the corresponding initial value problem with Cauchy data $y(x,0)=y_0(x)\in\mathcal{S}(\mathbb{R})$ in the Schwartz space on the real line? Well, one first computes the reflection coefficient $r(z)\in\mathcal{S}(\mathbb{R})$ associated to $y_0$ through the direct scattering transform. There is not much freedom in this computation as the map $y_0\rightarrow r$ is a bijection from $\mathcal{S}(\mathbb{R})$ onto $\mathcal{S}(\mathbb{R})\cap\{r:\,\sup_{z\in\mathbb{R}}|r(z)|<1\}$, see \cite{BC}. After that, one considers the so-called inverse scattering problem which is most conveniently formulated as the following Hilbert boundary value problem - this insight grew out of the works by Manakov, Shabat and Zakharov in the mid 1970s:
\begin{prob}\label{HP3} For any $x\in\mathbb{R}$ and $t>0$, determine $X(z)=X(z;x,t)\in\mathbb{C}^{2\times 2}$ such that
\begin{enumerate}
	\item[(1)] $X(z)$ is analytic in $\mathbb{C}\setminus\mathbb{R}$ and extends continuously from either side to the real axis.
	\item[(2)] On the real axis, the pointwise limits
	\begin{equation*}
		X_{\pm}(z):=\lim_{\epsilon\downarrow 0}X(z\pm\im\epsilon)
	\end{equation*}
	satisfy the boundary condition $X_+(z)=X_-(z)G(z)$ with
	\begin{equation}\label{e:266}
		z\in\mathbb{R}:\ \ \ \ \ G(z)=\begin{bmatrix}1-|r(z)|^2 & -\bar{r}(z)\e^{-2\im(2tz^2+xz)}\smallskip\\ r(z)\e^{2\im(2tz^2+xz)} & 1\end{bmatrix}\in\textnormal{GL}(2,\mathbb{C}).
	\end{equation}
	\item[(3)] Near $z=\infty$ we enforce the asymptotic behavior
	\begin{equation*}
		X(z)=\mathbb{I}+X_1z^{-1}+X_2z^{-2}+\mathcal{O}\big(z^{-3}\big),\ \ \ \ \ |z|\rightarrow\infty;\ \ \ \ X_k=X_k(x,t)=\big[X_k^{ij}(x,t)\big]_{i,j=1}^2,
	\end{equation*}
	where $\mathbb{I}\in\textnormal{GL}(2,\mathbb{C})$ is the identity matrix.
\end{enumerate}
\end{prob}
Indeed, provided this problem is solvable, its (unique) solution solves the PDE \eqref{e:25} with initial condition $y(x,0)=y_0(x)$ via the formula
\begin{equation}\label{e:26}
	y(x,t)=2\im X_1^{12}(x,t).
\end{equation}
In order to arrive at \eqref{e:26} we use a slight modification of Plemelj's argument that lead him to the linear system \eqref{e:22}: The jump matrix $G(z),z\in\Gamma$ in Problem \ref{HP1} was piecewise constant whereas \eqref{e:266} presents us with the other extreme, namely a jump matrix which depends on $z,x$ and $t$. However, assuming that Problem \ref{HP3} is solvable, we can define
\begin{equation*}
	Y(z)\equiv Y(z;x,t):=X(z;x,t)\e^{-\im(2tz^2+xz)\sigma_3},\ \ \ \ z\in\mathbb{C}\setminus\mathbb{R},
\end{equation*}
and then conclude that both,
\begin{equation}\label{e:27}
	A(z):=\frac{\partial Y}{\partial x}(z)(Y(z))^{-1}\ \ \ \ \textnormal{and}\ \ \ \ B(z):=\frac{\partial Y}{\partial t}(z)(Y(z))^{-1},\ \ \ \ (z,x,t)\in\mathbb{C}\times\mathbb{R}\times(0,\infty)
\end{equation}
are single-valued, entire functions in $z\in\mathbb{C}$. Moreover, using Problem \ref{HP3}, condition (3) and Liouville's theorem we obtain the explicit formul\ae
\begin{align*}
	A(z)\equiv A(z;x,t)=&\,-\im z\sigma_3+2\im\begin{bmatrix}0 & X_1^{12}\smallskip\\ -X_1^{21} & 0\end{bmatrix},\\
	B(z)\equiv B(z;x,t)=&\,-2\im z^2\sigma_3+4\im z\begin{bmatrix}0 & X_1^{12} \smallskip\\ -X_1^{21} & 0\end{bmatrix}+4\im\begin{bmatrix}-X_1^{12}X_1^{21} & X_2^{12}-X_1^{12}X_1^{22}\smallskip\\
	-X_2^{12}+X_1^{21}X_1^{11} & X_1^{21}X_1^{12}\end{bmatrix},
\end{align*}
in term of the matrix coefficients $X_k^{ij}=X_k^{ij}(x,t)$ and the third Pauli matrix $\sigma_3:=\bigl[\begin{smallmatrix}1 & 0\\ 0 & -1\end{smallmatrix}\bigr]$. The overdetermined system \eqref{e:27} forms one of the celebrated Lax pairs for the PDE \eqref{e:25}, indeed writing out the Frobenius integrability condition for $\frac{\partial Y}{\partial x}=AY,\frac{\partial Y}{\partial t}=BY$, equivalently the compatibility condition
\begin{equation*}
	AB-BA=\frac{\partial B}{\partial x}-\frac{\partial A}{\partial t},
\end{equation*}
we find entrywise the coupled partial differential equations
\begin{align*}
	2\im(X_1^{12})_x=&\,\,4(X_2^{12}-X_1^{12}X_1^{22}),\ \ \ \ \ \ \ -2\im(X_1^{21})_x=4(X_2^{21}-X_1^{21}X_1^{11}),\\
	4\im(X_2^{12}-X_1^{12}X_1^{22})_x-2\im(X_1^{12})_t=&\,-16X_1^{12}X_1^{21}X_1^{12},\\ -4\im(X_2^{21}-X_1^{21}X_1^{11})_x+2\im(X_1^{21})_t=&\,-16X_1^{21}X_1^{12}X_1^{21}.
\end{align*}
However, the coefficients back in condition (3) of Problem \ref{HP3} satisfy $\tr X_1=0$ and $X_k^{11}=\overline{X_k^{22}},X_k^{21}=\overline{X_k^{12}},k\in\mathbb{Z}_{\geq 1}$. Thus the above four coupled equations boil down to
\begin{equation*}
	y_x=4(X_2^{12}-X_1^{12}X_1^{22}),\ \ \ \ \ \ \ \ \ \ 2|y|^2y=4(X_2^{12}-X_1^{12},X_1^{22})_x+\im y_t,
\end{equation*}
where we used \eqref{e:26} and so after substitution into one another
\begin{equation*}
	\im y_t+y_{xx}-2|y|^2y=0,
\end{equation*}
which is the defocusing NLS \eqref{e:25}. Summarizing this short computation, provided Problem \ref{HP3} is solvable, the linear Hilbert boundary value problem solves the nonlinear evolution equation \eqref{e:25} via the formula \eqref{e:26}. This is a far reaching generalization of Plemelj's idea to construct solutions of certain linear ODE systems in terms of solutions of Hilbert boundary value problems, nevertheless the common approach to both problems is clearly visible. In Section \ref{weakRHP} below we will address the solvability question of Problem \ref{HP3} which requires a different toolset than the one used by Plemelj in his analysis of Problem \ref{HP1}.
\subsection{Painlev\'e special function theory} Painlev\'e functions form a family of special functions which is widely regarded as the substitute for classical special functions in nonlinear mathematical physics (such Airy, Bessel or Hypergeometric functions). Although Painlev\'e transcendents are non-expressible in terms of a finite number of contour integrals over elementary, elliptic or finite genus algebraic functions, several of their key analytic and asymptotic properties can be studied with Riemann-Hilbert techniques. These techniques are therefore in a way the analogues of contour integral representations and steepest descent asymptotic methods used in the analysis of classical special functions. For a comprehensive introduction to Painlev\'e special functions and Riemann-Hilbert techniques associated with them, including several references to their applications in mathematical physics, we recommend the two monographs \cite{GLS,FIKN}. Similar to the last subsection we will showcase the Riemann-Hilbert approach to one particular Painlev\'e equation, namely
\begin{equation}\label{e:28}
	u_{xx}=xu+2u^3,\ \ \ \ u=u(x):\mathbb{R}\rightarrow\mathbb{R},
\end{equation}
the homogeneous Painlev\'e-II equation, which, to a certain extent, is a nonlinear Airy equation. How do we solve the associated initial value problem? Well, in the Riemann-Hilbert approach to \eqref{e:28} one parametrizes the solutions $u=u(x)$ of \eqref{e:28} \textit{not} in terms of Cauchy data, but in terms of the monodromy data of an associated linear system of ordinary differential equations: Precisely, define the generically two-dimensional real manifold
\begin{equation*}
	\mathcal{S}:=\big\{(s_1,s_2,s_3)\in\mathbb{C}^3:\,s_1-s_2+s_3+s_1s_2s_3=0,\ \ s_1=\bar{s}_3,\ \ s_2=\bar{s}_2\big\},
\end{equation*}
and introduce with $s_{k+3}=-s_k,k=1,2,3$ for $k=1,\ldots,6$ the triangular matrices
\begin{equation*}
	S_k:=\begin{bmatrix}1 & 0\\ s_k& 0\end{bmatrix},\ \ k\equiv 1\mod 2;\ \ \ \ \ \ \ S_k:=\begin{bmatrix}1 & s_k\\ 0 & 1\end{bmatrix},\ \ k\equiv 0\mod 2.
\end{equation*}
Now consider the below Hilbert boundary value problem.
\begin{prob}\label{HP4} For any $x\in\mathbb{R}$, determine $X(z)=X(z;x,s)\in\mathbb{C}^{2\times 2}$ with $s:=(s_1,s_2,s_3)\in\mathcal{S}$ such that
\begin{enumerate}
	\item[(1)] $X(z)$ is analytic in $\mathbb{C}\setminus\bigcup_{i=1}^6\Gamma_i$ and continuous on the closed sectors $\overline{\Omega}_i$. Compare Figure \ref{fig4} for the six rays $\Gamma_i$ and the sectors $\Omega_i$ in between them.
	\item[(2)] On the rays $\Gamma_i,i=1,\ldots,6$ the pointwise limits
	\begin{equation*}
		X_{\pm}(z):=\lim_{\substack{w\rightarrow z\\ w\in\Omega_{\pm}}}X(w),\ \ \ \Omega_+:=\Omega_{i+1},\ \ \ \Omega_-:=\Omega_i,\ \ \ \Omega_7\equiv\Omega_1,
	\end{equation*}
	satisfy the boundary condition $X_+(z)=X_-(z)G(z)$ with
	\begin{equation*}
		z\in\Gamma_i:\ \ \ \ G(z)=\e^{-\theta(z,x)\sigma_3}S_i\,\e^{\theta(z,x)\sigma_3}\in\textnormal{GL}(2,\mathbb{C});\ \ \ \ \theta(z,x):=\im\left(\frac{4}{3}z^3+xz\right).
	\end{equation*}
	\item[(3)] Near $z=\infty$ we enforce the asymptotic behavior
	\begin{equation*}
		X(z)=\mathbb{I}+X_1z^{-1}+X_2z^{-2}+\mathcal{O}\big(z^{-3}\big),\ \ \ |z|\rightarrow\infty;\ \ \ \ X_k=X_k(x,s)=\big[X_k^{ij}(x,s)\big]_{i,j=1}^2.
	\end{equation*}
\end{enumerate}
\end{prob}
\begin{figure}[tbh]
\includegraphics[width=0.55\textwidth]{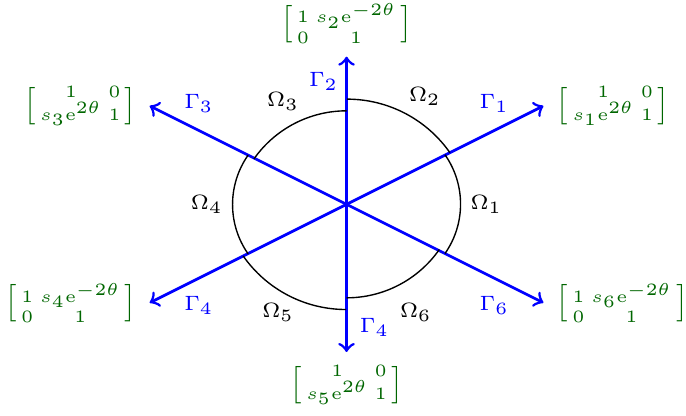}
\caption{The oriented contour $\bigcup_{i=1}^6\Gamma_i$ in \textcolor{blue}{blue} together with the sectors $\Omega_i$ in between. The six rays are $\Gamma_i=\{z\in\mathbb{C}:\,\textnormal{arg}\,z=\frac{\pi}{6}+\frac{\pi}{3}(i-1)\}$ and we indicate the values of $G(z)$ on them in \textcolor{ao}{green}.}
\label{fig4}
\end{figure}
The route between Problem \ref{HP4} and \eqref{e:28} goes as follows: for any choice of parameters $(s_1,s_2,s_3)\in\mathcal{S}$ the Hilbert boundary value problem \ref{HP4} is meromorphically with respect to $x$ solvable, cf. \cite{BIK}. In turn the unique solution to said Problem leads to a solution of \eqref{e:28} via
\begin{equation}\label{e:29}
	u(x)\equiv u(x|s)=2X_1^{12}(x,s),
\end{equation}
which is the analogue of the NLS formula \eqref{e:26}. Additionally, \eqref{e:29} satisfies $\bar{u}(x)=u(\bar{x})$, i.e. the solution is real-valued on the real axis and conversely, every on the real axis real-valued solution of \eqref{e:28} admits a unique Riemann-Hilbert representation \eqref{e:29} for suitable $s$. In short, we do parametrize solutions of \eqref{e:28} in terms of the monodromy data $(s_1,s_2,s_3)\in\mathcal{S}$ using a bijection between the initial value solution space of \eqref{e:28} and $\mathcal{S}$.\smallskip

The derivation of \eqref{e:29}, and thus the Riemann-Hilbert approach to the second Painlev\'e equation is due to Flaschka, Newell \cite{FN} and Jimbo, Miwa, Ueno \cite{JMU,JM1,JM2} in the early 1980s. In a nutshell, we use again Plemelj's basic idea: suppose Problem \ref{HP4} is solvable and set
\begin{equation*}
	Y(z)\equiv Y(z;s,s):=X(z;x,s)\e^{-\theta(z,x)\sigma_3},\ \ \ z\in\mathbb{C}\setminus\bigcup_{i=1}^6\Gamma_i.
\end{equation*}
Then both,
\begin{equation*}
	A(z):=\frac{\partial Y}{\partial z}(Y(z))^{-1}\ \ \ \ \ \ \textnormal{and}\ \ \ \ \ \ \ B(z):=\frac{\partial Y}{\partial x}(z)(Y(z))^{-1},\ \ \ \ \ (z,x)\in\mathbb{C}\times(\mathbb{R}\setminus D)
\end{equation*}
are single-valued, entire functions in $z\in\mathbb{C}$ away from a discrete set $D\subset\mathbb{R}$. Hence, using condition (3) in Problem \ref{HP4} and Liouville's theorem, we obtain the explicit formul\ae
\begin{align*}
	A(z)\equiv A(z;x,s)=&\,\,-4\im z^2\sigma_3+8\im z\begin{bmatrix}0 & X_1^{12}\smallskip\\
	-X_1^{21} & 0\end{bmatrix}+\begin{bmatrix}-\im x-8\im X_1^{12}X_1^{21} & 8\im(X_2^{12}-X_1^{12}X_1^{22})\smallskip \\ -8\im(X_2^{21}-X_1^{21}X_1^{11})& \im x+8\im X_1^{21}X_1^{12}\end{bmatrix},\\
	B(z)\equiv B(z;x,s)=&\,\,-\im z\sigma_3+2\im\begin{bmatrix}0 & X_1^{12}\smallskip\\
	-X_1^{21} & 0\end{bmatrix},
\end{align*}
in terms of the matrix coefficients $X_k$ and from those, reading the compatibility condition
\begin{equation*}
	AB-BA=\frac{\partial B}{\partial z}-\frac{\partial A}{\partial x},
\end{equation*}
entrywise, the following coupled ordinary differential equations,
\begin{align*}
	(X_1^{12})_x=&\,\,-2\im(X_2^{12}-X_1^{12}X_1^{22}),\ \ \ \ \ \ \ \ (X_1^{21})_x=2\im(X_2^{21}-X_1^{21}X_1^{11}),\\
	-2(X_2^{12}-X_1^{12}X_1^{22})_x=&\,\,X_1^{12}(-\im x-8\im X_1^{12}X_1^{21}),\\
	-2(X_2^{21}-X_1^{12}X_1^{11})_x=&\,\,X_1^{21}(\im x+8\im X_1^{21}X_1^{12}).
\end{align*}
However Problem \ref{HP4} also possesses certain implicit symmetries that yield $\tr X_1=0$, $X_k^{12}=(-1)^{k-1}X_k^{21},k\in\mathbb{Z}_{\geq 1}$ and which simplify the above four equations to
\begin{equation*}
	u_x=-4\im(X_2^{12}-X_1^{12}X_1^{22}),\ \ \ \ \ \ \ \ \ \im(xu+2u^3)=4(X_2^{12}-X_1^{12}X_1^{22})_x,
\end{equation*}
where we used \eqref{e:29}. Substituting both equations into one another we find
\begin{equation*}
	u_{xx}=xu+2u^3,
\end{equation*}
and therefore Painlev\'e-II \eqref{e:28}. In summary, the Hilbert boundary value problem \ref{HP4} linearizes the nonlinear Painlev\'e-II ODE and therefore shares common ground with \eqref{e:22} and \eqref{e:25}. Before we move to our next example it will be important to mention that the monodromy data $s=(s_1,s_2,s_3)\in\mathcal{S}$ forms a complete set of first integrals for \eqref{e:28}. Building on this feature, one can in principle use the highly transcendental equations
\begin{equation}\label{e:30}
	s_i=s_i(x,u,u_x)\equiv\textnormal{const},\ \ \ \ i=1,2,3
\end{equation}
in the further analysis of the Painlev\'e-II transcendents. This idea is at the heart of the \textit{isomonodromy method} and the corresponding direct monodromy approach to Painlev\'e equations as pioneered by Its, Novokshenov, Kapaev and Kitaev in the mid 1980s, cf. \cite{IN} or \cite[Chapter $7$]{FIKN}. Their approach relies on complex WKB techniques applied to Problem \ref{HP4} and certain a priori assumptions on the behavior of $u=u(x)$. We will later on highlight a different method, the inverse monodromy approach, a.k.a. Deift-Zhou method \cite{DZ}, which (in the extended form of Deift, Venakides, Zhou \cite{DVZ}) has become extremely popular in nonlinear mathematical physics since its discovery back in 1993. This method is suitable for the analysis of Painlev\'e functions and other problems that we are about to discuss.

\subsection{The Heisenberg antiferromagnet}\label{Hei} Our next example is concerned with a quantum mechanical model for an antiferromagnetic spin chain. It was introduced by Lieb, Schultz and Mattis \cite{LSM} in 1961 and is known as the spin-$\frac{1}{2}$XY model, equivalently the Heisenberg XX0 antiferromagnet. In this model $\frac{1}{2}$-spins are situated on a one-dimensional periodic and isotropic lattice, we allow only nearest neighbor interactions between the spins, neighboring spins tend to point in \textit{opposite} directions and a moderate external transverse magnetic field influences the spins. Subject to these four constraints the Hamiltonian of the model equals, cf. \cite[$(2.1)$]{LSM},
\begin{equation*}
	\mathscr{H}:=\sum_{j=1}^N\sigma_j^x\sigma_{j+1}^x+\sum_{j=1}^N\sigma_j^y\sigma_{j+1}^y-h\sum_{j=1}^N\sigma_j^z,
\end{equation*}
where we set the nearest neighbor coupling constant to unity, $0<h<2$ denotes the strength of the external magnetic field and
\begin{equation*}
	\sigma^x\equiv\sigma_1:=\begin{bmatrix}0 & 1\\ 1 & 0\end{bmatrix},\ \ \ \sigma^y\equiv\sigma_2:=\begin{bmatrix}0 & -\im\\ \im & 0\end{bmatrix},\ \ \ \sigma^z=\sigma_3=\begin{bmatrix}1 & 0\\ 0 & -1\end{bmatrix},
\end{equation*}
are the Pauli spin matrices. The paper \cite{LSM} computed the model's ground state, all elementary excitations and the free energy through creation-annihilation operator techniques. Here we shall focus on an important correlation function, the so-called \textit{emptiness formation probability}
\begin{equation*}
	P_n:=\textnormal{Prob}\big\{\textnormal{there are}\,\,n\in\mathbb{Z}_{\geq 1}\,\,\textnormal{adjacent parallel spins up in the ground state}\big\},
\end{equation*}
which was evaluated in the thermodynamic limit $N\rightarrow\infty$ for fixed $h$ in \cite[$(10.2),(10.3)$]{EFIK} by the quantum inverse scattering method. The result is the Fredholm determinant formula
\begin{equation}\label{e:31}
	P_n=\det\big(1-U_n\upharpoonright_{L^2(J,\d x)}\big):=\sum_{k=0}^{\infty}\frac{(-1)^k}{k!}\int_J\cdots\int_J\det\big[U_n(x_i,x_j)\big]_{i,j=1}^k\d x_1\cdots\d x_k,
\end{equation}
where $U_n:L^2(J,\d x)\rightarrow L^2(J,\d x)$ denotes the trace class operator with kernel
\begin{equation*}
	U_n(\lambda,\mu):=\frac{1}{2\pi\im}\frac{1}{\sinh(\lambda-\mu)}\left\{1-\left(\frac{\e^{2\lambda}+\im}{\e^{2\lambda}-\im}\right)^n\left(\frac{\e^{2\mu}-\im}{\e^{2\mu}+\im}\right)^n\right\},
\end{equation*}
acting on the interval $J:=(-\Lambda,\Lambda)\subset\mathbb{R}$ with $\cosh 2\Lambda=\frac{2}{h}>1$. The central analytical challenges associated with $P_n$ consist in accessing its large $n$ asymptotic behavior and in identifying an underlying integrable system for it. Both questions yield physically relevant information for the spin chain model, so how do we go about them? First, one factorizes the operator $U_n$ as
\begin{equation}\label{opfac}
	U_n=E^{-1}D_nK_n(E^{-1}D_n)^{-1},
\end{equation}
where $D_n$ denotes multiplication $(D_nf)(z):=z^{-\frac{n}{2}}\e^{\lambda(z)}(z+\im)$ with the invertible change of variables $\lambda=\lambda(z)$ defined as
\begin{equation}\label{e:32}
	z=:-\im\frac{\e^{2\lambda}-\im}{\e^{2\lambda}+\im},\ \ \ \ z(J)=\big\{z\in\mathbb{C}:\,|z|=1,\ \textnormal{arg}\,\in(\phi,2\pi-\phi)\big\},\ \ \e^{\im\phi}=z(-\Lambda),\ \phi\in\left(\frac{\pi}{2},\pi\right),
\end{equation}
$E$ is the evaluation $(Ef)(z):=f(\lambda(z))$ and $K_n:L^2(\Gamma,|\d z|)\rightarrow L^2(\Gamma,|\d z|)$ the trace class operator with kernel
\begin{equation}\label{e:33}
	K_n(z,w):=\frac{z^{\frac{n}{2}}w^{-\frac{n}{2}}-z^{-\frac{n}{2}}w^{\frac{n}{2}}}{2\pi\im(z-w)},
\end{equation}
acting on the arc $\Gamma:=z(J)\subset\mathbb{C}$ defined in \eqref{e:32} and shown in Figure \ref{fig5}. Second, one notices that \eqref{e:33} is a particular instance of an \textit{integrable} kernel, i.e. a kernel of the form
\begin{equation*}
	K(x,y)=\frac{\sum_{j=1}^Nf_j(x)g_j(y)}{x-y},\ \ \ \ x,y\in\Gamma\ \ \ \ \ \ \ \textnormal{with}\ \ \ \ \ \ \ \sum_{j=1}^Nf_j(x)g_j(x)=0,\ \ x\in\Gamma,
\end{equation*}
for some functions $f_j,g_j\in L^{\infty}(\Gamma,|\d z|),1\leq j\leq N$. The associated integral operator $K:L^2(\Gamma,|\d z|)\rightarrow L^2(\Gamma,|\d z|)$ has many remarkable properties, most importantly $(1-K)^{-1}$ (if existent) can be computed in terms of the solution of a naturally associated Hilbert boundary value problem. This insight was formalized in the important 1990 paper \cite{IIKS} by Its, Izergin, Korepin and Slavnov, see \cite{D0} for a concise review, and reads in case of \eqref{e:33} as follows:
\begin{prob}\label{HP5} For any $n\in\mathbb{Z}_{\geq 1}$ and $0<h<2$, determine $X(z)=X(z;n,h)\in\mathbb{C}^{2\times 2}$ such that
\begin{enumerate}
	\item[(1)] $X(z)$ is analytic in $\mathbb{C}\setminus\Gamma$ and extends continuously from either side to $\Gamma\setminus\{\xi,\bar{\xi}\}$. See Figure \ref{fig5} for the jump contour $\Gamma$ and its endpoints $\xi$ and $\bar{\xi}$.
	\item[(2)] On the arc $\Gamma\subset\mathbb{C}$ the pointwise limits
	\begin{equation*}
		X_{\pm}(z):=\lim_{\substack{w\rightarrow z\\ w\in\Omega_{\pm}}}X(w),
	\end{equation*}
	satisfy the boundary condition $X_+(z)=X_-(z)G(z)$ with 
	\begin{equation*}
		z\in\Gamma:\ \ \ \ \ \ G(z)=\begin{bmatrix}0 & -z^n\\ z^{-n} & 2\end{bmatrix}\in\textnormal{GL}(2,\mathbb{C}).
	\end{equation*}
	\begin{figure}[tbh]
\includegraphics[width=0.3\textwidth]{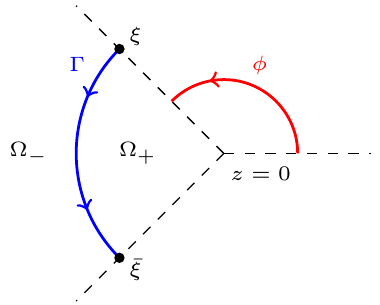}
\caption{The oriented contour $\Gamma$ in \textcolor{blue}{blue} together with the sectors $\Omega_{\pm}$. The endpoints are $\xi:=\e^{\im\phi}=z(-\Lambda), \bar{\xi}=\e^{-\im\phi}$ and the angle $\phi$ in \eqref{e:32} is shown in \textcolor{red}{red}.}
\label{fig5}
\end{figure}
	\item[(3)] In a small neighborhood of $z=\xi$ and $z=\bar{\xi}$ we impose the blow-up constraint
	\begin{equation*}
		\|X(z)\|\leq C\,\big|\ln\left|\frac{z-\xi}{z-\bar{\xi}}\right|\big|,\ \ \ \ \ \ C>0.
	\end{equation*}
	\item[(4)] Near $z=\infty$ we require the asymptotic normalization
	\begin{equation*}
		X(z)=\mathbb{I}+X_1z^{-1}+\mathcal{O}\big(z^{-2}\big),\ \ \ |z|\rightarrow\infty;\ \ \ X_k=X_k(n,h)=\big[X_k^{ij}(n,h)\big]_{i,j=1}^2.
	\end{equation*}
\end{enumerate}
\end{prob}
Problem \ref{HP5} and the emptiness formation probability \eqref{e:31} are related in the following way: as shown in \cite[Proposition $6.12$]{DIZ}, the Hilbert boundary value problem \ref{HP5} is uniquely solvable for any $n\in\mathbb{Z}_{\geq 1}, h\in(0,2)$ and its solution leads to the below recursion for $P_n$,
\begin{equation}\label{e:34}
	\frac{P_{n+1}}{P_n}=X^{11}(0;n,h),
\end{equation}
in terms of the $(11)$-entry of the solution to Problem \ref{HP5}. This identity is more complicated than the representation formul\ae\,\eqref{e:26}, \eqref{e:29} and its derivation does not use a Lax pair as in the NLS or Painlev\'e-II example. Nevertheless, \eqref{e:34} enables us to use Riemann-Hilbert techniques in the analysis of the spin chain model, i.e. Plemelj's basic idea occurs once more.\bigskip

In order to obtain \eqref{e:34}, we first use the operator factorization \eqref{opfac} and Sylvester's determinant identity,
\begin{equation*}
	P_n=\det\big(1-U_n\upharpoonright_{L^2(J,\d x)}\big)=\det\big(1-K_n\upharpoonright_{L^2(\Gamma,|\d z|)}\big).
\end{equation*}
Next we realize that
\begin{equation*}
	K_{n+1}(z,w)=K_n(z,w)+\frac{1}{2\pi\im}z^{\frac{n}{2}}w^{-\frac{n}{2}-1}=:K_n(z,w)+(\alpha_n\otimes\beta_n)(z,w),
\end{equation*}
where $\alpha_n\otimes\beta_n$ is the finite rank integral operator on $L^2(\Gamma,|\d z|)$ with kernel $(\alpha_n\otimes\beta_n)(z,w)=\alpha_n(z)\beta_n(w)$ in terms of $\alpha_n(z):=\frac{1}{2\pi\im}z^{\frac{n}{2}}$ and $\beta_n(w):=w^{-\frac{n}{2}-1}$. But \cite{IIKS} showed that Problem \ref{HP5} is solvable if and only if $(1-K_n)^{-1}$ exists, so we can further simplify \eqref{e:31},
\begin{align*}
	P_{n+1}=\det\big(1-K_n-\alpha_n\otimes\beta_n\upharpoonright_{L^2(\Gamma,|\d z|)}\big)=P_n\,\det\big(1-(1-K_n)^{-1}(\alpha_n\otimes\beta_n)\upharpoonright_{L^2(\Gamma,|\d z|)}\big),
\end{align*}
and evaluate the remaining Fredholm determinant with the general theory of finite rank operators,
\begin{equation}\label{e:35}
	\frac{P_{n+1}}{P_n}=1-\int_{\Gamma}\big((1-K_n)^{-1}\alpha_n\big)(z)z^{-\frac{n}{2}}\,\frac{\d z}{z}.
\end{equation}
At this point we use the general fact, cf. \cite{IIKS} or \eqref{e:45} below, that the solution of Problem \ref{HP5} can be expressed in integral form, similar to Plemelj's quasi-regular equation \eqref{e:7},
\begin{equation*}
	X(z)=\mathbb{I}-\int_{\Gamma}F(w)g^{\top}(w)\frac{\d w}{w-z},\ \ z\notin\Gamma;\ \ \ \ \ F_i(z):=\big((1-K_n)^{-1}f_i\big)(z),\ \ z\in\Gamma,
\end{equation*}
with $F(z):=(F_1(z),F_2(z))^{\top},g(z):=(g_1(z),g_2(z))^{\top}$ and
\begin{equation*}
	f_1(z):=\frac{1}{2\pi\im}z^{\frac{n}{2}},\,\,f_2(z):=-\frac{1}{2\pi\im}z^{-\frac{n}{2}},\,\,g_1(z):=z^{-\frac{n}{2}},\,\,g_2(z):=z^{\frac{n}{2}},\ \ \ \ \ \textnormal{arg}\,z\in(\phi,2\pi-\phi).
\end{equation*}
Setting $z=0$ in the integral formula for $X(z)$ and comparing with \eqref{e:35} yields
\begin{equation*}
	\frac{P_{n+1}}{P_n}=X^{11}(0;n,h),
\end{equation*}
i.e. the representation \eqref{e:34}. In summary, the Hilbert boundary value problem \ref{HP5} linearizes to a certain extent (note that \eqref{e:34} is a formula for the discrete logarithmic derivative of $P_n$ and not $P_n$ itself) the highly transcendental Fredholm determinant \eqref{e:31}. We will display the usefulness of \eqref{e:34} after the next subsection.
\subsection{Invariant random matrix models} Our fifth and last example in this section concerns the statistical behavior of eigenvalues of a random matrix drawn from a particular ensemble. We will focus on an especially well-studied ensemble, the so called \textit{unitary ensemble} $\mathcal{M}(n)$ of $n\times n$ Hermitian matrices equipped with the probability measure
\begin{equation}\label{e:36}
	P(M)\,\d M= \frac{1}{Z_{n,N}}\,\e^{-N\textnormal{tr}\,V(M)}\,\d M,\ \ \ \ \ \ Z_{n,N}:=\int_{\mathcal{M}(n)}\e^{-N\textnormal{tr}\,V(M)}\,\d M<\infty.
\end{equation} 
Here, $\d M$ is the Haar measure on $\mathcal{M}(n)\simeq\mathbb{R}^{n^2}$, $N\in\mathbb{Z}_{\geq 1}$ a scaling parameter and $V:\mathbb{R}\rightarrow\mathbb{R}$ assumed to be real analytic satisfying the growth condition
\begin{equation*}
	\frac{V(x)}{\ln(x^2+1)}\rightarrow\infty\ \ \ \ \textnormal{as}\ \ \ |x|\rightarrow\infty,
\end{equation*}
in order to ensure that \eqref{e:36} is a bona fide model. Measures of the form \eqref{e:36} appeared seemingly first in the work of Hurwitz in 1897 \cite{Hu}, they appeared in a paper by Wishart in 1928 \cite{Wish} and in particular in Wigner's work in the 1950s \cite{Wig} who used them to model the statistical properties of highly excited energy levels of complex nuclei, see \cite{DF} for an excellent historic review of the subject. A classical fact of the setup \eqref{e:36} is that the location of the eigenvalues of a matrix drawn from the unitary ensemble form a determinantal point process. This detail allows us to compute several of the key statistical properties of the random matrix model from the underlying Christoffel-Darboux kernel
\begin{equation}\label{e:37}
	K_{n,N}(x,y)=\e^{-\frac{N}{2}V(x)}\e^{-\frac{N}{2}V(y)}\sum_{k=0}^{N-1}\frac{1}{h_k^2}\,p_k(x)p_k(y),\ \ \ x,y\in\mathbb{R},
\end{equation}
which is defined in terms of the sequence $\{p_n(x)\}_{n=0}^{\infty}\subset\mathbb{C}[x]$ of monic orthogonal polynomials with respect to the measure $\e^{-NV(x)}\,\d x$ supported on $\mathbb{R}$, i.e.
\begin{equation*}
	\int_{-\infty}^{\infty}p_j(x)p_k(x)\e^{-NV(x)}\,\d x=h_k^2\delta_{jk},\ \ \ \ \ \ p_j(x)=x^j+\mathcal{O}(x^{j-1}).
\end{equation*}
Indeed, focusing on one important statistical property of the model \eqref{e:36}, the \textit{gap probability},
\begin{equation*}
	E_n(s):=\textnormal{Prob}\big\{M\in\mathcal{M}(n)\ \textnormal{has no eigenvalues in the interval}\,(-s,s)\subset\mathbb{R},s>0\big\},
\end{equation*}
admits the determinantal representation
\begin{equation}\label{gap}
	E_n(s)=\det\big(1-K_{n,N}\upharpoonright_{L^2((-s,s),\d x)}\big),
\end{equation}
where $K_{n,N}$ is the finite rank operator with kernel \eqref{e:37}. Formula \eqref{gap} goes back to the seminal works of Dyson, Gaudin and Mehta in the 1970s and it is the analogue of the emptiness formation probability formula \eqref{e:31} in the spin chain model. Thus, in order to compute the likelihood of large eigenvalue gaps from \eqref{gap}, we need to carefully analyze the asymptotic behavior of the kernel \eqref{e:37}. However, in all but a few cases, the orthogonal polynomials that built the kernel are not known explicitly, so how can we effectively use \eqref{e:37}? Well, as it happens, the polynomials $\{p_n(x)\}_{n=0}^{\infty}$ and the kernel $K_{n,N}(x,y)$ itself can be characterized through a Hilbert boundary value problem. This remarkable fact was discovered by Fokas, Its and Kitaev \cite{FIK1,FIK2} in the early 1990s and the details are as follows:
\begin{prob}\label{HP6} For any $n\in\mathbb{Z}_{\geq 0}$ and $N\in\mathbb{Z}_{\geq 1}$, determine $X(z)=X(z;n,N)\in\mathbb{C}^{2\times 2}$ such that
\begin{enumerate}
	\item[(1)] $X(z)$ is analytic in $\mathbb{C}\setminus\mathbb{R}$ and extends continuously from either side to the real axis.
	\item[(2)] On the real axis, the pointwise limits
	\begin{equation*}
		X_{\pm}(z):=\lim_{\epsilon\downarrow 0}X(z\pm\im\epsilon)
	\end{equation*}
	satisfy the boundary condition $X_+(z)=X_-(z)G(z)$ with
	\begin{equation}\label{e:38}
		z\in\mathbb{R}:\ \ \ \ \ \ G(z)=\begin{bmatrix}1 & \e^{-NV(z)}\\ 0 & 1\end{bmatrix}\in\textnormal{GL}(2,\mathbb{C}).
	\end{equation}
	\item[(3)] Near $z=\infty$ we enforce the asymptotic behavior
	\begin{equation*}
		X(z)=\Big(\mathbb{I}+X_1z^{-1}+\mathcal{O}\big(z^{-2}\big)\Big)z^{n\sigma_3}, \ \ |z|\rightarrow\infty.
	\end{equation*}
\end{enumerate}
\end{prob}
Fokas, Its and Kitaev showed that this problem is uniquely solvable for a given $n\in\mathbb{Z}_{\geq 0}$ if and only if the $n$th monic orthogonal polynomial $p_n(x)$ exists. Moreover, in case when its solvability is ensured, then we have for any $x,y\in\mathbb{R}$,
\begin{equation}\label{e:400}
	p_n(x)=X^{11}(x;n,N)\ \ \ \ \textnormal{and}\ \ \ \ K_{n,N}(x,y)=\e^{-\frac{N}{2}V(x)}\left[\frac{(X_+(y))^{-1}X_+(x)}{2\pi\im(x-y)}\right]^{21}\e^{-\frac{N}{2}V(y)},
\end{equation}
in terms of the $(11)$-entry of $X(z)$ for the polynomial and the $(21)$-entry of the matrix product for the kernel. Here are the details: assuming the solvability of Problem \ref{HP6}, the jump constraint \eqref{e:38} states that the first column of $X(z)$ is an entire function,
\begin{align}
	X_+^{11}(z)=X_-^{11}(z),\ \ \ X_+^{21}(z)=X_-^{21}(z),\ \ \ \ X_+^{12}(z)=&\,\,X_-^{12}(z)+X_-^{11}(z)\e^{-NV(z)}\label{e:39}\\
	X_+^{22}(z)=&\,\,X_-^{22}(z)+X_-^{21}(z)\e^{-NV(z)},\ \ z\in\mathbb{R}.\nonumber
\end{align}
But the asymptotic normalization at $z=\infty$ enforces at the same time that
\begin{equation*}
	X^{11}(z)=z^n+X_1^{11}z^{n-1}+\mathcal{O}\big(z^{n-2}\big),\ \ \ \ \ X^{21}(z)=X_1^{21}z^{n-1}+\mathcal{O}\big(z^{n-2}\big),
\end{equation*}
and
\begin{equation}\label{e:40}
	X^{12}(z)=X_1^{12}z^{-n-1}+\mathcal{O}\big(z^{-n-2}\big),\ \ \ \ \ \ X^{22}(z)=z^{-n}+X_1^{22}z^{-n-1}+\mathcal{O}\big(z^{-n-2}\big).
\end{equation}
Hence, $X^{11}(z)$ must be a monic polynomial of degree $n$ and $X^{21}(z)$ a polynomial of degree at most $n-1$ by Liouville's theorem. Returning with this information to \eqref{e:39} and using the Plemelj-Sokhotski formula we then find for $z\in\mathbb{C}\setminus\mathbb{R}$,
\begin{equation}\label{e:4000}
	X^{12}(z)=\frac{1}{2\pi\im}\int_{-\infty}^{\infty}X^{11}(w)\e^{-NV(w)}\frac{\d w}{w-z}\ \ \ \textnormal{and}\ \ \ X^{22}(z)=\frac{1}{2\pi\im}\int_{-\infty}^{\infty}X^{21}(w)\e^{-NV(w)}\frac{\d w}{w-z}.
\end{equation}
But geometric progression yields for $z\notin\mathbb{R}$,
\begin{equation*}
	X^{12}(z)=-\frac{1}{2\pi\im}\sum_{k=0}^{n-1}\frac{1}{z^{k+1}}\int_{-\infty}^{\infty}X^{11}(w)w^k\e^{-NV(w)}\,\d w+\frac{z^{-n}}{2\pi\im}\int_{-\infty}^{\infty}X^{11}(w)\e^{-NV(w)}\frac{w^n}{w-z}\,\d w
\end{equation*}
and thus after comparison with \eqref{e:40}
\begin{equation*}
	\int_{-\infty}^{\infty}X^{11}(w)w^k\e^{-NV(w)}\,\d w=0,\ \ \ k=0,\ldots,n-1.
\end{equation*}
Since $X^{11}(z)$ has already been found to be a monic polynomial, it must therefore be the $n$th monic orthogonal polynomial $p_n(x)$ for the measure $\e^{-NV(x)}\,\d x$ on $\mathbb{R}$, so the first identity in \eqref{e:400} follows. The kernel identity is slightly more complicated and involves $X^{21}(z)$: From the integral formula of $X^{22}(z)$ in \eqref{e:4000}, by geometric progression for $z\in\mathbb{C}\setminus\mathbb{R}$,
\begin{equation*}
	X^{22}(z)=-\frac{1}{2\pi\im}\sum_{k=0}^{n-2}\frac{1}{z^{k+1}}\int_{-\infty}^{\infty}X^{21}(w)w^k\e^{-NV(w)}\,\d w+\frac{z^{-n+1}}{2\pi\im}\int_{-\infty}^{\infty}X^{21}(w)\e^{-NV(w)}\frac{w^{n-1}}{w-z}\,\d w,
\end{equation*}
so after comparison with \eqref{e:40},
\begin{equation*}
	\int_{-\infty}^{\infty}X^{21}(w)w^k\e^{-NV(w)}\,\d w=0,\ \ \ k=0,\ldots,n-2;\ \ \ \ \frac{\im}{2\pi}\int_{-\infty}^{\infty}X^{21}(w)w^{n-1}\e^{-NV(w)}\,\d w=1.
\end{equation*}
This shows that $X^{21}(z)$ must be proportional to $p_{n-1}(z)$, i.e. we have for some $\gamma_{n-1}\in\mathbb{C}\setminus\{0\}$,
\begin{equation*}
	X^{21}(z)=\gamma_{n-1}p_{n-1}(z),\ \ \ z\in\mathbb{C},
\end{equation*}
and thus by orthogonality and the very definition of the sequence $\{p_n(x)\}_{n=0}^{\infty}$,
\begin{equation*}
	-2\pi\im=\int_{-\infty}^{\infty}X^{21}(w)w^{n-1}\e^{-NV(w)}\,\d w=\gamma_{n-1}\int_{-\infty}^{\infty}p_{n-1}^2(w)\e^{-NV(w)}\,\d w=\gamma_{n-1}h_{n-1}^2.
\end{equation*}
Next we note that any solution of Problem \ref{HP6} must be unimodular in the entire plane, i.e. $\det X(z)\equiv 1$ for all $z\in\mathbb{C}$. Indeed, the function $\det X(z)$ is analytic for $z\in\mathbb{C}\setminus\mathbb{R}$ and continuous on the closed upper and lower half-planes with
\begin{equation*}
	\det X_+(z)=\det X_-(z),\ \ \ z\in\mathbb{R}
\end{equation*}
since $\det G(z)\equiv 1$. Hence, $\det X(z)$ is entire and with $\det X(z)\rightarrow 1$ from Problem \ref{HP6}, condition (3) we indeed arrive at $\det X(z)\equiv 1$. This in mind, we now use the Christoffel-Darboux identity in \eqref{e:37},
\begin{equation}\label{e:41}
	K_{n,N}(x,y)=\e^{-\frac{N}{2}V(x)}\e^{-\frac{N}{2}V(y)}\frac{1}{h_{n-1}^2}\frac{p_n(x)p_{n-1}(y)-p_{n-1}(x)p_n(y)}{x-y}
\end{equation}
and our formul\ae\,$X^{11}(z)=p_n(z),X^{21}(z)=\gamma_{n-1}p_{n-1}(z)$ with $\gamma_{n-1}h_{n-1}^2=-2\pi\im$. A simple matrix multiplication and comparison with \eqref{e:41} leads to the second identity in \eqref{e:400}. Once more we summarize our discussion: the Hilbert boundary value problem \eqref{HP6} allows us to access both, orthogonal polynomials and their Christoffel-Darboux kernel. In turn, the Riemann-Hilbert characterization allows us to rigorously analyze relevant statistical quantities in the random matrix model \eqref{e:36}, provided we can efficiently derive a large $n$ asymptotic expansion for $X(z;n,N)$. The development of such an efficient scheme has been one of the many highlights over the past 30 years in the Riemann-Hilbert toolbox and we shall return to this important advancement after the next section.
\section{Hilbert boundary value problems in $L^2$-spaces}\label{weakRHP}
The boundary values problems encountered in the last section are different from Plemelj's initial problem \ref{HP1} in that they include jump contours extending to infinity, jump contours with open ends and jump contours that self-intersect. It is therefore not obvious how Plemelj's solvability proof can be lifted to the problems of Section \ref{ex:sec}. For this reason we now give a short overview of the relevant theory in which Plemelj's analysis of \eqref{e:7} in a space of continuous functions is extended to spaces of integrable functions - for a more thorough discussion we direct the interested reader to the articles \cite{Z,DZ0} by Zhou and Deift, Zhou. Here are the basic two assumptions of this section:

\begin{assu}\label{assu:2} Let $\Sigma$ be a contour consisting of a finite union of smooth oriented curves in $\mathbb{CP}^1$ with finitely many self-intersections. 
\end{assu}
\begin{assu}\label{assu:3} Let $G:\Sigma\rightarrow\textnormal{GL}(p,\mathbb{C})$ be a map such that $G,G^{-1}\in L^{\infty}(\Sigma,|\d z|), G-\mathbb{I}\in L^2(\Sigma,|\d z|)$ and $G$ has zero winding.
\end{assu}
Now consider the following extended $L^2$-Hilbert boundary value problem.
\begin{prob}[The $L^2$-problem]\label{HP7} Given a pair $(\Sigma,G)$ subject to Assumptions \ref{assu:2} and \ref{assu:3}, determine two $\mathbb{C}^{p\times p}$-valued functions $X_{\pm}\in L^2(\Sigma,|\d z|)$ such that
\begin{enumerate}
	\item[(1)] There exists a $\mathbb{C}^{p\times p}$-valued function $H\in L^2(\Sigma,|\d z|)$ so that
	\begin{equation*}
		X_{\pm}(z)-\mathbb{I}=\frac{1}{2\pi\im}\int_{\Sigma}H(w)\frac{\d w}{w-z_{\pm}},
	\end{equation*}
	where
	\begin{equation*}
		\int_{\Sigma}H(w)\frac{\d w}{w-z_{\pm}}:=\lim_{\lambda\rightarrow z_{\pm}}\int_{\Sigma}H(w)\frac{\d w}{w-\lambda}
	\end{equation*}
	denotes the a.e. existent non-tangential limits from the $\pm$ side of $\Sigma$. 
	\item[(2)] We have $X_+(z)=X_-(z)G(z)$ a.e. on $\Sigma$.
\end{enumerate}
\end{prob}
Note how the abstract Problem \ref{HP7} captures all our examples in Section \ref{ex:sec} and Plemelj's problem \ref{HP1}: by condition (1), and the Plemelj-Sokhotski formula in $L^2(\Sigma,|\d z|)$, see \cite[Chapter II]{St}, the functions $X_{\pm}$ are the $L^2$-boundary values of the \textit{extension}
\begin{equation}\label{e:42}
	X(z):=\mathbb{I}+\frac{1}{2\pi\im}\int_{\Sigma}H(w)\frac{\d w}{w-z},\ \ \ z\in\mathbb{C}\setminus\Gamma.
\end{equation}
But \eqref{e:42} is analytic in $\mathbb{C}\setminus\Gamma$, has non-tangential limits $X_{\pm}(z)\in L^2(\Sigma,|\d z|)$ a.e. on $\Gamma$, those satisfy $X_+(z)=X_-(z)G(z)$ a.e. on $\Gamma$ and we have $X(z)\rightarrow\mathbb{I}$ as $z\rightarrow\infty$. 
\begin{rem}\label{L2bound} If $f\in L^q(\Sigma,|\d z|)$ for some $1\leq q<\infty$, then the Cauchy operators
\begin{equation*}
	(C_{\pm}f)(z):=\int_{\Sigma}f(w)\frac{\d w}{w-z_{\pm}},\ \ \ z\in\Sigma
\end{equation*}
exist a.e. for $z\in\Sigma$ and if $1<q<\infty$,
\begin{equation}\label{crucialbound}
	\|C_{\pm}f\|_{L^q(\Sigma,|\d z|)}\leq c\|f\|_{L^q(\Sigma,|\d z|)},\ \ \ c=c(q)>0.
\end{equation}
This says that the Cauchy operators are, in particular, bounded linear operators on $L^2(\Sigma,|\d z|)$.
\end{rem}
Next, we derive the analogue of Plemelj's singular integral equation \eqref{e:7} for Problem \ref{HP7}: if $X_{\pm}(z)$ solve Problem \ref{HP7}, then by the Plemelj-Sokhotski formula and condition (2),
\begin{equation}\label{e:43}
	H(z)=X_+(z)-X_-(z)=X_-(z)\big(G(z)-\mathbb{I}\big)\ \ \ \ \ \textnormal{a.e. on}\,\,\Sigma.
\end{equation}
On the other hand
\begin{equation*}
	K(z):=\frac{1}{2\pi\im}\int_{\Sigma}X_-(w)\big(G(w)-\mathbb{I}\big)\frac{\d w}{w-z},\ \ z\in\mathbb{C}\setminus\Sigma
\end{equation*}
satisfies, again by Plemelj-Sokhotski,
\begin{equation*}
	K_{\pm}(z)=-X_-(z)+\frac{1}{2\pi\im}\int_{\Sigma}X_-(w)\big(G(w)-\mathbb{I}\big)\frac{\d w}{w-z_-}\pm X_{\pm}(z)\ \ \ \ \ \textnormal{a.e. on}\,\,\Sigma,
\end{equation*}
so by property (1) and \eqref{e:43}, $\rho(z):=X_-(z)$ solves the singular integral equation
\begin{equation}\label{e:44}
	\rho(z)=\mathbb{I}+\frac{1}{2\pi\im}\int_{\Sigma}\rho(w)\big(G(w)-\mathbb{I}\big)\frac{\d w}{w-z_-}\ \ \ \ \ \ \textnormal{a.e. on}\,\,\Sigma.
\end{equation}
Moreover, starting from \eqref{e:44} we can construct a solution of Problem \ref{HP7}:
\begin{prop}\label{van} Every solution pair $X_{\pm}\in L^2(\Sigma,|\d z|)$ of Problem \ref{HP7} leads to a solution of the singular integral equation
\begin{equation*}
	X_-(z)=\mathbb{I}+\frac{1}{2\pi\im}\int_{\Sigma}X_-(w)\big(G(w)-\mathbb{I}\big)\frac{\d w}{w-z_-}.
\end{equation*}
Conversely, if $\rho(z)$ solves \eqref{e:44} with $\rho-\mathbb{I}\in L^2(\Sigma,|\d z|)$, then the non-tangential limits $X_{\pm}(z)$ of
\begin{equation}\label{e:45}
	X(z):=\mathbb{I}+\frac{1}{2\pi\im}\int_{\Sigma}\rho(w)\big(G(w)-\mathbb{I}\big)\frac{\d w}{w-z},\ \ \ z\in\mathbb{C}\setminus\Sigma
\end{equation}
solve Problem \ref{HP7}.
\end{prop}
\begin{proof} By Assumption \ref{assu:2}, $\rho(G-\mathbb{I})\in L^2(\Sigma,|\d z|)$ if $\rho-\mathbb{I}\in L^2(\Sigma,|\d z|)$, thus the non-tangential limits $X_{\pm}(z)$ of \eqref{e:45} are easily seen to satisfy condition (1) in Problem \ref{HP7}. Moreover, by Plemelj-Sokhotski, we find from \eqref{e:44} that $\rho(z)=Y_-(z)$ a.e. on $\Sigma$ and so, also by Plemelj-Sokhotski, from \eqref{e:45},
\begin{equation*}
	X_+(z)-X_-(z)=\rho(z)\big(G(z)-\mathbb{I}\big)=X_-(z)\big(G(z)-\mathbb{I}\big)\ \ \ \textnormal{a.e. on}\ \ \Sigma
\end{equation*}
which yields condition (2) in Problem \ref{HP7}. This concludes the proof.
\end{proof}
Having established in \eqref{e:44} the analogue of Plemelj's \eqref{e:7}, how do we solve this equation? In short, we also rely on a Fredholm Alternative argument, similar to Lemma \ref{Plem1} and \ref{Plem2}: let 
\begin{equation}\label{e:46}
	G(z)=\big(\mathbb{I}-V^-(z)\big)^{-1}\big(\mathbb{I}+V^+(z)\big)\ \ \textnormal{a.e. on}\ \Sigma,
\end{equation}
denote a pointwise factorization of $G$ with functions $(\mathbb{I}\pm V^{\pm})^{\pm  1}-\mathbb{I}\in L^2(\Sigma,|\d z|)\cap L^{\infty}(\Sigma,|\d z|)$. This factorization allows us to rewrite the central integral equation \eqref{e:44} in the compact form
\begin{equation*}
	\textnormal{a.e. on}\ \ \Sigma:\ \ \ \ \rho(z)=\mathbb{I}+\big(C_w\rho\big)(z),\ \ \ \ \ (C_wf)(z):=\big(C_+(fV^-)\big)(z)+\big(C_-(fV^+)\big)(z)
\end{equation*}
with the Cauchy operators $C_{\pm}$ of Remark \ref{L2bound}. In turn we have the following central result which states that the operator $1-C_w$ is injective if and only if a certain homogeneous version of Problem \ref{HP7} is only trivially solvable.
\begin{prop}[{\cite[Proposition $3.3$]{Z}}]\label{vanprop} $\rho\in L^2(\Sigma,|\d z|)$ solves the homogeneous equation
\begin{equation*}
	\rho(z)=\big(C_w\rho\big)(z)\ \ \ \textnormal{a.e. on}\ \ \Sigma
\end{equation*}
if and only if $X_{\pm}(z):=\rho(z)(\mathbb{I}\pm V^{\pm}(z)),z\in\Sigma$ solves the homogeneous version of Problem \ref{HP7}, i.e. the problem of determining $X_{\pm}\in L^2(\Sigma,|\d z|)$ so that\\

\noindent\,\,$(1)$ There exists a $\mathbb{C}^{p\times p}$-valued function $H\in L^2(\Sigma,|\d z|)$ such that
	\begin{equation*}
		X_{\pm}(z)=\frac{1}{2\pi\im}\int_{\Sigma}H(w)\frac{\d w}{w-z_{\pm}}.
	\end{equation*}
\noindent\,\,$(2)$ We have $X_+(z)=X_-(z)G(z)$ a.e. on $\Sigma$.
\end{prop}
In order to state the Fredholm Alternative needed in the solvability analysis of Problem \ref{HP7} we note that all our jump matrices encountered in this note possess much more regularity than what is assumed in Assumption \ref{assu:3}, in fact the matrices are piecewise smooth, admit local analytical continuations and at possible intersection points satisfy a cyclic constraint in the style of \eqref{e:4} - for instance in Problem \ref{HP4} the total monodromy at $z=0$ is trivial\footnote{The importance of such cyclic constraints, viewed as formal power series identities, cannot be overemphasized, see \cite[$\S$25]{BDT} for further details and context.}. These properties (see \cite[page $103-104$]{FIKN} for a formalization) together with Assumptions \ref{assu:2} and \ref{assu:3} ensure that the following far reaching generalizations of Lemma \ref{Plem1} and \ref{Plem2} hold true, at least for all matrix-valued Hilbert boundary value problems in this article.
\begin{prop}[{\cite[Proposition $4.3$]{Z}}]\label{anaFred} If the factors $V^{\pm}(z)=V^{\pm}(z;\zeta)$ in \eqref{e:46} depend on a parameter $\zeta$ analytically, then either $(1-C_w)^{-1}$ is meromorphic in $\zeta$ or $1-C_w$ is invertible for no $\zeta$.
\end{prop}
Proposition \ref{anaFred} is an analytic Fredholm theorem and together with Proposition \ref{van} yields two central results, given that $1-C_w$ is a Fredholm operator of index zero, compare Remark \ref{tech}.
\begin{theo}[{\cite[Proposition $4.2$]{Z}}] If the factors $V^{\pm}(z)=V^{\pm}(z;\zeta)$ in \eqref{e:46} depend on a parameter $\zeta$ analytically, then either $X(z;\zeta)$ in \eqref{e:45} is meromorphic in $\zeta$, or the associated $\zeta$-modified Problem \ref{HP7} is not solvable for any $\zeta$.
\end{theo}
\begin{theo}[Zhou's vanishing lemma]\label{Zhou} The $L^2$-Hilbert boundary value problem \ref{HP7} is solvable if and only if the corresponding homogeneous version of the problem, see Proposition \ref{vanprop}, is only trivially solvable.
\end{theo}
In the usual, smooth, setting of a Hilbert boundary value problem, the homogeneous version of it has the same analyticity and jump behavior, but the asymptotic normalization at $z=\infty$ is replaced by
\begin{equation*}
	X(z)=\mathcal{O}\big(z^{-1}\big),\ \ \ \ \ |z|\rightarrow\infty.
\end{equation*}
This requirement is completely analogous to the behaviors which Plemelj enforced in his accompanying and associated problems, compare for instance Problem \ref{HP2}.\smallskip

Before moving to the next section, we will give one application of Theorem \ref{Zhou} to Problem \ref{HP3} and en route verify that the defocusing NLS with Cauchy data $y_0(x)\in\mathcal{S}(\mathbb{R})$ is solvable: Suppose $X(z)$ solves Problem \ref{HP3} but with condition (3) replaced by
\begin{equation}\label{e:47}
	X(z)=\mathcal{O}\big(z^{-1}\big),\ \ \ |z|\rightarrow\infty.
\end{equation}
Define $Y(z):=X(z)X^{\dagger}(\bar{z})$ for $z\in\mathbb{C}\setminus\mathbb{R}$ with $X^{\dagger}$ denoting the conjugate transpose matrix of $X$. By property (1) in Problem \ref{HP3} we find that $Y$ is analytic in the upper $z$-plane, continuous down to the real line and by \eqref{e:47} decays of $\mathcal{O}(z^{-2})$ as $z\rightarrow\infty$ in the upper $z$-plane. Hence, by Cauchy's theorem, we have $\int_{\mathbb{R}}Y_+(z)\,\d z=0$, and, adding to this its conjugate transpose, find in turn
\begin{align*}
	0=&\,\,\int_{-\infty}^{\infty}\Big(X_+(z)X_-^{\dagger}(z)+X_-(z)X_+^{\dagger}(z)\Big)\,\d z=\int_{-\infty}^{\infty}X_-(z)\Big(G(z)+G^{\dagger}(z)\Big)X_-^{\dagger}(z)\,\d z\\
	=&\,\,2\int_{-\infty}^{\infty}X_-(z)\begin{bmatrix}1-|r(z)|^2 & 0\\ 0 & 1\end{bmatrix}X_-^{\dagger}(z)\,\d z.
\end{align*}
Now read off the diagonal entries in the last matrix equation and conclude that
\begin{equation*}
	0=\int_{-\infty}^{\infty}\Big(\big|X_-^{k1}(z)\big|^2\big(1-|r(z)|^2\big)+\big|X_-^{k2}(z)\big|^2\Big)\,\d z,\ \ \ \ \ k=1,2,
\end{equation*}
so by the fact that $r\in\mathcal{S}(\mathbb{R})\cap\{r:\,\sup_{z\in\mathbb{R}}|r(z)|<1\}$, cf. \cite{BC}, we have $X_-(z)\equiv 0$ on $\mathbb{R}$ by continuity of $X_-(z)$. In turn, from condition (2) in Problem \ref{HP3} also $X_+(z)\equiv 0$ on $\mathbb{R}$ and thus all together $X(z)$ is analytic in the upper $z$-plane, continuous in the closed upper $z$-plane and
\begin{equation*}
	\|X(z)\|\leq C,\ \ \ \Im z>0,\ \ \ \ \ \ \ \sup_{z\in\mathbb{R}}\|X_+(z)\|=0.
\end{equation*}
Using Carlson's theorem \cite[Theorem $5.1.2$ and Corollary $5.1.3$]{Sim1}, this allows us to conclude that $X(z)\equiv 0$ for $\Im z\geq 0$ and with the same logic for $\Im z\leq 0$ all together $X(z)\equiv 0$. By Zhou's vanishing lemma \ref{Zhou} we then deduce that Problem \ref{HP3} is solvable for any $x\in\mathbb{R},t>0$ and so the solution of the initial value problem \eqref{e:25} with $y_0(x)\in\mathcal{S}(\mathbb{R})$ exists.\bigskip

A similar application of Theorem \ref{Zhou} to Problem \ref{HP4} does not yield the desired solvability result for general choices of $s\in\mathcal{S}$. Indeed, for some choices $s\in\mathcal{S}$, the solution $X(z)=X(z;x,s)$ of Problem \ref{HP4} ceases to exist for certain $x\in\mathbb{R}$ and this is because the corresponding Painlev\'e-II transcendent $u=u(x|s)$ has poles on the real axis, compare Subsection \ref{PIIconn}. On the other hand, the solvability of Problems \ref{HP5} and \ref{HP6} is always guaranteed and can be established without using Theorem \ref{Zhou}.

\section{A Hilbert boundary value problem as Swiss army knife}\label{cool}
By now we have seen several Hilbert boundary value problems in OPSFA related problems and we have acquainted ourselves to a certain extent with their abstract solvability theory in Sections \ref{Plejcon} and \ref{weakRHP}. Still, besides providing a positive solution to RHP \ref{RHP:1} in interpretation (2), a reader unaccustomed to Hilbert boundary value problems might not yet grasp their usefulness: after all, although they seem to linearize nonlinear dynamical systems such as \eqref{e:25} and \eqref{e:28} or rephrase certain Fredholm determinants and orthogonal polynomials, those underlying Problems (i.e. Problems \ref{HP3}, \ref{HP4}, \ref{HP5} and \ref{HP6}) are in all, but a few trivial, cases \textit{not explicitly solvable}. So what is the whole point of a Hilbert boundary value problem? Well, the last 30 years have clearly shown that Hilbert boundary value problems really possess all the fundamental properties of a contour integral formula known, and appreciated, in classical special function theory. To the point, Hilbert boundary value problems underlie a large class of integrable models and as such allow us to\bigskip

$\diamond$ systematically derive dynamical systems (continuous ones, discrete ones or hybrids) for the quantities under consideration. This feature was first used by Plemelj in his work \cite{P0} on Hilbert's 21st problem. We have seen the same approach in action while discussing \eqref{e:25} and \eqref{e:28} and will further discuss it in the upcoming subsections. In these modern applications one must mention the pioneering roles of Shabat and Zakharov in the derivation of nonlinear partial differential equations from Hilbert boundary value problems, the role of Krein \cite{Kre1} in obtaining differential equations from integral equations and last, but not least, the role of Its  \cite[page $377,417$]{KBI} and his emphasis on studying correlation functions in statistical mechanics and quantum field theories via Hilbert boundary value problems and their associated differential equations.\smallskip

$\diamond$ analyze the models asymptotically in their thermodynamical limits. It is this feature which puts Hilbert boundary value problems on the same ground as their linear counterparts, i.e. contour integral formul\ae, and which has turned the \textit{Riemann-Hilbert approach} to OPSFA problems into an unprecedented success story over the past 30 years. However the relevant asymptotic techniques did not grow over night: early progress - using Gelfand-Levitan type integral equations - on the asymptotic analysis of nonlinear wave equations that are solvable by the inverse scattering method was achieved in the early 1970s, namely by Shabat \cite{Sh}, Manakov \cite{Man} and Ablowitz, Newell \cite{AN}. These works were not always fully rigorous but the gaps were covered in the 1980s by Buslaev, Sukhanov, by Novokshenov and by Novokshenov, Sukhanov, see \cite[page $182$]{DIZ0} for references. The first step of using directly a Hilbert boundary value problem for asymptotic questions seems to have originated in the works of Manakov \cite{Man} and Its \cite{Its} on the NLS equation. Their techniques were subsequently extended by Its, Petrov to the sine-Gordon equation, by Bikbaev, Its to the Landau-Lifshitz equation and by Its, Novokshenov to the modified KdV equation, see again \cite[page $182$]{DIZ0} for references to the relevant papers published in the mid 1980s. However, aside from the NLS case, these works ``only" managed to asymptotically localize the initial Hilbert boundary value problems around certain special points. These points are the analogues of stationary phase points in the classical steepest descent method, but the model problems near them could not be explicitly solved, unfortunately. As it turned out, those local Hilbert boundary value problems are precisely the ones one faces in the isomonodromy deformation approach to the Painlev\'e transcendents as pioneered by Its, Novokshenov, Kapaev and Kitaev, cf. \cite{IN}. Thus, up to the early 1990s, one was able to use Hilbert boundary value problems in the asymptotic analysis of nonlinear wave and Painlev\'e equations, however one needed certain a priori information about the solution's behavior. This pitfall was then bypassed in the groundbreaking method of Deift and Zhou \cite{DZ} in 1993, the \textit{nonlinear steepest descent method}. This method, in the 1997 extended version of Deift, Venakides, Zhou \cite{DVZ}, has become the standard tool in the asymptotic analysis of Hilbert boundary value problems and all upcoming asymptotic results in this section have been originally derived from it or re-derived with it.\bigskip

We will now summarize the technical essence of the Deift-Zhou nonlinear steepest descent method: in complete analogy to the classical steepest descent method, the nonlinear steepest descent method exploits the analytic and asymptotic properties of a Hilbert boundary value problem's jump matrix $G(z)$. For instance, in case of Problem \ref{HP3} the jump matrix in \eqref{e:266} is highly oscillatory when $x$ or $t$ tend to infinity. Thus one would like to transform these oscillations to exponentially small contributions through the use of a contour deformation argument. And indeed, noticing that \eqref{e:266} admits the factorizations
\begin{align*}
	z\in\mathbb{R}:\ \ G(z)=&\,\begin{bmatrix}1 & -\bar{r}(z)\e^{-2\im(2tz^2+xz)}\\ 0  & 1\end{bmatrix}\begin{bmatrix}1 & 0\\ r(z)\e^{2\im(2tz^2+xz)} & 1\end{bmatrix}\\
	=&\,\begin{bmatrix}1 & 0\\ \frac{r(z)}{1-|r(z)|^2}\,\e^{2\im(2tz^2+xz)} & 1\end{bmatrix}\begin{bmatrix}1-|r(z)|^2 & 0\smallskip\\ 0 & \frac{1}{1-|r(z)|^2}\end{bmatrix}\begin{bmatrix}1 & -\frac{\bar{r}(z)}{1-|r(z)|^2}\,\e^{-2\im(2tz^2+xz)} \\ 0 & 1\end{bmatrix}
\end{align*}
one may now deform the jump contour in Problem \ref{HP3} (according to the asymptotic regime at hand while taking into account the analytic properties of the reflection coefficient) and thus localize the problem near the relevant stationary point $z_0=-\frac{x}{4t}$. Away from that point, the deformed Hilbert boundary value problem will be asymptotically well-behaved in the sense that its jump matrix converges to the identity in $L^1(\Sigma,|\d z|)\cap L^2(\Sigma,|\d z|)\cap L^{\infty}(\Sigma,|\d z|)$ as the parameters become large. We are thus facing a \textit{small norm problem} away from the stationary point and such a problem is amenable to an iterative solution method. Here is the technical core of the argument:
\begin{theo}[\cite{DZ}]\label{essence} Assume that the jump matrix $G(z)=G(z;t)$ in Problem \ref{HP7} depends on a parameter $t$ and that there exist $\epsilon,t_0,c>0$ such that
\begin{equation}\label{e:48}
	\|G(\cdot;t)-\mathbb{I}\|_{L^1(\Sigma,|\d z|)\cap L^2(\Sigma,|\d z|)\cap L^{\infty}(\Sigma,|\d z|)}\leq \frac{c}{t^{\epsilon}}\ \ \ \forall\,t\geq t_0.
\end{equation}
Then Problem \ref{HP7} is uniquely solvable for $t$ sufficiently large, i.e. there exists $t_1>0$ such that Problem \ref{HP7} is uniquely solvable for all $t\geq t_1$. Moreover, the extension $X=X(z;t)$ \eqref{e:45} defined in terms of the problem's solution satisfies
\begin{equation*}
	\|X(z;t)-\mathbb{I}\|\leq\frac{C}{(1+|z|)t^{\epsilon}}\ \ \ \ \ \forall\,t\geq t_1
\end{equation*}
uniformly on compact subsets of $\mathbb{C}\setminus\Sigma$ with $C>0$.
\end{theo}
\begin{proof} Rewrite the integral equation \eqref{e:44} as 
\begin{equation*}
	\rho(z;t)-\mathbb{I}=\frac{1}{2\pi\im}\int_{\Sigma}\big(G(w;t)-\mathbb{I}\big)\frac{\d w}{w-z_-}+\frac{1}{2\pi\im}\int_{\Sigma}\big(\rho(w;t)-\mathbb{I}\big)\big(G(w;t)-\mathbb{I}\big)\frac{\d w}{w-z_-}\ \ \ \textnormal{a.e. on}\ \ \Sigma
\end{equation*}
and abbreviate (recall Remark \ref{L2bound}) the operator
\begin{equation*}
	(Tf)(z;t):=\frac{1}{2\pi\im}\int_{\Sigma}f(w;t)\big(G(w;t)-\mathbb{I}\big)\frac{\d w}{w-z_-}\equiv \big(C_-f(G(\cdot;t)-\mathbb{I})\big)(z;t).
\end{equation*}
However by assumption \eqref{e:48} and \eqref{crucialbound}, in $L^2(\Sigma,|\d z|)$ operator norm,
\begin{equation*}
	\|T\|\leq \|C_-\|\,\|G(\cdot;t)-\mathbb{I}\|_{L^{\infty}(\Sigma,|\d z|)}\leq\frac{c}{t^{\epsilon}}\ \ \ \ \ \ \ \forall\,t\geq t_0,
\end{equation*}
which shows that that the integral equation
\begin{equation*}
	f(z;t)-(Tf)(z;t)=\frac{1}{2\pi\im}\int_{\Sigma}\big(G(w;t)-\mathbb{I}\big)\frac{\d w}{w-z_-}=:R(z;t)\equiv\big(C_-(G(\cdot;t)-\mathbb{I}\big)(z;t)\ \ \ \textnormal{a.e. on}\ \ \Sigma
\end{equation*}
is solvable in $L^2(\Sigma,|\d z|)$ for sufficiently large $t$ by the Neumann series. Moreover, since its right hand side $R(z;t)$ satisfies
\begin{equation*}
	\|R(\cdot;t)\|_{L^2(\Sigma,|\d z|)}\leq \|C_-\|\,\|G(\cdot;t)-\mathbb{I}\|_{L^2(\Sigma,|\d z|)}\leq\frac{c}{t^{\epsilon}}\ \ \ \ \ \ \ \forall\,t\geq t_0,
\end{equation*}
the unique solution $\rho(\cdot;t)-\mathbb{I}\in L^2(\Sigma,|\d z|)$ of \eqref{e:44} also satisfies 
\begin{equation*}
	\|\rho(\cdot;t)-\mathbb{I}\|_{L^2(\Sigma,|\d z|)}\leq \frac{c}{t^{\epsilon}}\ \ \ \ \ \forall\,t\geq t_0.
\end{equation*}
In summary, using Proposition \ref{van}, the Hilbert boundary value problem \ref{HP7} is uniquely solvable for sufficiently large $t$ and we find from \eqref{e:45} in turn (this is the only time we use the $L^1(\Sigma,|\d z|)$ bound in \eqref{e:48}),
\begin{align*}
	\|X(z;t)-\mathbb{I}\|\leq&\,\frac{1}{2\pi}\left\|\int_{\Sigma}\big(G(w;t)-\mathbb{I}\big)\frac{\d w}{w-z}\right\|+\frac{1}{2\pi}\left\|\int_{\Sigma}\big(\rho(w;t)-\mathbb{I}\big)\big(G(w;t)-\mathbb{I}\big)\frac{\d w}{w-z}\right\|\\
	\leq&\,\frac{1}{2\pi}\frac{1}{\textnormal{dist}(z,\Sigma)}\Big(\big\|G(\cdot;t)-\mathbb{I}\big\|_{L^1(\Sigma,|\d z|)}+\big\|\rho(\cdot;t)-\mathbb{I}\big\|_{L^2(\Sigma,|\d z|)}\,\big\|G(\cdot;t)-\mathbb{I}\big\|_{L^2(\Sigma,|\d z|)}\Big)\\
	\leq&\,\frac{C}{(1+|z|)t^{\epsilon}}\ \ \ \ \ \forall\,t\geq t_1\geq t_0,
\end{align*}
uniformly on compact subsets of $\mathbb{C}\setminus\Sigma$. This concludes the proof of Theorem \ref{essence}.
\end{proof}
Near the stationary point(s) non-trivial technicalities of the nonlinear steepest descent method appear: one has to, either explicitly or implicitly, construct appropriate model functions in terms of classical or Painlev\'e special functions or, more generally, in terms of solutions to other dynamical systems that approximate the Hilbert boundary value problem asymptotically near the stationary point(s). Combining such local constructions with the aforementioned estimates away from the stationary point(s) one then typically arrives at a \textit{global small norm problem} which is amenable to Theorem \ref{essence}
\begin{rem} The above outlined Riemann-Hilbert nonlinear steepest descent techniques have been successfully applied in high performance numerical simulations, see the recent monograph \cite{TO} for a summary of research efforts in this direction. 
\end{rem}

Theory aside, we now discuss a series of impressive results which have been derived with Riemann-Hilbert techniques.
\subsection{Connection formul\ae\,for Painlev\'e-II transcendents}\label{PIIconn} Let us return to \eqref{e:28} and focus on the solution family parametrized by the monodromy data
\begin{equation*}
	s\equiv (s_1,s_2,s_3)=(-\im\sqrt{\gamma},0,\im\sqrt{\gamma}),\ \ \ \gamma\geq 0.
\end{equation*}
This one-parameter family is known as the Ablowitz-Segur family and it plays a dominant role in random matrix theory and integrable probability, see Subsection \ref{Ulam} below and Section \ref{ex:6}. Depending on the values of $\gamma$, the solution $u=u(x|s)$ of \eqref{e:28} has different analytic and asymptotic properties, see Figure \ref{fig6}. 
\begin{figure}[tbh]
\includegraphics[width=0.45\textwidth]{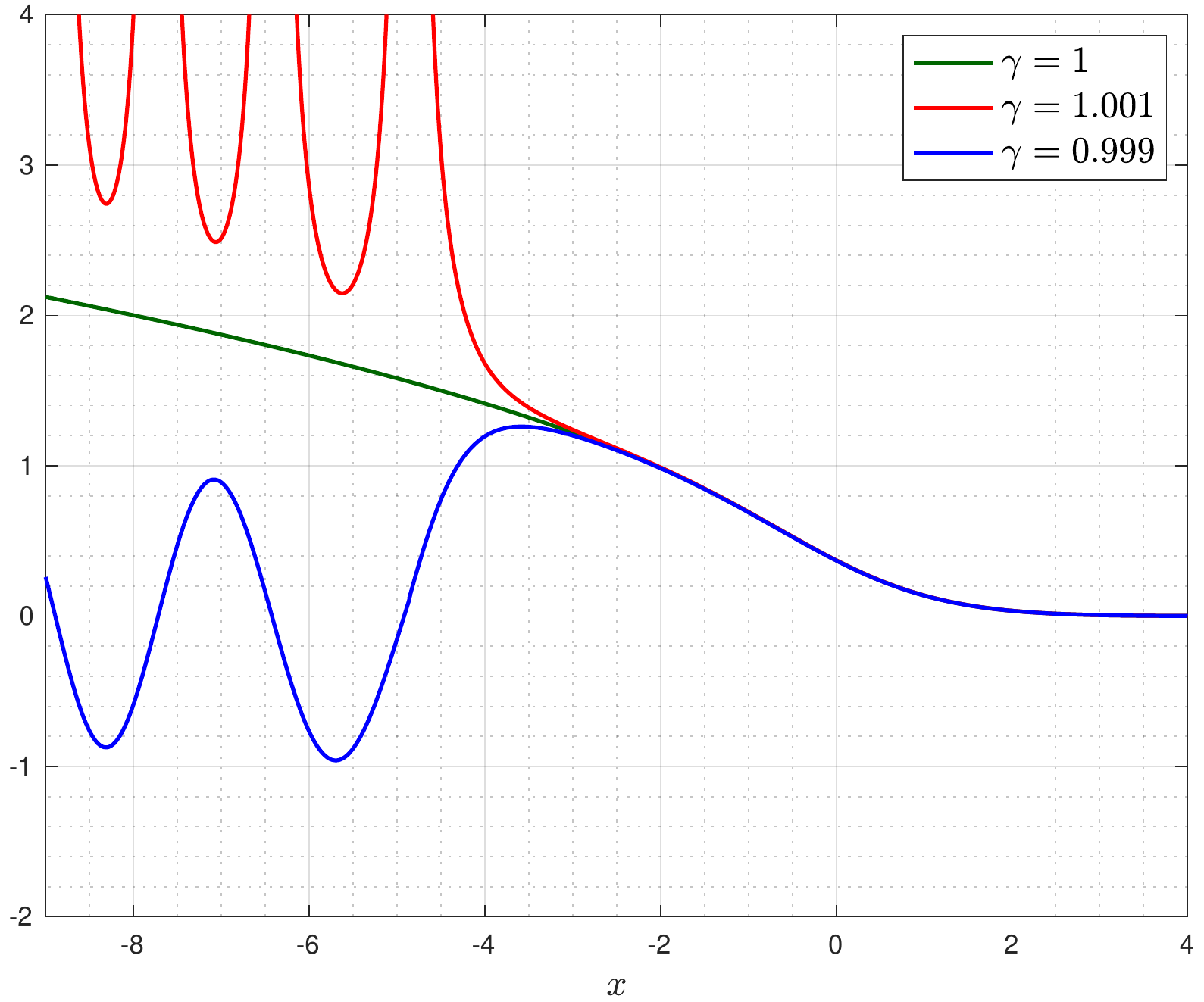}
\caption{The sensitive dependence of $u=u(x|s)$ for $x\leq 0$ on the monodromy data $s=(-\im\sqrt{\gamma},0,\im\sqrt{\gamma})$ for three values of $\gamma$. In \textcolor{ao}{green} we show $\gamma=1$, in \textcolor{red}{red} $\gamma=1.001$ and in \textcolor{blue}{blue} $\gamma=0.999$.}
\label{fig6}
\end{figure}

For one, the Ablowitz-Segur solution $u=u(x|s)$ is smooth and bounded for $x>0$ and if $x\rightarrow+\infty$, then
\begin{equation}\label{e:49}
	u(x|s)\sim\sqrt{\gamma}\,\frac{x^{-\frac{1}{4}}}{2\sqrt{\pi}}\e^{-\frac{2}{3}x^{\frac{3}{2}}},
\end{equation}
which is proportional to the super exponential decay behavior of the Airy function on the positive real axis (recall that \eqref{e:28} can be viewed as a nonlinear Airy equation). However, things change dramatically on the half ray $x\leq 0$: if $\gamma\in[0,1)$, then $u(x|s)$ is smooth, bounded and with the following oscillatory behavior as $x\rightarrow-\infty$,
\begin{align}\label{e:50}
	u(x|s)\sim(-x)^{-\frac{1}{4}}&\,\sqrt{-2\beta}\,\cos\left(\frac{2}{3}(-x)^{\frac{3}{2}}+\beta\ln\big(8(-x)^{\frac{3}{2}}\big)+\phi\right)\\
	&\,\beta=\frac{1}{2\pi}\ln(1-\gamma),\ \ \ \ \ \phi=\frac{\pi}{4}-\textnormal{arg}\,\Gamma(\im\beta).\nonumber
\end{align}
Expansion \eqref{e:50} is due to Ablowitz and Segur \cite{AS1,AS2} and was partially proven in \cite{HM,CM} by Hastings, McLeod and Clarkson, using Gelfand-Levitan type integral equations.
The first complete proof of the leading order in \eqref{e:50} is due to Its, Kapaev, Suleimanov and Kitaev \cite{IK,Su,Ki} via the isomonodromy method. The first nonlinear steepest descent based proof of \eqref{e:50} is in the paper \cite{DZ2} by Deift and Zhou. Next, for $\gamma=1$, the solution $u(x|s)$ is still smooth for $x\leq 0$ but now unbounded, in fact as $x\rightarrow-\infty$
\begin{equation}\label{e:51}
	u(x|s)\sim\sqrt{-\frac{x}{2}},
\end{equation}
which, to leading order, can be found in the work of Hastings and McLeod \cite{HM} and to all orders in \cite{DZ2}, using once more the nonlinear steepest descent method. Finally, if $\gamma>1$, the qualitative behavior of $u(x|s)$ changes completely: its smoothness is destroyed at finite $x<0$ and solutions blow up, compare the below $x\rightarrow-\infty$ asymptotic expansion,
\begin{align}\label{e:52}
	u(x|s)\sim&\,\frac{\sqrt{-x}}{\sin(\frac{2}{3}(-x)^{\frac{3}{2}}+\hat{\beta}\ln(8(-x)^{\frac{3}{2}})+\varphi)}\\
	&\,\,\,\,\,\hat{\beta}=\frac{1}{2\pi}\ln(\gamma-1),\ \ \ \ \ \varphi=\frac{\pi}{2}-\textnormal{arg}\,\Gamma\left(\frac{1}{2}+\im\hat{\beta}\right),\nonumber
\end{align}
which is uniform away from the zeros of the trigonometric function appearing in the denominator of \eqref{e:52}. The singular expansion \eqref{e:52} is originally due to Kapaev \cite{Kap} via isomonodromy techniques and was re-derived in \cite{BoI} by nonlinear steepest descent techniques. Note that formul\ae\,\eqref{e:49}, \eqref{e:50}, \eqref{e:51}, \eqref{e:52} solve the connection problem for the underlying Painlev\'e-II transcendent, namely the problem of completely determining the asymptotic behavior of $u(x|s)$ as $x\rightarrow-\infty$, given the same description of $u(x|s)$ as $x\rightarrow+\infty$, or vice versa. Such connection problems are standard exercises in classical special function theory and commonly solved based on contour integral formul\ae. However, the Ablowitz-Segur transcendent $u(x|s)$ does \textit{not} admit any such explicit formul\ae\, so it is a great achievement of the Riemann-Hilbert method to be able to solve this genuinely nonlinear connection problem in terms of classical special functions. Even better, one can uniformize the qualitatively different asymptotic behaviors \eqref{e:50}, \eqref{e:51}, \eqref{e:52} at $x=-\infty$ through the use of Jacobi elliptic functions. This recent breakthrough is contained in the paper \cite{Bo} and was made possible by nonlinear steepest descent techniques.

\subsection{The Heisenberg antiferromagnet} We now return to \eqref{e:31} and discuss two central results for $P_n$ proven in \cite[Section $6$]{DIZ}. First, the leading order asymptotic behavior,
\begin{equation*}
	\ln P_n=n^2\ln\sin\frac{\phi}{2}+o(n),\ \ \ n\rightarrow\infty,
\end{equation*}
which holds uniformly on compact subsets of $(\frac{\pi}{2},\pi)\ni\phi$. Hence, an overall Gaussian decay of $P_n$ with a decay rate that decreases with increasing external magnetic field. Second, the emptiness formation probability \eqref{e:31} really depends on two variables, $n$ and $h$, equivalently $n$ and $\xi=\e^{\im\phi}$, compare \eqref{e:32}. As it turns out $P_n$ shares an intimate relation with the Painlev\'e-VI equation, more precisely $P_n$ is the Jimbo-Miwa-Ueno $\tau$-function \cite[Section $5$]{JMU} of a particular Painlev\'e-VI equation when viewed in those variables. Here are the details, see \cite[Theorem $6.48$]{DIZ}: Set $t:=\bar{\xi}/\xi$, then
\begin{equation*}\label{e:53}
	t\frac{\d}{\d t}\ln P_n=\frac{t-1}{4y(y-1)(y-t)}\nu^2+\frac{n}{2y}\nu;\ \ \ \ \ \ 
	\nu:=(y-1)\frac{(1-n)y+nt}{t-1}-t\frac{\d y}{\d t},
\end{equation*}
where $y=y(t;n)$ satisfies the Painlev\'e-VI equation
\begin{align*}\label{e:54}
	\frac{\d^2y}{\d t^2}=\frac{1}{2}\left(\frac{1}{y}+\frac{1}{y-1}+\frac{1}{y-t}\right)&\,\left(\frac{\d y}{\d t}\right)^2-\left(\frac{1}{t}+\frac{1}{t-1}+\frac{1}{y-t}\right)\frac{\d y}{\d t}\\
	&\,\,+\frac{y(y-1)(y-t)}{t^2(t-1)^2}\left(\frac{1}{2}(n-1)^2-\frac{tn^2}{2y^2}+\frac{t(t-1)}{2(y-t)^2}\right).\nonumber
\end{align*}
This impressive connection between $P_n$ and Painlev\'e special function theory puts the spin chain model firmly on integrable systems ground and as such should be compared with the famous Jimbo, Miwa 1980 result \cite{JM} on the diagonal Ising correlation functions and their relation to a $\sigma$-version of Painlev\'e-VI.

\subsection{Ulam's problem}\label{Ulam} In this section we discuss one of the most beautiful and influential results in mathematical physics that was derived by Riemann-Hilbert techniques over the past 20 years, see \cite{Cor} for a taste of the following discovery's impact: Choose a permutation $\pi\in S_n$. For $1\leq i_1<i_2<\ldots<i_k\leq n$ we say that $(\pi(i_1),\ldots,\pi(i_k))$ is an \textit{increasing subsequence} of $\pi$ \textit{of length} $k$ if
\begin{equation*}
	\pi(i_1)<\pi(i_2)<\ldots<\pi(i_k).
\end{equation*}
For instance, if $\pi=(4,2,1,3,5)\in S_5$, then $(4,5),(2,3),(2,5),(1,3),(1,5),(3,5)$ are subsequences of length $2$ and $(2,3,5),(1,3,5)$ are subsequences of length $3$. In turn, the \textit{maximal length} $\ell_n(\pi)$ of all increasing subsequences of $\pi$ equals $\ell_5(\pi)=3$ in this example. Now randomize the setup and pick $\pi\in S_n$ with uniform distribution, i.e. we have
\begin{equation*}
	\textnormal{Prob}\big\{\ell_n\leq N\big\}=\frac{1}{n!}\,\#\{\pi\in S_n:\,\ell_n(\pi)\leq N\}.
\end{equation*}
How does $\ell_n$ behave statistically for large $n$? This question was raised by Ulam in 1961 \cite{U} and he himself conjectured the first moment convergence
\begin{equation}\label{e:55}
	\lim_{n\rightarrow\infty}\frac{\mathbb{E}(\ell_n)}{\sqrt{n}}=c
\end{equation}
for some constant $c$. Ulam's problem became known as the problem of rigorously verifying the convergence \eqref{e:55} and in turn computing the constant. Many people with different methods contributed to this challenge including Erd\H{o}s, Hammersley, Logan, Shepp and Vershik, Kerov, ultimately settling on $c=2$, see the recent monograph \cite[page $3-4$]{BDS} or \cite[Section $1.1$]{Rom} for references. The next statistical quantity of $\ell_n$ to be considered was its variance. High precision numerical simulations by Odlyzko in the early 1990s indicated that
\begin{equation*}
	\lim_{n\rightarrow\infty}\frac{\textnormal{Var}(\ell_n)}{\sqrt[3]{n}}=c_0\approx 0.819\ldots,
\end{equation*}
and from simulations by Odlyzko and Rains
\begin{equation}\label{e:56}
	\lim_{n\rightarrow\infty}\frac{\mathbb{E}(\ell_n)-2\sqrt{n}}{\sqrt[6]{n}}=c_1\approx-1.758\ldots,
\end{equation}
which constituted a higher order simulation of $\mathbb{E}(\ell_n)$ than \eqref{e:55}. At this point of the story, in 1999, the Riemann-Hilbert/integrable systems community in the form of Baik, Deift and Johansson \cite{BDJ} picked up the loose ends and proved the following spectacular result, settling en route \eqref{e:55}, \eqref{e:56} and then some. 
\begin{theo}[{\cite[Theorem $1.1$ and $1.2$]{BDJ}}]\label{Jinho} Let $u(x|s)$ denote the unique solution of \eqref{e:28} with the monodromy data $s=(-\im,0,\im)$, compare Subsection \ref{PIIconn}. Define the distribution function
\begin{equation}\label{e:57}
	F(x):=\exp\left[-\int_x^{\infty}(y-x)u^2(y|s)\,\d y\right],\ \ \ x\in\mathbb{R}.
\end{equation}
Then, for any $x\in\mathbb{R}$,
\begin{equation*}
	\lim_{n\rightarrow\infty}\textnormal{Prob}\left\{\frac{\ell_n-2\sqrt{n}}{\sqrt[6]{n}}\leq x\right\}=F(x)\equiv\textnormal{Prob}\big\{\chi\leq x\big\},
\end{equation*}
and for any $k\in\mathbb{Z}_{\geq 1}$,
\begin{equation*}
	\lim_{n\rightarrow\infty}\mathbb{E}(\chi_n^k)=\mathbb{E}(\chi^k),\ \ \ \ \ \ \ \chi_n:=\frac{\ell_n-2\sqrt{n}}{\sqrt[6]{n}}.
\end{equation*}
\end{theo}
The distribution function $F(x)$ in \eqref{e:57}, and it really is a distribution function, see our discussion on the analytic and asymptotic properties of $u(x|s)$ in Subsection \ref{PIIconn}, had appeared a few years prior to the Baik-Deift-Johansson theorem. Namely in the 1994 work of Tracy and Widom \cite{TW} where it was shown to capture the fluctuations of the largest eigenvalue of a matrix drawn from the ensemble \eqref{e:36} with quadratic external field, a.k.a.  the \textit{Gaussian unitary ensemble} (GUE), in the large $n$ limit. The appearance of a random matrix theory distribution function in a problem of combinatorics was a big surprise at the time, however it turned out not to be an isolated curiosity. The Tracy-Widom distribution $F(x)$ \eqref{e:57} (that name stuck after \cite{BDJ}) has become an ubiquitous limiting distribution that made its way to the financial markets, to wireless communication networks, stochastic growth processes and statistical mechanics, just to name a few areas. See \cite{Q} for an engaging outreach article on the Tracy-Widom law. In fact, twenty years after \cite{BDJ} it has become clear that the Tracy-Widom distribution is a modern day bell curve in that it describes the random fluctuations of a large class of integrable probabilistic models, see Figure \ref{fig7} for a graph of $F(x)$ and $F'(x)$.
\begin{figure}[tbh]
\includegraphics[width=0.465\textwidth]{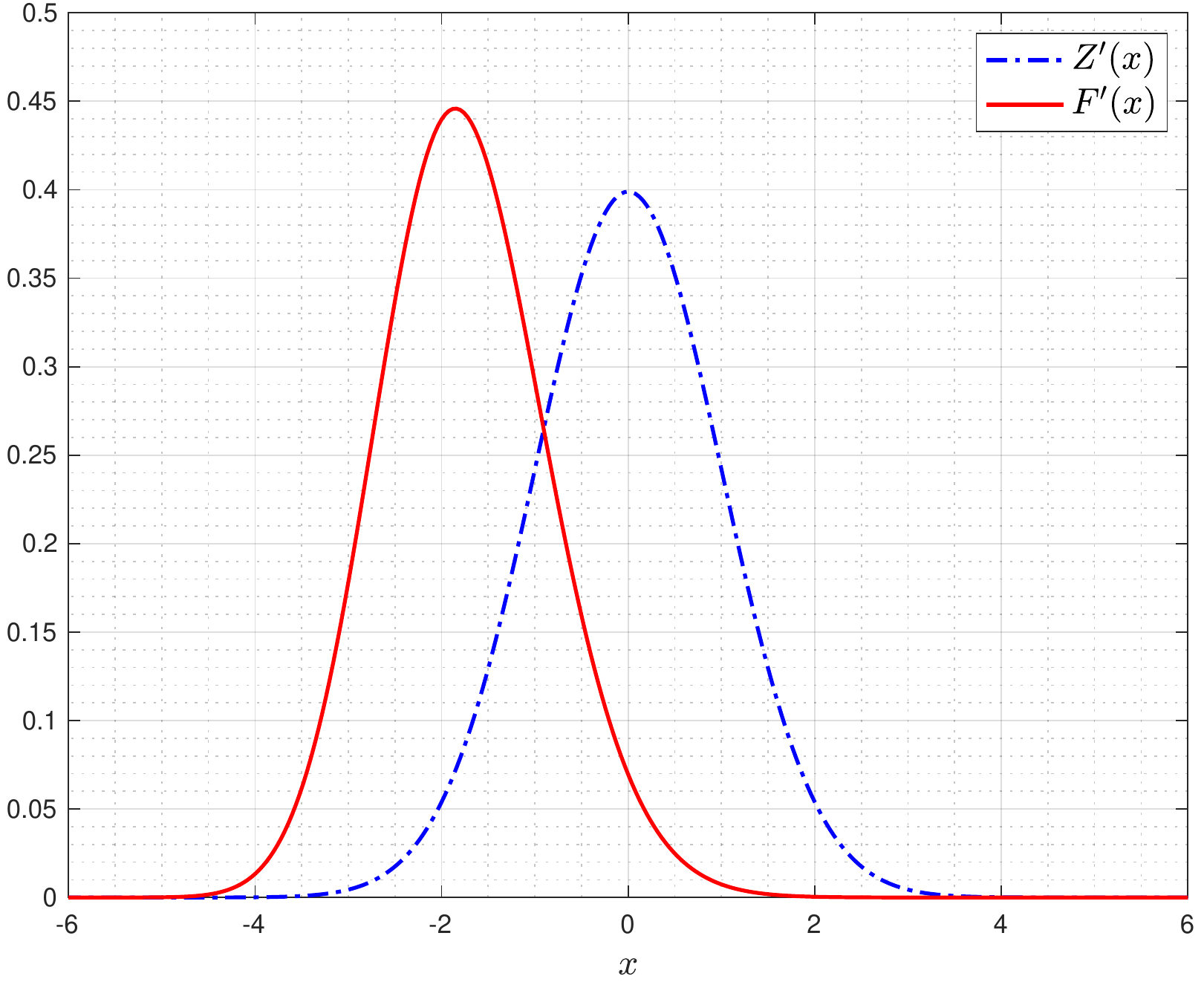}\,\,\,\,\,\,\,\,\,\,\,\,
\includegraphics[width=0.46\textwidth]{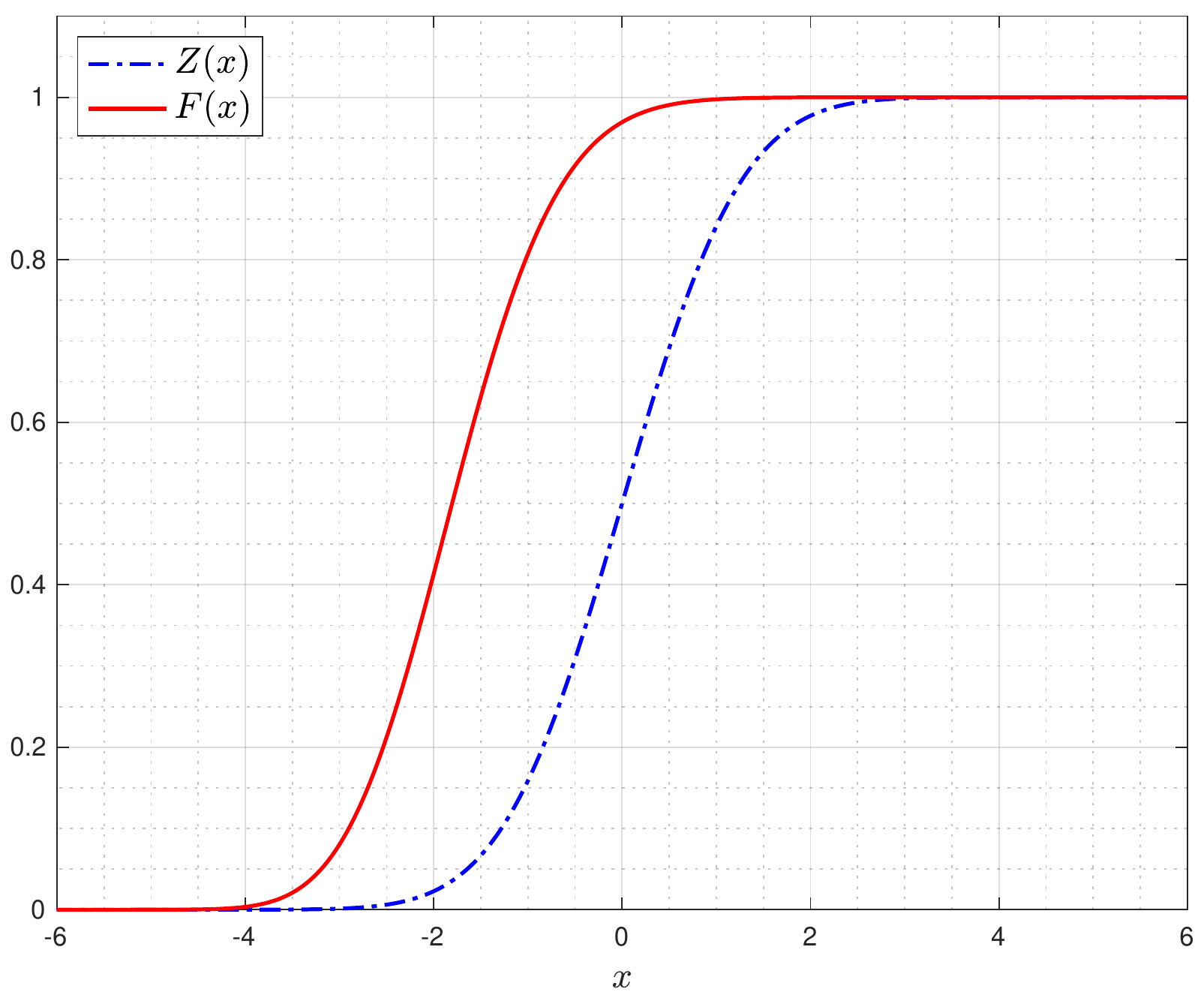}
\caption{Comparing the Tracy-Widom distribution $F(x)$ in \textcolor{red}{red} to the standard normal distribution $Z(x)$ in \textcolor{blue}{blue}. On the left, probability density functions, on the right cumulative distribution functions. The asymmetry of the Tracy-Widom distribution is clearly visible.}
\label{fig7}
\end{figure}

The proof of Theorem \ref{Jinho} is beyond the material covered in this short note, it is a tour de force journey through (de)-Poissonization arguments, Toeplitz determinants and finally the asymptotic nonlinear steepest descent analysis of orthogonal polynomials on the unit circle. These polynomials can be characterized via a Hilbert boundary value problem in the style of Problem \ref{HP6}, we refer the interested reader to the monographs \cite{BDS,Rom}.

\subsection{Universality in invariant random matrix models} Our last discussion in this section concerns the global and local statistical eigenvalue behavior in the random matrix model \eqref{e:36}. First, through our assumptions placed on the external field $V:\mathbb{R}\rightarrow\mathbb{R}$, the model's mean eigenvalue density $\frac{1}{n}K_{n,N}(x,x)$, compare \eqref{e:37}, has a limit
\begin{equation}\label{e:58}
	\lim_{\substack{n,N\rightarrow\infty\\ \frac{n}{N}\rightarrow 1}}\frac{1}{n}K_{n,N}(x,x)=\rho_V(x)\geq 0,
\end{equation}
whose support $\Sigma_V:=\overline{\{x\in\mathbb{R}:\,\rho_V(x)>0\}}$ is a finite union of intervals. The \textit{global} limiting behavior \eqref{e:58} contains the most basic result in random matrix theory as special case, namely the Wigner semicircular law which asserts the almost sure convergence of the rescaled empirical spectral distribution in the GUE to the limiting distribution with density
\begin{equation*}
	\rho_W(x)=\frac{1}{\pi}\sqrt{(2-x^2)_+}.
\end{equation*}
\begin{figure}[tbh]
\includegraphics[width=0.453\textwidth]{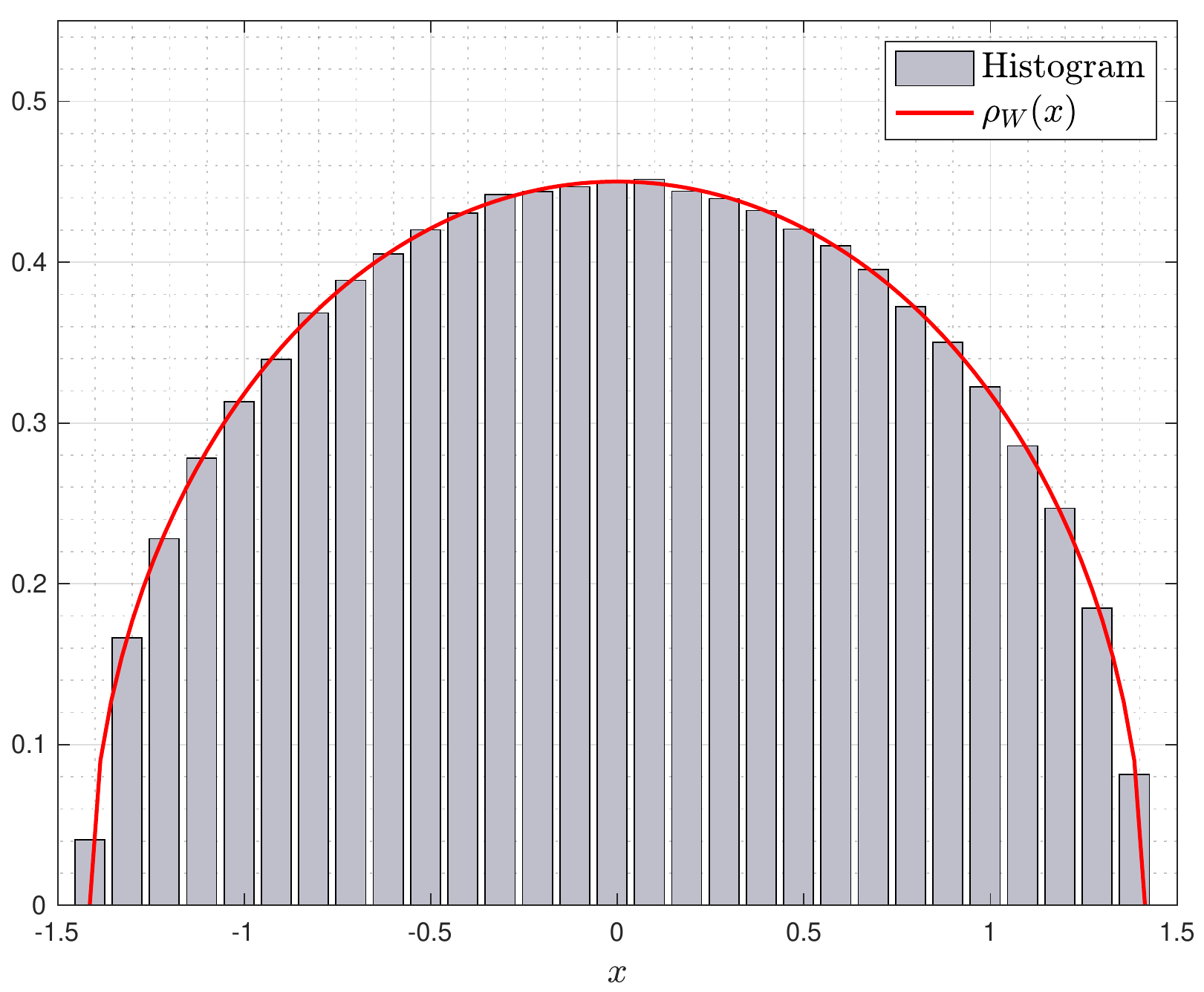}\,\,\,\,\,\,\,\,\,\,\,\,
\includegraphics[width=0.45\textwidth]{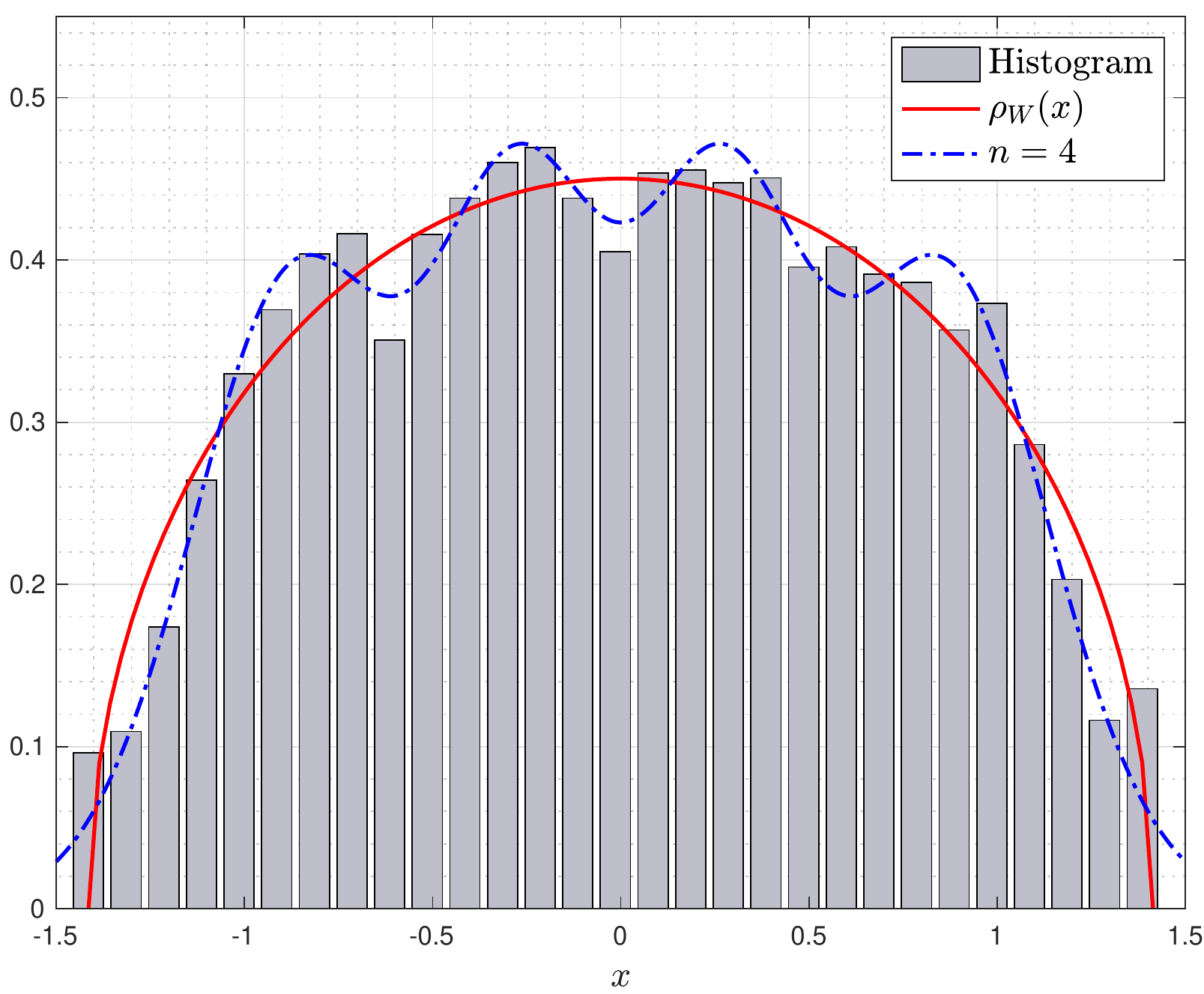}
\caption{Wigner's semicircular law for a rescaled $500\times 500$ GUE matrix on the left. Plotted is the rescaled histogram of its $500$ eigenvalues, sampled $100$ times, against the semicircular density in \textcolor{red}{red}. On the right we compare Wigner's law to the exact eigenvalue density for $n=4$ in \textcolor{blue}{blue} and to the associated eigenvalue histogram, again sampled $100$ times.}
\label{fig8}
\end{figure}

While the limiting density $\rho_V$ in \eqref{e:58} is determined by the potential $V$, the \textit{local} eigenvalue statistics are to a large extent independent of $V$ in the large $n,N$ limit. Precisely, and this is the celebrated \textit{universality property}, the local scaling behavior of $K_{n,N}(x,y)$ is completely determined by the local characteristica of $\rho_V$. For instance, in case of \eqref{e:36}, only two qualitatively different scenarios arise generically, compare Figure \ref{fig8} for the GUE: if $x^{\ast}\in\Sigma_V$ is such that $\rho_V(x^{\ast})>0$, then
\begin{equation}\label{e:59}
	\lim_{n\rightarrow\infty}\frac{1}{n\rho_V(x^{\ast})}K_{n,n}\left(x^{\ast}+\frac{\lambda}{n\rho_V(x^{\ast})},x^{\ast}+\frac{\mu}{n\rho_V(x^{\ast})}\right)=\frac{\sin\pi(\lambda-\mu)}{\pi(\lambda-\mu)}
\end{equation}
uniformly on compact subsets of $\mathbb{R}\ni\lambda,\mu$. This is the 1996 bulk universality result of Pastur, Shcherbina \cite{PS} which was proven with the resolvent method. Soon after that Bleher, Its \cite{BI} and Deift, Kriecherbauer, McLaughlin, Venakides, Zhou \cite{DKMVZ} fine tuned the Deift-Zhou nonlinear steepest descent method of Section \ref{cool} to Problem \ref{HP6} and arrived at the same result. Especially \cite{DKMVZ} has become the standard reference for nonlinear steepest descent techniques used in the asymptotic analysis of orthogonal polynomials, and it has been extended in various directions. Those directions lead in turn to generalizations of the bulk universality class \eqref{e:59} and extensions to other types of orthogonal polynomial random matrix models, see Section \ref{read} below for a few references. The second case concerns the scaling limit of $K_{n,n}(x,y)$ at a boundary point $x^{\ast}\in\partial\Sigma_V$ where $\rho_V$ typically vanishes square root like, certainly in the GUE. In this case \eqref{e:59} is replaced by the limit
\begin{equation}\label{e:60}
	\lim_{n\rightarrow\infty}\frac{1}{(cn)^{\frac{2}{3}}}K_{n,n}\left(x^{\ast}+\frac{\lambda}{(cn)^{\frac{2}{3}}},x^{\ast}+\frac{\mu}{(cn)^{\frac{2}{3}}}\right)=\frac{\textnormal{Ai}(\lambda)\textnormal{Ai}'(\mu)-\textnormal{Ai}'(\lambda)\textnormal{Ai}(\mu)}{\lambda-\mu},
\end{equation}
if $x^{\ast}$ is a right boundary point and with a sign change $(\lambda,\mu)\mapsto(-\lambda,-\mu)$ in the left hand side of \eqref{e:60} for a left boundary point. The limit is again uniform on compact subsets of $\mathbb{R}\ni\lambda,\mu$, it is expressed in terms of the Airy function $w=\textnormal{Ai}(z)$ and $c>0$ denotes some constant.
The first proof of the soft edge universality \eqref{e:60} is due to Deift, Kriecherbauer, McLaughlin, Venakides, Zhou \cite{DKMVZ} and in more explicit form due to Deift, Gioev \cite{DG}, in both cases based on nonlinear steepest descent techniques.\smallskip

In summary of this short subsection, the earliest universality results in random matrix theory were almost exclusively derived with Riemann-Hilbert nonlinear steepest descent techniques and thus paved the way for subsequent discoveries. Nowadays a wide array of different techniques is available in the derivation of RMT universality classes, some more suitable to certain ensembles than others. Most recently, the work of Erd\H{o}s and Yau \cite{EY} put forth a toolbox useful in the universality analysis of Wigner random matrices, a broader class of symmetric matrices with independent entries. It is an open question whether the Wigner ensemble is accessible via Riemann-Hilbert techniques.

\section{One more example}\label{ex:6}
In this section we will venture into the world of \textit{operator-valued} Hilbert boundary value problems, a topic which has received far less attention so far, especially from a mathematically rigorous perspective. The only exception is the work of Its, Kozlowski \cite{IKo} and to a certain extent the paper \cite{IS} by Its, Slavnov which we will discuss towards the end. Here is our motivation to think about such problems: consider the stochastic heat equation (SHE)
\begin{equation*}
	\mathcal{Z}_T=\frac{1}{2}\mathcal{Z}_{XX}-\mathcal{Z}\dot{\mathscr{W}},\ \ \ \ \ \ \ \ \mathcal{Z}=\mathcal{Z}(X,T):\mathbb{R}\times(0,\infty)\rightarrow\mathbb{R}
\end{equation*}
with distributional initial data $\mathcal{Z}(X,0)=\delta(X)$ and a Gaussian white noise $\dot{\mathscr{W}}=\dot{\mathscr{W}}(X,T)$. Using the Hopf-Cole transformation $\mathcal{H}(X,T):=-\ln \mathcal{Z}(X,T)$ the SHE is mapped to the celebrated Kardar-Parisi-Zhang (KPZ) equation 
\begin{equation}\label{SH:0}
	\mathcal{H}_T=\frac{1}{2}\mathcal{H}_{ZZ}-\frac{1}{2}(\mathcal{H}_Z)^2+\dot{\mathscr{W}}
\end{equation}
which has become the quintessential model for random surface growth processes with several remarkable connections to a number of different physical phenomena, see \cite{Cor0} for a detailed summary. One of the long-standing goals in the analysis of the KPZ equation was achieved in 2011 when Amir, Corwin and Quastel \cite{ACQ} derived the exact probability distribution of its solution with narrow wedge initial data. This allowed in turn to finally verify the celebrated KPZ scaling hypothesis for the KPZ equation itself. Here are the relevant details:
\begin{theo}[{\cite[Theorem $1.1$]{ACQ}}]\label{hot} For any $X\in\mathbb{R}$ and $T>0$ the Hopf-Cole solution $\mathcal{H}(X,T)=-\ln \mathcal{Z}(X,T)$ to the KPZ equation \eqref{SH:0} with initial data $\mathcal{Z}(X,0)=\delta(X)$ has the probability distribution
\begin{equation*}
	F_T(s):=\textnormal{Prob}\bigg\{\mathcal{H}(X,T)-\frac{X^2}{2T}-\frac{T}{24}\geq -s\bigg\},\ \ s\in\mathbb{R},
\end{equation*}
which is independent of $X\in\mathbb{R}$ and which equals
\begin{equation}\label{SH:1}
	F_T(s)=\int_{\Gamma}\det\big(1-K_{\sigma_{T,z}}\upharpoonright_{L^2(\kappa_T^{-1}s,\infty)}\big)\e^{-z}\frac{\d z}{z}.
\end{equation}
In \eqref{SH:1} we choose a Hankel contour $\Gamma$ around $\mathbb{R}_{\geq 0}$, we abbreviate $\kappa_T:=(T/2)^{1/3}$ and $K_{\sigma}$ is the integral operator with kernel
\begin{equation*}
	K_{\sigma}(x,y):=\int_{-\infty}^{\infty}\sigma(t)\textnormal{Ai}(x+t)\textnormal{Ai}(y+t)\,\d t\ \ \ \ \ \ \textnormal{and}\ \ \ \ \ \ \sigma_{T,z}(t):=\frac{z}{z-\e^{-\kappa_Tt}}.
\end{equation*}
\end{theo}
The explicit formula \eqref{SH:1} makes it possible to analyze the statistics of the KPZ equation in the two extreme cases $T\downarrow 0$ and $T\rightarrow\infty$. First, under the KPZ scaling $T^{1/3}$, the statistics converges to the Tracy-Widom distribution \eqref{e:57} for large $T$, cf. \cite[Corollary $1.3$]{ACQ},
\begin{equation}\label{SH:2}
	\lim_{T\rightarrow\infty}F_T\big((T/2)^{1/3}x\big)=F(x),\ \ \ \ \ x\in\mathbb{R}.
\end{equation}
Second, under the scaling $T^{1/4}$, the statistics converges to the standard normal distribution for small $T$, cf. \cite[Proposition $1.8$]{ACQ},
\begin{equation}\label{SH:3}
	\lim_{T\downarrow 0}F_T\big(2^{-1/2}(\pi T)^{1/4}(x-\ln\sqrt{2\pi T})\big)=Z(x),\ \ \ \ \ \ x\in\mathbb{R}.
\end{equation}
Hence, the KPZ equation interpolates between two universality classes and \eqref{SH:1} is therefore a \textit{crossover distribution}.\bigskip

From a Riemann-Hilbert/integrable systems viewpoint we are interested in the Fredholm determinant entering the Amir-Corwin-Quastel formula \eqref{SH:1}. Integrating by parts we can rewrite its kernel in the form
\begin{equation}\label{SH:4}
	(x-y)K_{\sigma}(x,y)=\int_{-\infty}^{\infty}\big(\textnormal{Ai}(x+t)\textnormal{Ai}'(y+t)-\textnormal{Ai}'(x+t)\textnormal{Ai}(y+t)\big)\d\sigma(t)
\end{equation}
with the measure $\d\sigma(t):=\sigma'(t)\,\d t$. Besides the particular choice of $\sigma=\sigma_{T,z}$ in \eqref{SH:1}, equality in \eqref{SH:4} holds as soon as $\sigma:\mathbb{R}\rightarrow\mathbb{C}$ is smooth and 
\begin{equation*}
	\lim_{t\rightarrow-\infty}\sigma(t)=0,\ \ \ \ \ \ \ \ \ \ \ \ \lim_{t\rightarrow+\infty}\sigma(t)=1\ \ \ \ \textnormal{exponentially fast}.
\end{equation*}
Note that \eqref{SH:4} includes the Fermi factor choice 
\begin{equation*}
	\sigma(t)=\frac{1}{1+\e^{-\alpha t}},\,\,\,\alpha>0,
\end{equation*}
which plays an important role in the analysis of the Moshe-Neuberger-Shapiro random matrix ensemble \cite[$(14)$]{Joh}, \cite[$(23)$]{LW} and in the analysis of non-interacting quantum many-body systems \cite[$(158)$]{DMS}. Now given the formal - think of $\sigma(t)=\chi_{[0,\infty)}(t)$ - similarities between \eqref{SH:4} and the standard Airy kernel \eqref{e:60}
one might ask for an analogue of the Tracy-Widom Painlev\'e-II formula \eqref{e:57} for the Fredholm determinant of \eqref{SH:4} on $L^2(s,\infty)$, perhaps even an underlying Lax pair which would further strengthen the connection between Theorem \ref{hot} and integrable systems theory. The task of finding an analogue of \eqref{e:57} has been completed in \cite[Section $5$]{ACQ} using the techniques of \cite{TW}, see Theorem \ref{tb:2} below. Here, we will employ Riemann-Hilbert methods and en route reproduce the Amir-Corwin-Quastel formula \cite[Proposition $5.2$]{ACQ} while identifying an underlying Lax pair. Our approach is motivated by the duality between the kernels involved, namely
\begin{equation*}
	K_{\textnormal{Ai}}(x,y)=\frac{\sum_{i=1}^nf_i(x)g_i(y)}{x-y}\ \ \ \ \ \ \ \leftrightarrow\ \ \ \ \ \ \ K_{\sigma}(x,y)=\frac{\int f(x,t)g(y,t)\d\sigma(t)}{x-y},
\end{equation*}
i.e. the matrix-valued methods of Subsection \ref{Hei} will be lifted to operator-valued ones.
\subsection{Technicalities up front} Throughout we are considering kernels $K_{\sigma}$ of the form \eqref{SH:4} where $\d\sigma$ is a positive Borel probability measure on $\mathbb{R}$ for which all moments are finite, i.e.
\begin{equation*}
	\int_{\mathbb{R}}|t|^k\d\sigma(t)<\infty,\ \ \  k\in\mathbb{Z}_{\geq 0},
\end{equation*}
and which is absolutely continuous with respect to the Lebesgue measure.
\begin{definition} Let $p\in\mathbb{Z}_{\geq 1}$. We require the following abbreviations:
\begin{enumerate}
	\item[(1)] The direct sum Hilbert space
	\begin{equation*}
		\mathcal{H}_p:=\bigoplus_{j=1}^pL^2(\mathbb{R},\d\sigma)=\Big\{f=(f_1,\ldots,f_p)^{\top}\in\mathbb{C}^{p\times 1}:\,f_j\in L^2(\mathbb{R},\d\sigma)\Big\},
	\end{equation*}
	equipped with the standard inner product $\langle f,g\rangle_{\mathcal{H}_p}:=\sum_{j=1}^p\int_{\mathbb{R}}f_j(t)\overline{g_j(t)}\d\sigma(t)$ and induced norm $\|f\|_{\mathcal{H}_p}:=\sqrt{\langle f,f\rangle_{\mathcal{H}_p}}$.
	\item[(2)] The space $L^2(\mathbb{R},\d\sigma;\mathbb{C}^{p\times p})$ of $p\times p$ matrix-valued functions with entries in $L^2(\mathbb{R},\d\sigma)$, equipped with the norm
	\begin{equation*}
		\|X\|_{L^2(\mathbb{R},\d\sigma;\mathbb{C}^{p\times p})}:=\left[\int_{\mathbb{R}}\textnormal{tr}\big(X(t)X^{\dagger}(t)\big)\d\sigma(t)\right]^{1/2},\ \ \ \ X=\big[X^{ij}\big]_{i,j=1}^p.
	\end{equation*}
	\item[(3)] The space $\mathcal{I}(\mathcal{H}_p)$ of Hilbert-Schmidt integral operators on $\mathcal{H}_p$ of the form
	\begin{equation*}
		\big(Kf\big)(x)=\int_{\mathbb{R}}K(x,y)f(y)\,\d\sigma(y),
	\end{equation*}
	with kernel $K\in L^2(\mathbb{R}^2,\d\sigma\otimes\d\sigma;\mathbb{C}^{p\times p})$.
	\item[(4)] The matrix identity operator $\mathbb{I}_p$ on $\mathcal{H}_p$.
\end{enumerate}
\end{definition}
Next, fix some region $\Omega\subset\mathbb{C}$ and consider $K=K(z)\in\mathcal{I}(\mathcal{H}_p)$ which depends on an auxiliary variable $z\in\Omega$. We now define the notion of an \textit{analytic integral operator}, cf. \cite[page $1781$]{IKo}.
\begin{definition}\label{anadef} We say that $K(z)\in\mathcal{I}(\mathcal{H}_p)$ with kernel $K(z|x,y)$ is analytic in $z\in\Omega$, if
\begin{enumerate}
	\item[(1)] for any $(x,y)\in\mathbb{R}^2$, the map $z\mapsto K(z|x,y)\in\mathbb{C}^{p\times p}$ is analytic in $\Omega$.
	\item[(2)] for any $z\in\Omega$, the map $(x,y)\mapsto K(z|x,y)$ is in $L^2(\mathbb{R}^2,\d\sigma\otimes\d\sigma;\mathbb{C}^{p\times p})$.
\end{enumerate}
\end{definition}
Furthermore, if $\Sigma\subset\Omega\subset\mathbb{C}$ is an oriented contour which consists of a finite union of smooth oriented curves in $\mathbb{CP}^1$ with finitely many self-intersections, we will need boundary values of $K(z)\in\mathcal{I}(\mathcal{H}_p)$, cf. \cite[page $1782$]{IKo}.
\begin{definition}\label{condef} We say that an analytic in $z\in\Omega\setminus\Sigma$ operator $K(z)\in\mathcal{I}(\mathcal{H}_p)$ admits continuous boundary values $K_{\pm}(z)\in\mathcal{I}(\mathcal{H}_p)$ on $\Sigma'\subset\Sigma$ with kernels $K_{\pm}(z|x,y)$ if
\begin{enumerate}
	\item[(1)] for any $(x,y)\in\mathbb{R}^2$, the map $z\mapsto K_{\pm}(z|x,y)\in\mathbb{C}^{p\times p}$ is continuous on $\Sigma'$.
	\item[(2)] for any $(x,y)\in\mathbb{R}^2$, the non-tangental limits
	\begin{equation*}
		\lim_{\lambda\rightarrow z}K(\lambda|x,y)=K_{\pm}(z|x,y),\ \ \ \ \lambda\in\,\pm\,\textnormal{side of}\,\,\Sigma'\ \textnormal{at}\,z
	\end{equation*}
	exist.
\end{enumerate}
\end{definition}
At this point we record two basic results about the kernel \eqref{SH:4}.
\begin{lem}\label{L:1} Let $K_{\sigma}:L^2(s,\infty)\rightarrow L^2(s,\infty)$ be given by
\begin{equation}\label{SH:5}
	(K_{\sigma}f)(x):=\int_s^{\infty}K_{\sigma}(x,y)f(y)\,\d y,
\end{equation}
in terms of the kernel \eqref{SH:4}. Then $K_{\sigma}$ is trace class for any $s\in\mathbb{R}$.
\end{lem}
\begin{proof} If $\sigma(t)$ is the aforementioned Fermi factor, then the claim is proven in \cite[Proposition $1.1$]{Joh}. For general $\sigma$ we argue as follows: write
\begin{equation*}
	K_{\sigma}(x,y)=\int_{-\infty}^{\infty}\left[\int_0^{\infty}\textnormal{Ai}(x+t+u)\textnormal{Ai}(y+t+u)\,\d u\right]\d\sigma(t),
\end{equation*}
and obtain that for any $x_1,\ldots,x_n\in(s,\infty)$ and $z_1,\ldots,z_n\in\mathbb{C}$,
\begin{equation*}
	\sum_{i,j=1}^nz_i\overline{z_j}K_{\sigma}(x_i,x_j)=\int_{-\infty}^{\infty}\bigg[\int_0^{\infty}\bigg|\sum_{i=1}^nz_i\textnormal{Ai}(x_i+t+u)\bigg|^2\d u\bigg]\d\sigma(t)\geq 0.
\end{equation*}
Thus, with $K_{\sigma}(x,y)=K_{\sigma}(y,x)$, the function $K_{\sigma}(x,y)$ is Hermitian positive definite, so the claim follows provided we can show that
\begin{equation}\label{SH:6}
	\int_s^{\infty}K_{\sigma}(x,x)\,\d x<\infty,
\end{equation}
see \cite[Theorem $2.12$]{Sim}. In order to get to \eqref{SH:6} we compute the elementary integral,
\begin{equation*}
	\int_s^{\infty}\textnormal{Ai}^2(x+u)\,\d u=\big(\textnormal{Ai}'(x+s)\big)^2-(x+s)\textnormal{Ai}^2(x+s),\ \ \ \ x\in\mathbb{R},
\end{equation*}
apply Fubini's theorem,
\begin{equation*}
	\int_s^{\infty}K_{\sigma}(x,x)\,\d x=\int_{-\infty}^{\infty}\left[\int_{s+t}^{\infty}\Big\{\big(\textnormal{Ai}'(x)\big)^2-x\textnormal{Ai}^2(x)\Big\}\d x\right]\d\sigma(t),
\end{equation*}
and estimate for any $y\in\mathbb{R}$, cf. \cite[$\S9.7$]{NIST},
\begin{equation*}
	\int_y^{\infty}\Big|\big(\textnormal{Ai}'(x)\big)^2-x\textnormal{Ai}^2(x)\Big|\d x\leq c|y|^{\frac{3}{2}},\ \ \ c>0.
\end{equation*}
Hence \eqref{SH:6} follows from our finite moment assumption placed on the Borel measure $\d\sigma$.
\end{proof}
Next to the trace class property of $K_{\sigma}$ we will also need the following result which is well-known at least for the Airy operator induced by the kernel \eqref{e:60}, i.e. for $\sigma(t)=\chi_{[0,\infty)}(t)$ in \eqref{SH:4}.
\begin{prop}\label{L:2} For every $s\in\mathbb{R}$, the self-adjoint operator $K_{\sigma}$ on $L^2(s,\infty)$ satisfies $0\leq K_{\sigma}\leq 1$. Moreover, $1-K_{\sigma}$ is invertible on $L^2(s,\infty)$.
\end{prop}
\begin{proof} We employ the standard Fourier transform trick, cf. \cite[Lemma $6.15$]{BDS}, and therefore use
\begin{equation*}
	\textnormal{Ai}(x)=\frac{1}{2\pi}\int_{-\infty}^{\infty}\e^{\im(\frac{1}{3}y^3+xy)}\,\d y,\ \ \ \ x\in\mathbb{R}.
\end{equation*}
First compute for any $f\in L^2(s,\infty)$,
\begin{equation}\label{SH:7}
	\langle f,K_{\sigma}f\rangle_{L^2(s,\infty)}=\int_{-\infty}^{\infty}\bigg[\int_t^{\infty}\left|\int_{-\infty}^{\infty}\textnormal{Ai}(x+u)f_s(x)\,\d x\right|^2\d u\bigg]\d\sigma(t)\geq 0,
\end{equation}
with $f_s(x):=f(x)\chi_{(s,\infty)}(x)$ using the characteristic function of $(s,\infty)\subset\mathbb{R}$. But from the Fourier representation, for any $u\in\mathbb{R}$,
\begin{equation*}
	\int_{-\infty}^{\infty}\textnormal{Ai}(x+u)f_s(x)\,\d x=\frac{1}{\sqrt{2\pi}}\int_{-\infty}^{\infty}\e^{\im(\frac{1}{3}y^3+uy)}\check{f}_s(-y)\,\d y=\check{g}(-u),
\end{equation*} 
where $\check{f}_s(y):=\frac{1}{\sqrt{2\pi}}\int_{-\infty}^{\infty}f_s(x)\e^{-\im xy}\d x$ and we set $g(y):=\e^{\frac{\im}{3}y^3}\check{f}_s(-y)$. Hence, back in \eqref{SH:7}, we find
\begin{align*}
	0\leq\langle f,K_{\sigma}f\rangle_{L^2(s,\infty)}=\int_{-\infty}^{\infty}\bigg[\int_t^{\infty}|\check{g}(-u)|^2\,\d u\bigg]\d\sigma(t)\leq\int_{-\infty}^{\infty}\bigg[\int_{-\infty}^{\infty}|\check{g}(-u)|^2\,\d u\bigg]\d\sigma(t)
\end{align*}
and thus, since $\d\sigma$ is a probability measure,
\begin{equation*}
	0\leq\langle f,K_{\sigma}f\rangle_{L^2(s,\infty)}\leq\int_{-\infty}^{\infty}|\check{g}(-u)|^2\,\d u=\int_{-\infty}^{\infty}|g(u)|^2\,\d u=\int_{-\infty}^{\infty}|f_s(y)|^2\,\d y=\langle f,f\rangle_{L^2(s,\infty)},
\end{equation*}
where we used Plancherel's theorem in the first and second equality. This shows all together that $0\leq K_{\sigma}\leq 1$ and so, by self-adjointness, also $\|K_{\sigma}\|\leq 1$ in operator norm. Finally, since $K_{\sigma}$ is compact, it is sufficient to establish that $1-K_{\sigma}$ is injective in order to conclude that $1-K_{\sigma}$ is invertible on $L^2(s,\infty)$. So assume there is $f\in L^2(s,\infty)\setminus\{0\}$ such that $K_{\sigma}f=f$. Then all inequalities above must be equalities for this $f$, i.e. in particular
\begin{equation*}
	\langle f,K_{\sigma}f\rangle=\int_{-\infty}^{\infty}\bigg[\int_t^{\infty}|\check{g}(-u)|^2\,\d u\bigg]\d\sigma(t)=\int_{-\infty}^{\infty}|\check{g}(-u)|^2\,\d u,
\end{equation*}
and which implies that
\begin{equation*}
	\int_{-\infty}^{\infty}\bigg[\int_{-\infty}^t|\check{g}(-u)|^2\,\d u\bigg]\d\sigma(t)=0.
\end{equation*}
But $\d\sigma$ is a positive Borel probability measure which is absolutely continuous with respect to the Lebesgue measure, so we must have
\begin{equation*}
	\sigma\textnormal{-}\textnormal{a.e.}:\ \ \ \int_{-\infty}^t|\check{g}(-u)|^2\,\d u=0\ \ \ \ \ \ \Rightarrow\ \ \ \ \ \ \textnormal{a.e.}:\ \ \ \check{g}(-t)=\int_s^{\infty}\textnormal{Ai}(x+t)f(x)\,\d x=0,
\end{equation*}
and thus by the analytic properties of the Airy function, $f=0\in L^2(s,\infty)$, a contradiction.
\end{proof}
We now combine Lemma \ref{L:1} with Proposition \ref{L:2} and conclude that the Fredholm determinant
\begin{equation}\label{SH:8}
	F_{\sigma}(s):=\det\big(1-K_{\sigma}\upharpoonright_{L^2(s,\infty)}\big)=\exp\left[-\sum_{n=1}^{\infty}\frac{1}{n}\tr_{L^2(s,\infty)}K_{\sigma}^n\right],\ \ \ \ s\in\mathbb{R}
\end{equation}
is well-defined on one hand and on the other hand satisfies $0<F(s)<1$ for all $s\in\mathbb{R}$. Even better, since
\begin{equation*}
	\int_s^{\infty}\Big\{\big(\textnormal{Ai}'(x)\big)^2-x\textnormal{Ai}^2(x)\Big\}\d x\,\,>\int_{s'}^{\infty}\Big\{\big(\textnormal{Ai}'(x)\big)^2-x\textnormal{Ai}^2(x)\Big\}\d x
\end{equation*}
whenever $s<s'$, we also have that $F_{\sigma}(s)$ is strictly increasing in $s$, so $F_{\sigma}(s)$ is a distribution function of some random variable (one also needs Deift's Lemma \cite[Lemma $2.20$]{ACQ} for the regularity of $F_{\sigma}(s)$ in $s$). Additionally, see \cite[Theorem $3$]{Sosh}, there is a determinantal process with correlation kernel $K_{\sigma}$. These features of \eqref{SH:8} and \eqref{SH:4} are not surprising in light of Johansson's discussion of the interpolating process in \cite[Section $1.4$]{Joh}. However we emphasize that \cite{Joh} uses the Fermi factor choice for $\sigma$, here we are in a more general setup.
\subsection{Solvability of the operator-valued problem} We now proceed to characterize \eqref{SH:8} through a naturally associated Hilbert boundary value problem - very much as in Subsection \ref{Hei}, but the problem is going to be operator-valued from the start. First, write
\begin{equation*}
	F_{\sigma}(s)=\det\big(1-K_{\sigma,s}\upharpoonright_{L^2(0,\infty)}\big),
\end{equation*}
where $K_{\sigma,s}$ has kernel $K_{\sigma,s}(x,y):=K_{\sigma}(x+s,y+s)$. Then consider the following operators (we suppress the explicit $s$-dependence in our notation).
\begin{definition}\label{frank} Let $M_i(z)\otimes K_j(z)\in\mathcal{I}(\mathcal{H}_1),i.j=1,2$ denote the $(0,\infty)\ni z$-parametric family of rank one integral operators with kernels
\begin{equation*}
	\big(M_i(z)\otimes K_j(z)\big)(x,y):=m_i(z|x)k_j(z|y),\ \ \ \ x,y\in\mathbb{R},
\end{equation*}
defined in terms of the $(0,\infty)\ni z$-parametric family of functions
\begin{equation}\label{opdef}
	m_1(z|x)=k_2(z|x):=\textnormal{Ai}'(z+s+x),\ \ \ \ \ \ \ m_2(z|x)=k_1(z|x):=\textnormal{Ai}(z+s+x).
\end{equation}
Equivalently, $M_i(z)$ are the operators on $\mathcal{H}_1$ which multiply by the functions $m_i(z|x)$,
\begin{equation*}
	\big(M_i(z)f\big)(x):=m_i(z|x)f(x),
\end{equation*}
and $K_j(z)$ are the integral operators on $\mathcal{H}_1$ with kernel $k_j(z|y)$,
\begin{equation*}
	\big(K_j(z)f\big)(x):=\int_{-\infty}^{\infty}k_j(z|y)f(y)\,\d\sigma(y).
\end{equation*}
\end{definition}
Here is the relevant operator-valued Hilbert boundary value problem.

\begin{prob}\label{HP8} For any $s\in\mathbb{R}$, determine $X(z)=X(z;s)\in\mathcal{I}(\mathcal{H}_2)$ such that
\begin{enumerate}
	\item[(1)] $X(z)=\mathbb{I}_2+X_0(z)$ and $X_0(z)\in\mathcal{I}(\mathcal{H}_2)$ with kernel $X_0(z|x,y)$ is analytic in $\mathbb{C}\setminus[0,\infty)$.
	\item[(2)] $X(z)$ admits continuous boundary values $X_{\pm}(z)\in\mathcal{I}(\mathcal{H}_2)$ on $(0,\infty)\subset\mathbb{R}$, oriented from $-\infty$ to $+\infty$, which satisfy $X_+(z)=X_-(z)G(z)$ with
	\begin{equation}\label{SH:9}
		G(z)=\mathbb{I}_2+2\pi\im\begin{bmatrix}M_1(z)\otimes K_1(z) & -M_1(z)\otimes K_2(z)\smallskip\\
		M_2(z)\otimes K_1(z) & -M_2(z)\otimes K_2(z)\end{bmatrix}.
	\end{equation}
	\item[(3)] There exist operators $F_i$ on $\mathcal{H}_1$ which multiply by $z$-independent functions $f_i\in L^1(\mathbb{R},\d\sigma)\cap L^{\infty}(\mathbb{R},\d\sigma)$ such that for $|z|<\epsilon$, in terms of any branch for $\ln:\mathbb{C}\setminus[0,\infty)\rightarrow\mathbb{C}$,
	\begin{equation}\label{SH:10}
		X(z)=\mathbb{I}_2-\begin{bmatrix}F_1\otimes K_1(0) & -F_1\otimes K_2(0)\smallskip\\
		F_2\otimes K_1(0) & -F_2\otimes K_2(0)\end{bmatrix}\ln z+X_{\textnormal{reg}}(z),
	\end{equation}
	where $X_{\textnormal{reg}}(z)\in\mathcal{I}(\mathcal{H}_2)$ with kernel $X_{\textnormal{reg}}(z|x,y)$ is analytic at $z=0$ and satisfies
	\begin{equation}\label{SH:11}
		 \|X_{\textnormal{reg}}(z|x,y)\|\leq c|1+s+x|^{\frac{3}{4}}|1+s+y|^{\frac{3}{4}}\ \ \ \textnormal{as}\ \ z\rightarrow 0
	\end{equation}
	for all $(x,y)\in\mathbb{R}^2$ with $c>0$.
	\item[(4)] For any compact subset $V\subset\mathbb{C}$ with $\textnormal{dist}(0,V)>0$ there exists $c=c(V)>0$ such that for $z\in\mathbb{C}\setminus V$,
	\begin{equation}\label{SH:12}
		\|X_0(z|x,y)\|\leq\frac{c}{1+|z|}|s+x|^{\frac{3}{4}}|s+y|^{\frac{3}{4}}
	\end{equation}
	uniformly in $(x,y)\in\mathbb{R}^2$.
\end{enumerate}
\end{prob}
Our first result concerns the unique solvability of Problem \ref{HP8}
\begin{theo}\label{tb:1} The Hilbert boundary value problem \ref{HP8} is uniquely solvable for any $s\in\mathbb{R}$ and its solution takes the form
\begin{equation}\label{SH:13}
	X(z)=\mathbb{I}_2+\int_0^{\infty}\begin{bmatrix}N_1(w)\otimes K_1(w) & -N_1(w)\otimes K_2(w)\smallskip\\
	N_2(w)\otimes K_1(w) & -N_2(w)\otimes K_2(w)\end{bmatrix}\frac{\d w}{w-z},\ \ \ z\in\mathbb{C}\setminus[0,\infty),
\end{equation}
where $N_i(w)$ are the operators on $\mathcal{H}_1$ which multiply by the functions $n_i(w|x)$ determined from the integral equation
\begin{equation}\label{SH:14}
	(1-K_{\sigma,s}\upharpoonright_{L^2(0,\infty)})n_i(\cdot|x)=m_i(\cdot|x),\ \ \ i=1,2,
\end{equation}
with $x\in\mathbb{R}$.
\end{theo}
\begin{proof} Suppose first that Problem \ref{HP8} is solvable. By condition (1), estimate \eqref{SH:12} and the moment assumption on $\d\sigma$, the Fredholm determinant
\begin{equation*}
	e(z):=\det\big(\mathbb{I}_2+X_0(z)\upharpoonright_{\mathcal{H}_2}\big),\ \ \ z\in\mathbb{C}\setminus[0,\infty)
\end{equation*}
is well-defined in the indicated $z$-domain by Hadamard's inequality (by condition (1) only, $X_0(z)$ is a Hilbert-Schmidt operator, but (4) ensures that the ordinary Fredholm determinant on $\mathcal{H}_2$ converges). Moreover, by Morera's and Fubini's theorem, analyticity of $X_0(z)$ in $\mathbb{C}\setminus[0,\infty)$ implies analyticity of $e(z)$ in $\mathbb{C}\setminus[0,\infty)$. Furthermore, using \eqref{SH:12} again and the dominated convergence theorem, the boundary values $e_{\pm}(z),z\in(0,\infty)$ satisfy
\begin{equation}\label{SH:15}
	e_{\pm}(z)=\det\big(X_{\pm}\upharpoonright_{\mathcal{H}_2}\big),\ \ \ z\in(0,\infty).
\end{equation}
Our immediate goal is to derive a scalar Hilbert boundary value problem for $e(z)$: from condition (2) we have $G(z)=\mathbb{I}_2+G_0(z)$ where $G_0(z)\in\mathcal{I}(\mathcal{H}_2)$ is trace class on $\mathcal{H}_2$ and satisfies the estimate, compare Definition \ref{frank} and \cite[$\S 9.7$]{NIST},
\begin{equation*}
	\|G_0(z|x,y)\|\leq c|z+s+x|^{\frac{1}{4}}|z+s+y|^{\frac{1}{4}},\ \ \ \ c>0,
\end{equation*}
uniformly in $(z,x,y)\in\mathbb{R}^3$. This ensures existence of the determinant $\det(\mathbb{I}_2+G_0(z)\upharpoonright_{\mathcal{H}_2})$ for $z\in\mathbb{R}$ and since
\begin{equation*}
	\tr_{\mathcal{H}_2}G_0(z)=0,\ \ \ \ \ \big(G_0(z)\big)^2=0,\ \ \ z\in\mathbb{R},
\end{equation*}
we obtain from the Plemelj-Smithies formula \cite[Chapter II, Theorem $3.1$]{GGK} that $\det(G(z)\upharpoonright_{\mathcal{H}_2})=1$ for all $z\in\mathbb{R}$. With this in hand we return to \eqref{SH:15}, apply condition (2) and the multiplicativity property of Fredholm determinants, in turn
\begin{equation*}
	e_+(z)=e_-(z),\ \ z\in(0,\infty).
\end{equation*}
Moreover, estimate \eqref{SH:12} yields $e(z)\rightarrow 1$ as $z\rightarrow\infty$, so we are left to address the behavior of $e(z)$ in a vicinity of $z=0$. By \eqref{SH:10} and the estimate \eqref{SH:11} we see that the operator
\begin{equation*}
	P(z):=\det\big(\mathbb{I}_2+X_{\textnormal{reg}}(z)\upharpoonright_{\mathcal{H}_2}\big)\big(\mathbb{I}_2+X_{\textnormal{reg}}(z)\big)^{-1},\ \ \ |z|<\epsilon
\end{equation*}
on $\mathcal{H}_2$ is well-defined even if $\det(\mathbb{I}_2+X_{\textnormal{reg}}(z)\upharpoonright_{\mathcal{H}_2})$ vanishes at some points in the disk $|z|<\epsilon$, cf. \cite[Chapter VIII, Theorem $1.1$]{GGK}. Hence from \eqref{SH:10}, as $z\rightarrow 0$ and $z\notin[0,\infty)$,
\begin{equation*}
	e(z)=\det\big(\mathbb{I}_2+X_{\textnormal{reg}}(z)\upharpoonright_{\mathcal{H}_2}\big)-\big\langle P(z)F,K(0)\big\rangle_{\mathcal{H}_2}\ln z,
\end{equation*}
with $F=[F_1,F_2]^{\top},K(0)=\big[K_1(0),K_2(0)\big]^{\top}$. This shows that $|e(z)|\leq c|\ln z|$ when $z\rightarrow 0$ and so all together, $e$ is analytic in $\mathbb{C}\setminus\{0\}$ with a possible blow up at the origin that is logarithmic at worst. Hence, $e(z)$ must be entire and with its limiting behavior at $z=\infty$, thus $e(z)\equiv 1$ for all $z\in\mathbb{C}$. As a consequence of this, we now know that any solution $X(z)\in\mathcal{I}(\mathcal{H}_2)$ of Problem \ref{HP8} is invertible for all $z\in\mathbb{C}\setminus[0,\infty)$ and so are its continuous boundary values $X_{\pm}(z),z\in(0,\infty)$. Thus, if $X^{(1)}(z),X^{(2)}(z)$ are two solutions of Problem \ref{HP8}, then
\begin{equation*}
	\mathcal{I}(\mathcal{H}_2)\ni Y(z):=X^{(1)}(z)\big(X^{(2)}(z)\big)^{-1}=\mathbb{I}_2+Y_0(z),\ \ \ \ z\in\mathbb{C}\setminus[0,\infty)
\end{equation*}
is analytic in $\mathbb{C}\setminus[0,\infty)$, has continuous boundary values $Y_{\pm}(z)$ on $(0,\infty)$ and it obeys a Problem in the style of Problem \ref{HP8} with
\begin{equation*}
	Y_+(z)=Y_-(z),\ \ z\in(0,\infty).
\end{equation*}
Moreover, the kernel $Y_0(z|x,y)$ of $Y_0(z)$ has at worst a square logarithmic singularity at $z=0$, so the map $z\mapsto Y_0(z|x,y)$ is entire $(\sigma\otimes\sigma)$-almost everywhere. But $Y_0(z|x,y)\rightarrow 0$ as $z\rightarrow\infty$ also $(\sigma\otimes\sigma)$-almost everywhere, so by Liouville's theorem we conclude $Y(z)=\mathbb{I}_2$ and which proves that Problem \ref{HP8}, if solvable, is uniquely solvable.\smallskip

Second, we prove solvability of Problem \ref{HP8} and recall to this end that $1-K_{\sigma,s}$ is invertible on $L^2(0,\infty)$ for any $s\in\mathbb{R}$, see Proposition \ref{L:2}. Thus the integral equation \eqref{SH:14} is uniquely solvable and by the boundedness of the resolvent we find for its solution, uniformly in $x\in\mathbb{R}$,
\begin{equation}\label{SH:16}
	\|n_i(\cdot|x)\|_{L^2(0,\infty)}\leq c\|m_i(\cdot|x)\|_{L^2(0,\infty)},\ \ \ \ c>0, \ \ \ i=1,2.
\end{equation}
Now consider $X(z)$ as in \eqref{SH:13} and note that the kernel of $X_0(z)$ (the integral piece) is built out of the functions
\begin{equation*}
	X_0^{ij}(z|x,y)=\int_0^{\infty}n_i(w|x)k_j(w|y)\frac{\d w}{w-z},\ \ \ \ z\notin[0,\infty),\ \ \ (x,y)\in\mathbb{R}^2.
\end{equation*}
However, using \eqref{SH:16} and standard estimates for Airy functions, cf. \cite[$\S 9.7$]{NIST} we obtain by Cauchy-Schwarz
\begin{equation}\label{SH:17}
	\big|X_0^{ij}(z|x,y)\big|\leq\frac{c}{\textnormal{dist}(z,\mathbb{R}_{\geq 0})}\,|s+x|^{\frac{3}{4}}|s+y|^{\frac{3}{4}},
\end{equation}
uniformly in $(x,y)\in\mathbb{R}^2$. Thus $(x,y)\mapsto X_0(z|x,y)$ is in $L^2(\mathbb{R}^2,\d\sigma\otimes\d\sigma;\mathbb{C}^{2\times 2})$ for $z\in\mathbb{C}\setminus[0,\infty)$. Moreover, by the standard regularity properties of the resolvent operator, we also find from \eqref{SH:14} that $z\mapsto X_0^{ij}(z|x,y)$ is H\"older continuous on $\mathbb{C}\setminus[0,\infty)$ for any $(x,y)\in\mathbb{R}^2$ and $i,j=1,2$. Thus, by the Plemelj-Sokhostki formul\ae, $z\mapsto X_0(z|x,y)$ is analytic in $\mathbb{C}\setminus[0,\infty)$ for any $(x,y)\in\mathbb{R}^2$ and so the right hand side in \eqref{SH:13} analytic in $z\in\mathbb{C}\setminus[0,\infty)$ according to Definition \ref{anadef}. Having just established condition (1) in Problem \ref{HP8}, we now turn to (2), (3) and (4): First, estimate \eqref{SH:12} follows from \eqref{SH:17} and the Plemelj-Sokhotski formul\ae. Second, $z\mapsto X_0^{ij}(z|x,y)$ is H\"older continuous on $[0,\infty)$ for any $(x,y)\in\mathbb{R}^2$, thus $X(z)$ admits continuous boundary values $X_{\pm}(z)\in\mathcal{I}(\mathcal{H}_2)$ on $(0,\infty)$ according to Definition \ref{condef} by the Plemelj-Sokhotski formul\ae. Now verify \eqref{SH:9}, so first compute from \eqref{SH:13},
\begin{equation}\label{SH:18}
	X_+(z)-X_-(z)=2\pi\im\begin{bmatrix}N_1(z)\otimes K_1(z) & -N_1(z)\otimes K_2(z)\smallskip\\
	N_2(z)\otimes K_1(z) & -N_2(z)\otimes K_2(z)\end{bmatrix},\ \ \ \ \ z\in(0,\infty).
\end{equation}
On the other hand, by \eqref{SH:10} and \eqref{SH:13},
\begin{align}\label{SH:19}
	X_-(z)G(z)=&\,X_-(z)\left\{\mathbb{I}_2+2\pi\im\begin{bmatrix}M_1(z)\otimes K_1(z) & -M_1(z)\otimes K_2(z)\smallskip\\
	M_2(z)\otimes K_1(z) & -M_2(z)\otimes K_2(z)\end{bmatrix}\right\}\nonumber\\
	=&\,X_-(z)+2\pi\im\left\{\mathbb{I}_2+\int_0^{\infty}\begin{bmatrix}N_1(w)\otimes K_1(w) & -N_1(w)\otimes K_2(w)\smallskip\\
	N_2(w)\otimes K_1(w) & -N_2(w)\otimes K_2(w)\end{bmatrix}\frac{\d w}{w-z_-}\right\}\nonumber\\
	&\ \ \ \ \ \ \ \ \ \,\circ\begin{bmatrix}M_1(z)\otimes K_1(z) & -M_1(z)\otimes K_2(z)\smallskip\\
	M_2(z)\otimes K_1(z) & -M_2(z)\otimes K_2(z)\end{bmatrix},\ \ \ \ z\in(0,\infty),
\end{align}
and since, using here the general $(N_i\otimes K_j)(M_j\otimes K_{\ell})=\langle K_j,M_j\rangle_{\mathcal{H}_1}(N_i\otimes K_{\ell})$ fact,
\begin{align*}
	&\begin{bmatrix}N_1(w)\otimes K_1(w) & -N_1(w)\otimes K_2(w)\smallskip\\
	N_2(w)\otimes K_1(w) & -N_2(w)\otimes K_2(w)\end{bmatrix}\begin{bmatrix}M_1(z)\otimes K_1(z) & -M_1(z)\otimes K_2(z)\smallskip\\
	M_2(z)\otimes K_1(z) & -M_2(z)\otimes K_2(z)\end{bmatrix}\\
	&\ \ \ \ \ \ \ \ \,=(w-z)K_{\sigma,s}(w,z)\begin{bmatrix}N_1(w)\otimes K_1(z) & -N_1(w)\otimes K_2(z)\smallskip\\
	N_2(w)\otimes K_1(z) & -N_2(w)\otimes K_2(z)\end{bmatrix},\ \ \ (w,z)\in(0,\infty),\times(0,\infty),
\end{align*}
the chain of equalities \eqref{SH:19} simplifies further
\begin{align}\label{SH:20}
	X_-(z)G(z)=&\,X_-(z)+2\pi\im\begin{bmatrix}M_1(z)\otimes K_1(z) & -M_1(z)\otimes K_2(z)\smallskip\\
	M_2(z)\otimes K_1(z) & -M_2(z)\otimes K_2(z)\end{bmatrix}\\
	&\ \ \ \ \ \ \,+2\pi\im\int_0^{\infty}K_{\sigma,s}(w,z)\begin{bmatrix}N_1(w)\otimes K_1(z) & -N_1(w)\otimes K_2(z)\smallskip\\
	N_2(w)\otimes K_1(z) & -N_2(w)\otimes K_2(z)\end{bmatrix}\d w,\ \ z\in(0,\infty).\nonumber
\end{align}
But \eqref{SH:14} is equivalent to the operator equation
\begin{equation*}
	N_i(z)-\int_0^{\infty}K_{\sigma,s}(z,w)N_i(w)\,\d w=M_i(z),\ \ \ z\in(0,\infty),
\end{equation*}
which, by symmetry of $K_{\sigma,s}$, simplifies \eqref{SH:20} to
\begin{align*}
	X_-(z)G(z)=X_-(z)+2\pi\im\begin{bmatrix}N_1(z)\otimes K_1(z) & -N_1(z)\otimes K_2(z)\smallskip\\
	N_2(z)\otimes K_1(z) & -N_2(z)\otimes K_2(z)\end{bmatrix}\stackrel{\eqref{SH:18}}{=}X_+(z),\ \ \ z\in(0,\infty).
\end{align*}
Condition (2) in Problem \ref{HP8} has therefore been established and we are left with the singular behavior \eqref{SH:10} and \eqref{SH:11} near $z=0$. Fix $z\notin[0,\infty),0<\epsilon<1$ and rewrite \eqref{SH:13} as
\begin{align}\label{SH:200}
	X(z)&\,=\mathbb{I}_2+\ln\left(\frac{z-2\epsilon}{z}\right)\begin{bmatrix}N_1(z)\otimes K_1(z) & -N_1(z)\otimes K_2(z)\smallskip\\
	N_2(z)\otimes K_1(z) & -N_2(z)\otimes K_2(z)\end{bmatrix}\\
	&\,+\int_0^{2\epsilon}\left\{\begin{bmatrix}N_1(w)\otimes K_1(w) & -N_1(w)\otimes K_2(w)\smallskip\\
	N_2(w)\otimes K_1(w) & -N_2(w)\otimes K_2(w)\end{bmatrix}-\begin{bmatrix}N_1(z)\otimes K_1(z) & -N_1(z)\otimes K_2(z)\smallskip\\
	N_2(z)\otimes K_1(z) & -N_2(z)\otimes K_2(z)\end{bmatrix}\right\}\frac{\d w}{w-z}\nonumber\\
	&\,+\int_{2\epsilon}^{\infty}\begin{bmatrix}N_1(w)\otimes K_1(w) & -N_1(w)\otimes K_2(w)\smallskip\\
	N_2(w)\otimes K_1(w) & -N_2(w)\otimes K_2(w)\end{bmatrix}\frac{\d w}{w-z}.\nonumber
\end{align}
The remaining two integrals are operators in $\mathcal{I}(\mathcal{H}_2)$ which are analytic in the disk $|z|<\epsilon$ and the same is true for the operator multiplied by the logarithm: indeed we already showed that
\begin{equation}\label{SH:21}
	X_-(z)G(z)=X_-(z)+2\pi\im\begin{bmatrix}N_1(z)\otimes K_1(z) & -N_1(z)\otimes K_2(z)\smallskip\\
	N_2(z)\otimes K_1(z) & -N_2(z)\otimes K_2(z)\end{bmatrix},\ \ \ z\in(0,\infty),
\end{equation}
and since the trace class $G(z)\in\mathcal{I}(\mathcal{H}_2)$ is invertible with
\begin{equation*}
	\big(G(z)\big)^{-1}=\mathbb{I}_2-2\pi\im\begin{bmatrix}M_1(z)\otimes K_1(z) & -M_1(z)\otimes K_2(z)\smallskip\\
		M_2(z)\otimes K_1(z) & -M_2(z)\otimes K_2(z)\end{bmatrix},\ \ \ z\in\mathbb{R},
\end{equation*}
we also compute (as done previously for $X_-(z)G(z)$),
\begin{align}\label{SH:22}
	X_+(z)\big(G(z)\big)^{-1}=&\,X_+(z)\left\{\mathbb{I}_2-2\pi\im\begin{bmatrix}M_1(z)\otimes K_1(z) & -M_1(z)\otimes K_2(z)\smallskip\\
		M_2(z)\otimes K_1(z) & -M_2(z)\otimes K_2(z)\end{bmatrix}\right\}\nonumber\\
		=&\,X_+(z)-2\pi\im\left\{\mathbb{I}_2+\int_0^{\infty}\begin{bmatrix}N_1(w)\otimes K_1(w) & -N_1(w)\otimes K_2(w)\smallskip\\
	N_2(w)\otimes K_1(w) & -N_2(w)\otimes K_2(w)\end{bmatrix}\frac{\d w}{w-z_+}\right\}\nonumber\\
	&\ \ \ \ \ \ \ \ \ \,\circ\begin{bmatrix}M_1(z)\otimes K_1(z) & -M_1(z)\otimes K_2(z)\smallskip\\
	M_2(z)\otimes K_1(z) & -M_2(z)\otimes K_2(z)\end{bmatrix}\nonumber\\
	=&\,X_+(z)-2\pi\im\begin{bmatrix}N_1(z)\otimes K_1(z) & -N_1(z)\otimes K_2(z)\smallskip\\
	N_2(z)\otimes K_1(z) & -N_2(z)\otimes K_2(z)\end{bmatrix},\ \ z\in(0,\infty).
\end{align}
Hence, combining \eqref{SH:21} and \eqref{SH:22} with the explicit form \eqref{SH:9}, we deduce
\begin{equation*}
	X_{\pm}(z)\begin{bmatrix}M_1(z)\otimes K_1(z) & -M_1(z)\otimes K_2(z)\smallskip\\
		M_2(z)\otimes K_1(z) & -M_2(z)\otimes K_2(z)\end{bmatrix}=\begin{bmatrix}N_1(z)\otimes K_1(z) & -N_1(z)\otimes K_2(z)\smallskip\\
	N_2(z)\otimes K_1(z) & -N_2(z)\otimes K_2(z)\end{bmatrix},\ \ z\in(0,\infty),
\end{equation*}
and from this, for $z\in(0,\infty)$, the following operator composition formula for $N_i(z)$,
\begin{equation}\label{SH:23}
	N(z)=X(z)M(z);\ \ \ \ \ \ \ N(z):=\big[N_1(z),N_2(z)\big]^{\top},\ \ M(z):=\big[M_1(z),M_2(z)\big]^{\top},
\end{equation}
which is independent of the $\pm$ choice for the boundary values of $X(z)$ on $(0,\infty)$.  Now use \eqref{SH:23} back in \eqref{SH:200} and conclude that $N(z)$ near $z=0$ is of the form
\begin{equation*}
	N(z)=\big(\,\mathbb{I}_2+N_{\textnormal{reg}}(z)\big)M(z),\ \ \ \ \ \ \ \ \ \ \ \ N_{\textnormal{reg}}(z)\in\mathcal{I}(\mathcal{H}_2)\ \ \textnormal{is analytic at}\  z=0,
\end{equation*}
since $(N_i(z)\otimes K_1(z))M_1(z)-(N_i(z)\otimes K_2(z))M_2(z)=N_i(z)\otimes(M_1(z)K_1(z)-M_2(z)K_2(z))=0$, i.e. since the logarithmic term drops out. All together, we can expand $N_i(z)$ near $z=0$, which verifies the existence of the multiplication operators $F_i$ in \eqref{SH:10} with their indicated properties, and estimate the kernel
\begin{align*}
	\bigg|&\,\frac{1}{w-z}\left\{\begin{bmatrix}N_1(w)\otimes K_1(w) & -N_1(w)\otimes K_2(w)\smallskip\\
	N_2(w)\otimes K_1(w) & -N_2(w)\otimes K_2(w)\end{bmatrix}-\begin{bmatrix}N_1(z)\otimes K_1(z) & -N_1(z)\otimes K_2(z)\smallskip\\
	N_2(z)\otimes K_1(z) & -N_2(z)\otimes K_2(z)\end{bmatrix}\right\}(x,y)\bigg|\\
	&\,\ \ \ \ \ \ \ \ \,\leq c|1+s+x|^{\frac{3}{4}}|1+s+y|^{\frac{3}{4}},\ \ \ c>0,
\end{align*}
uniformly for $w\in[0,\epsilon]\subset\mathbb{R},|z|<\epsilon$ and any $(x,y)\in\mathbb{R}^2$. All combined in \eqref{SH:200} we have now established \eqref{SH:10} together with \eqref{SH:11}. This completes our proof.
\end{proof}
During our upcoming computations we will use the following consequences of Theorem \ref{tb:1}.
\begin{cor}\label{coooo:1} Let $X(z)$ denote the unique solution of Problem \ref{HP8}. Then
\begin{equation}\label{SH:24}
	\big(X(z)\big)^{-1}=\mathbb{I}_2-\int_0^{\infty}\begin{bmatrix}M_1(w)\otimes L_1(w) & -M_1(w)\otimes L_2(w)\smallskip\\
	M_2(w)\otimes L_1(w) & -M_2(w)\otimes L_2(w)\end{bmatrix}\frac{\d w}{w-z},\ \ \ z\in\mathbb{C}\setminus[0,\infty),
\end{equation}
where $L_i(w)$ are the integral operators on $\mathcal{H}_1$ with kernel $\ell_i(z|y)$ determined from the integral equation
\begin{equation}\label{SH:25}
	(1-K_{\sigma,s}\upharpoonright_{L^2(0,\infty)})\ell_i(\cdot|y)=k_i(\cdot|y),\ \ \ i=1,2,
\end{equation}
with $y\in\mathbb{R}$. Moreover, for any $z\in(0,\infty)$,
\begin{equation}\label{SH:26}
	N(z)=X(z)M(z),\ \ \ \ \ \ \ \ \ \ \ L(z)=K(z)\big(X(z)\big)^{-1}
\end{equation}
with $L(z):=\big[L_1(z),-L_2(z)\big], K(z):=\big[K_1(z),-K_2(z)\big]$ and near $z=0$,
\begin{equation*}
	N(z)=\big(\,\mathbb{I}_2+N_{\textnormal{reg}}(z)\big)M(z),\ \ \ \ \ \ \ \ \ \ L(z)=K(z)\big(\,\mathbb{I}_2+L_{\textnormal{reg}}(z)\big)
\end{equation*}
where $N_{\textnormal{reg}}(z)\in\mathcal{I}(\mathcal{H}_2)$ and $L_{\textnormal{reg}}(z)\in\mathcal{I}(\mathcal{H}_2)$ are analytic at $z=0$.
\end{cor}
\begin{proof} Let $Y(z)$ denote the right hand side of \eqref{SH:24} and $\overline{Y}_0(z)$ the underlying integrand without the Cauchy kernel. With \eqref{SH:13} we then compute ($\overline{X}_0(z)$ being similarly defined)
\begin{align*}
	X(z)Y(z)=\mathbb{I}_2+\int_0^{\infty}\big(\overline{X}_0(w)-\overline{Y}_0(w)\big)\frac{\d w}{w-z}-\int_0^{\infty}\int_0^{\infty}\overline{X}_0(w)\overline{Y}_0(\lambda)\frac{\d w}{w-z}\frac{\d\lambda}{\lambda-z},\ \ z\notin[0,\infty).
\end{align*}
But
\begin{equation*}
	\overline{X}_0(w)\overline{Y}_0(\lambda)=(w-\lambda)K_{\sigma,s}(w,\lambda)\begin{bmatrix}N_1(w)\otimes L_1(\lambda) & -N_1(w)\otimes L_2(\lambda)\smallskip\\
	N_2(w)\otimes L_1(\lambda) & -N_2(w)\otimes L_2(\lambda)\end{bmatrix},
\end{equation*}
so using partial fractions in the double integral and both integral equations \eqref{SH:14} and \eqref{SH:25}, 
\begin{equation}\label{symmetry}
	\int_0^{\infty}\big(\overline{X}_0(w)-\overline{Y}_0(w)\big)\frac{\d w}{w-z}-\int_0^{\infty}\int_0^{\infty}\overline{X}_0(w)\overline{Y}_0(\lambda)\frac{\d w}{w-z}\frac{\d\lambda}{\lambda-z}=0,\ \ \ z\notin[0,\infty),
\end{equation}
which shows that $X(z)Y(z)=\mathbb{I}_2$ for $z\notin[0,\infty)$. But any solution $X(z)\in\mathcal{I}(\mathcal{H}_2)$ of Problem \ref{HP8} is invertible, so a one-sided inverse is a two-sided inverse and thus identity \eqref{SH:24} holds for $z\notin[0,\infty)$. Moving ahead, the first half of \eqref{SH:26} and the associated local regularity statement have already been proven above. Focusing thus on the second half, the right hand side in \eqref{SH:24} is such that by the Plemelj-Sokhotski formula
\begin{equation}\label{SH:27}
	\big(X(z)\big)^{-1}_+-\big(X(z)\big)^{-1}_-=-2\pi\im\begin{bmatrix}M_1(z)\otimes L_1(z) & -M_1(z)\otimes L_2(z)\smallskip\\
	M_2(z)\otimes L_1(z) & -M_2(z)\otimes L_2(z)\end{bmatrix},\ \ \ z\in(0,\infty),
\end{equation}
as well as
\begin{align*}
	\big(G(z)\big)^{-1}\big(X(z)\big)^{-1}_-=&\,\left\{\mathbb{I}_2-2\pi\im\begin{bmatrix}M_1(z)\otimes K_1(z) & -M_1(z)\otimes K_2(z)\smallskip\\
		M_2(z)\otimes K_1(z) & -M_2(z)\otimes K_2(z)\end{bmatrix}\right\}\big(X(z)\big)^{-1}_-\\
		\stackrel{\eqref{SH:24}}{=}&\,\big(X(z)\big)^{-1}_--2\pi\im\begin{bmatrix}M_1(z)\otimes K_1(z) & -M_1(z)\otimes K_2(z)\smallskip\\
		M_2(z)\otimes K_1(z) & -M_2(z)\otimes K_2(z)\end{bmatrix}\\
		&\ \ \ \ \ \ \ \ \circ\left\{\mathbb{I}_2-\int_0^{\infty}\begin{bmatrix}M_1(w)\otimes L_1(w) & -M_1(w)\otimes L_2(w)\smallskip\\
	M_2(w)\otimes L_1(w) & -M_2(w)\otimes L_2(w)\end{bmatrix}\frac{\d w}{w-z_-}\right\}.
\end{align*}
But since
\begin{align*}
	&\,\begin{bmatrix}M_1(z)\otimes K_1(z) & -M_1(z)\otimes K_2(z)\smallskip\\
		M_2(z)\otimes K_1(z) & -M_2(z)\otimes K_2(z)\end{bmatrix}\begin{bmatrix}M_1(w)\otimes L_1(w) & -M_1(w)\otimes L_2(w)\smallskip\\
	M_2(w)\otimes L_1(w) & -M_2(w)\otimes L_2(w)\end{bmatrix}\\
	&\ \ \ \ \ \ \ \ 
	=(z-w)K_{\sigma,s}(z,w)\begin{bmatrix}M_1(z)\otimes L_1(w) & -M_1(z)\otimes L_2(w)\smallskip\\
	M_2(z)\otimes L_1(w) & -M_2(z)\otimes L_2(w)\end{bmatrix},\ \ \ (w,z)\in(0,\infty)\times(0,\infty),
\end{align*}
we find for any $z\in(0,\infty)$
\begin{equation}\label{SH:28}
	\big(G(z)\big)^{-1}\big(X(z)\big)^{-1}_-\stackrel{\eqref{SH:25}}{=}\big(X(z)\big)^{-1}_--2\pi\im\begin{bmatrix}M_1(z)\otimes L_1(z) & -M_1(z)\otimes L_2(z)\smallskip\\
	M_2(z)\otimes L_1(z) & -M_2(z)\otimes L_2(z)\end{bmatrix},
\end{equation}
and similarly
\begin{align*}
	G(z)\big(X(z)\big)^{-1}_+=&\,\left\{\mathbb{I}_2+2\pi\im\begin{bmatrix}M_1(z)\otimes K_1(z) & -M_1(z)\otimes K_2(z)\smallskip\\
		M_2(z)\otimes K_1(z) & -M_2(z)\otimes K_2(z)\end{bmatrix}\right\}\big(X(z)\big)^{-1}_+\\
\stackrel{\eqref{SH:24}}{=}&\,\,\big(X(z)\big)_+^{-1}+2\pi\im\begin{bmatrix}M_1(z)\otimes K_1(z) & -M_1(z)\otimes K_2(z)\smallskip\\
		M_2(z)\otimes K_1(z) & -M_2(z)\otimes K_2(z)\end{bmatrix}\\
		&\ \ \ \ \ \ \ \ \circ\left\{\mathbb{I}_2-\int_0^{\infty}\begin{bmatrix}M_1(w)\otimes L_1(w) & -M_1(w)\otimes L_2(w)\smallskip\\
	M_2(w)\otimes L_1(w) & -M_2(w)\otimes L_2(w)\end{bmatrix}\frac{\d w}{w-z_+}\right\}\\
	=&\,\,\big(X(z)\big)_+^{-1}+2\pi\im\begin{bmatrix}M_1(z)\otimes L_1(z) & -M_1(z)\otimes L_2(z)\smallskip\\
	M_2(z)\otimes L_1(z) & -M_2(z)\otimes L_2(z)\end{bmatrix},\ \ \ z\in(0,\infty).
\end{align*}
Comparing this last identity with \eqref{SH:28}, we deduce
\begin{equation*}
	\begin{bmatrix}M_1(z)\otimes K_1(z) & -M_1(z)\otimes K_2(z)\smallskip\\
	M_2(z)\otimes K_1(z) & -M_2(z)\otimes K_2(z)\end{bmatrix}\big(X(z)\big)^{-1}=\begin{bmatrix}M_1(z) \otimes L_1(z) & -M_1(z)\otimes L_2(z)\smallskip\\
	M_2(z)\otimes L_1(z) & -M_2(z)\otimes L_2(z)\end{bmatrix},\ z\in(0,\infty),
\end{equation*}
which is once more independent of the $\pm$ choice in $(X(z))^{-1}$. Reading the last operator equality entry wise we then find
\begin{equation*}
	\big[L_1(z),-L_2(z)\big]=\big[K_1(z),-K_2(z)\big]\big(X(z)\big)^{-1},\ \ z\in(0,\infty),
\end{equation*}
as claimed in \eqref{SH:26}. Finally, the expansion of $L(z)$ near $z=0$ follows as in the proof of Theorem \ref{tb:1} by splitting the formula \eqref{SH:24} as in \eqref{SH:200} and using that $K_1(z)(M_1(z)\otimes L_i(z))-K_2(z)(M_2(z)\otimes L_i(z))=(K_1(z)M_1(z)-K_2(z)M_2(z))\otimes L_i(z)=0$. This concludes our proof.
\end{proof}
Another useful consequence of Theorem \ref{tb:1} is the following result
\begin{cor} Let $(1-K_{\sigma,s}\upharpoonright_{L^2(0,\infty)})^{-1}=1+R_{\sigma,s}\upharpoonright_{L^2(0,\infty)}$, then the resolvent kernel is of the form
\begin{equation}\label{SH:29}
	(x-y)R_{\sigma,s}(x,y)=\int_{-\infty}^{\infty}\big(\ell_1(x|t)n_1(y|t)-\ell_2(x|t)n_2(y|t)\big)\,\d\sigma(t),\ \ \ \ x,y>0
\end{equation}
with the functions defined in \eqref{SH:14} and \eqref{SH:25}. In particular, for any $x>0$,
\begin{equation*}
 	\int_{-\infty}^{\infty}\big(\ell_1(x|t)n_1(x|t)-\ell_2(x|t)n_2(x|t)\big)\,\d\sigma(t)=0.
\end{equation*}
\end{cor}
\begin{proof} The representation formula \eqref{SH:29} is a general fact of integrable operators, cf. \cite[$(6),(7)$]{D0}. In our case the finite sum in \cite[$(1)$]{D0} is naturally replaced by a weighted integral but this has no significant impact on the formula for the resolvent kernel. In detail, since $K_{\sigma,s}$ is a Hilbert-Schmidt operator on $L^2(0,\infty)$, compare Lemma \ref{L:1}, so is $R_{\sigma,s}=(1-K_{\sigma,s})^{-1}K_{\sigma,s}$ and therefore $R_{\sigma,s}$ has a kernel $R_{\sigma,s}(x,y)$. Next, with $M$ denoting multiplication by the independent variable, we check that the commutators $[M,K_{\sigma,s}]$, resp. $[M,R_{\sigma,s}]$ have kernels $(x-y)K_{\sigma,s}(x,y)$, resp. $(x-y)R_{\sigma,s}(x,y)$. On the other hand, on $L^2(0,\infty)$,
\begin{equation*}
	\big[M,R_{\sigma,s}\big]=\big[M,(1-K_{\sigma,s})^{-1}-1\big]=\big[M,(1-K_{\sigma,s})^{-1}\big]=(1-K_{\sigma,s})^{-1}\big[M,K_{\sigma,s}\big](1-K_{\sigma,s})^{-1},
\end{equation*}
which yields for the underlying kernels, by \eqref{opdef}, \eqref{SH:14} and \eqref{SH:25}, 
\begin{equation*}
	(x-y)R_{\sigma,s}(x,y)=\int_{-\infty}^{\infty}\big(\ell_1(x|t)n_1(y|t)-\ell_2(x|t)n_2(y|t)\big)\,\d\sigma(t),\ \ x,y>0.
\end{equation*}
\end{proof}
The last Corollary concludes our content on the operator-valued Problem \ref{HP8}. We will now use this problem in the derivation of an integrable system for the Fredholm determinant $F_{\sigma}(s)$.
\subsection{An operator-valued Lax pair} From now on it will be convenient, but not essential, to view the aforementioned multiplication operators $M_i(z)$, see Definition \ref{frank}, and $N_i(w)$, see Theorem \ref{tb:1}, as integral operators on $\mathcal{H}_1$ with appropriate distributional kernels. To the point, we replace
\begin{align*}
	m_i(z|x)\,\,\,\mapsto &\,\,\,\,\,m_i(z|x,y):=m_i(z|x)\delta(x-y)(\sigma'(y))^{-1},\\
	n_i(w|x)\,\,\,\mapsto &\,\,\,\,\,n_i(w|x,y):=n_i(w|x)\delta(x-y)(\sigma'(y))^{-1},
\end{align*}
for $(x,y)\in\mathbb{R}^2$ with the (positive) Radon-Nikodym derivative $\sigma'(t):=\frac{\d\sigma}{\d t}$ and where, by definition,
\begin{equation*}
	\int_{-\infty}^{\infty}\delta(x-y)(\sigma'(y))^{-1}f(y)\,\d\sigma(y):=f(x),\ \ \ \ \ \ \ \ f\in L^2(\mathbb{R},\d\sigma).
\end{equation*}
Subject to this convention we then compute, recall \eqref{SH:23},
\begin{equation*}
	M_z(z|x,y)=\begin{bmatrix}0 & z+s+x\smallskip\\
	1 & 0\end{bmatrix}M(z|x,y),\ \ \ \ \ M_z(z|x,y)\equiv\frac{\partial }{\partial z}M(z|x,y),
\end{equation*}
which is the standard Airy system for the kernels and which translates into the integral operator identity
\begin{equation}\label{SH:30}
	M_z(z)=\big(zA_1+A_2\big)M(z),\ \ \ z\in\mathbb{R}.
\end{equation}
Here, $A_i$ are $z$-independent integral operators on $\mathcal{H}_2$ with distributional kernels
\begin{equation*}
	A_1(x,y):=\delta(x-y)\begin{bmatrix}0 & 1\\ 0 & 0\end{bmatrix}(\sigma'(y))^{-1},\ \ \ \ \ A_2(x,y):=\delta(x-y)\begin{bmatrix}0 & s+x\\ 1 & 0\end{bmatrix}(\sigma'(y))^{-1}.
\end{equation*}
Now return to Corollary \ref{coooo:1} and $z$-differentiate the first identity in \eqref{SH:26},
\begin{equation*}
	\frac{\partial N}{\partial z}(z)\stackrel{\eqref{SH:30}}{=}\Big(X_z(z)(X(z))^{-1}+X(z)\big(zA_1+A_2\big)(X(z))^{-1}\Big)N(z),\ \ z\in\mathbb{C}\setminus[0,\infty),
\end{equation*}
where $X(z)$ is the unique solution of Problem \ref{HP8}. As it turns out the right hand side in the last equation gives rise to an integral operator in $\mathcal{I}(\mathcal{H}_2)$ which is analytic for $z\neq 0$.
\begin{prop}\label{central} Define the integral operator
\begin{equation}\label{SH:300}
	Y(z):=X_z(z)(X(z))^{-1}+X(z)\big(zA_1+A_2\big)(X(z))^{-1},\ \ z\in\mathbb{C}\setminus[0,\infty)
\end{equation}
on $\mathcal{H}_2$. Then $Y(z)\in\mathcal{I}(\mathcal{H}_2)$ is analytic for $z\in\mathbb{C}\setminus\{0\}$.
\end{prop}
\begin{proof} Note that $Y(z)\in\mathcal{I}(\mathcal{H}_2)$ as consequence of Theorem \ref{tb:1}, Corollary \ref{coooo:1} and a direct computation of the kernel of $X(z)(zA_1+A_2)(X(z))^{-1}$. Moreover, $Y(z)$ admits continuous boundary values $Y_{\pm}(z)\in\mathcal{I}(\mathcal{H}_2)$ by consequence of Theorem \ref{tb:1} and Corollary \ref{coooo:1}, and we compute them directly,
\begin{align}\label{SH:31}
	Y_+(z)=\Big(X_{z-}(z)G(z)+&\,X_-(z)G_z(z)\Big)(G(z))^{-1}(X_-(z))^{-1}\\
	&\,+X_-(z)G(z)\big(zA_1+A_2\big)(G(z))^{-1}(X_-(z))^{-1},\ \ \ z>0.\nonumber
\end{align}
However, from \eqref{SH:9}, we derive the kernel identity
\begin{equation*}
	G_z(x|x,y)=\int_{-\infty}^{\infty}\Big\{\big(zA_1(x,t)+A_2(x,t)\big)G_0(z|t,y)-G_0(z|x,t)\big(zA_1(t,y)+A_2(t,y)\big)\Big\}\,\d\sigma(t)
\end{equation*}
with $G(z)=\mathbb{I}_2+G_0(z)$ and this leads us to the operator commutator identity
\begin{equation}\label{SH:32}
	G_z(z)=\big[zA_1+A_2,G(z)\big]\in\mathcal{I}(\mathcal{H}_2).
\end{equation}
Inserting \eqref{SH:32} into \eqref{SH:31}, after simplification,
\begin{equation*}
	Y_+(z)=X_{z-}(z)(X_-(z))^{-1}+X_-(z)\big(zA_1+A_2\big)(X_-(z))^{-1}=Y_-(z),\ \ \ z>0,
\end{equation*}
so $z\mapsto Y(z|x,y)$ is in fact continuous across $(0,\infty)\ni z$ for any $(x,y)\in\mathbb{R}^2$. But $(x,y)\mapsto Y_{\pm}(z|x,y)$ is also in $L^2(\mathbb{R}^2,\d\sigma\otimes\d\sigma;\mathbb{C}^{2\times 2})$ for any $z\in(0,\infty)$, so $Y(z)$ is analytic for $z\in\mathbb{C}\setminus\{0\}$ according to Definition \ref{anadef}. This completes our proof.
\end{proof}
Summarizing our computations so far, we have obtained the operator-valued ODE system
\begin{equation*}
	\frac{\partial N}{\partial z}(z)=Y(z)N(z),\ \ z\in\mathbb{C}\setminus\{0\},
\end{equation*}
where $Y(z)$ is analytic for $z\in\mathbb{C}\setminus\{0\}$. This is the analogue of Plemelj's \eqref{e:22} in our context and we will now evaluate the ``coefficient" operator $Y(z)$ from Liouville's theorem. First, return to the explicit formula \eqref{SH:13}, write $\frac{1}{w-z}=-\frac{1}{z}+\frac{w}{z(w-z)}$ and obtain
\begin{equation*}
	X(z)=\mathbb{I}_2-\frac{1}{z}\int_0^{\infty}\begin{bmatrix}N_1(w)\otimes K_1(w) & -N_1(w)\otimes K_2(w)\smallskip\\
	N_2(w)\otimes K_1(w) & -N_2(w)\otimes K_2(w)\end{bmatrix}\d w+X_{\textnormal{asy}}(z),\ \ \ z\in\mathbb{C}\setminus[0,\infty),
\end{equation*}
where $X_{\textnormal{asy}}(z)\in\mathcal{I}(\mathcal{H}_2)$ and for any compact $K\subset\mathbb{C}$ with $\textnormal{dist}(0,K)>0$ there exists $C>0$ such that for $z\in\mathbb{C}\setminus K$,
\begin{equation}\label{SH:33}
	\|X_{\textnormal{asy}}(z|x,y)\|\leq \frac{C}{1+|z|^2}|s+x|^{\frac{7}{4}}|s+y|^{\frac{7}{4}}
\end{equation}
uniformly in $(x,y)\in\mathbb{R}^2$. Similarly, from \eqref{SH:24},
\begin{equation*}
	(X(z))^{-1}=\mathbb{I}_2+\frac{1}{z}\int_0^{\infty}\begin{bmatrix}M_1(w)\otimes L_1(w) & -M_1(w)\otimes L_2(w)\smallskip\\
	M_2(w)\otimes L_1(w) & -M_2(w)\otimes L_2(w)\end{bmatrix}\d w+\widehat{X}_{\textnormal{asy}}(z),\ \ z\in\mathbb{C}\setminus[0,\infty),
\end{equation*}
and $\widehat{X}_{\textnormal{asy}}(z)\in\mathcal{I}(\mathcal{H}_2)$ behaves similarly to \eqref{SH:33}. Inserting these two formul\ae\,for $X(z)$ and $(X(z))^{-1}$ into \eqref{SH:300} we obtain the exact operator identity
\begin{align*}
	Y(z)=zA_1+&\,A_2+A_1\int_0^{\infty}\begin{bmatrix}M_1(w)\otimes L_1(w) & -M_1(w)\otimes L_2(w)\smallskip\\
	M_2(w)\otimes L_1(w) & -M_2(w)\otimes L_2(w)\end{bmatrix}\d w\\
	&-\int_0^{\infty}\begin{bmatrix}N_1(w)\otimes K_1(w) & -N_1(w)\otimes K_2(w)\smallskip\\
	N_2(w)\otimes K_1(w) & -N_2(w)\otimes K_2(w)\end{bmatrix}\d w\,A_1+Y_{\infty}(z),\ \ \ z\in\mathbb{C}\setminus[0,\infty),
\end{align*}
with $Y_{\infty}(z)\in\mathcal{I}(\mathcal{H}_2)$. Observe that for any compact $K\subset\mathbb{C}$ with $\textnormal{dist}(0,K)>0$ there exists $C>0$ such that for $z\in\mathbb{C}\setminus K$,
\begin{equation}\label{SH:34}
	\|Y_{\infty}(z|x,y)\|\leq \frac{C}{1+|z|}|s+x|^{\frac{7}{4}}|s+y|^{\frac{7}{4}},\ \ \ \ \ \forall\,(x,y)\in\mathbb{R}^2.
\end{equation}
\begin{lem}\label{symlem} We have the finite rank integral operator identity
\begin{equation*}
	\int_0^{\infty}M_i(w)\otimes L_j(w)\,\d w=\int_0^{\infty}N_i(w)\otimes K_j(w)\,\d w,
\end{equation*}
valid for any $i,j=1,2$.
\end{lem}
\begin{proof} This follows from \eqref{symmetry} by multiplying through with $z$ and sending $z\rightarrow\infty$ with $\Im z>0$, say. The double integral term vanishes in this limit by the dominated convergence theorem and the first term yields the stated identity by the same reasoning.
\end{proof}
The symmetry Lemma \ref{symlem} allows us to simplify our expression for $Y(z)$, we have now
\begin{equation}\label{SH:35}
	Y(z)=zA_1+A_0+Y_{\infty}(z),\ \ \ \ z\in\mathbb{C}\setminus[0,\infty),
\end{equation}
where $A_0$ is the integral operator on $\mathcal{H}_2$ with kernel $A_0(x,y)=\big[A_0^{ij}(x,y)\big]_{i,j=1}^2$ and
\begin{align*}
	A_0^{11}(&\,x,y)=\int_0^{\infty}\big(N_2(w)\otimes K_1(w)\big)(x,y)\,\d w=-A_0^{22}(x,y),\ \ \ \ A_0^{21}(x,y)=\delta(x-y)(\sigma'(y))^{-1},\\
	&A_0^{12}(x,y)=\delta(x-y)(s+x)(\sigma'(y))^{-1}-\int_0^{\infty}\big(N_2(w)\otimes K_2(w)+N_1(w)\otimes K_1(w)\big)(x,y)\,\d w.
\end{align*}
The representation \eqref{SH:35} mimics a Poincar\'e asymptotic expansion of $Y(z)$ at $z=\infty$. We will now combine it with a corresponding representation at $z=0$ which is computed from \eqref{SH:200} and its counterpart for \eqref{SH:24} when inserted into \eqref{SH:300}. The result equals
\begin{equation}\label{SH:36}
	Y(z)=-\frac{1}{z}\begin{bmatrix}N_1(0)\otimes L_1(0) & -N_1(0)\otimes L_2(0)\smallskip\\
	N_2(0) \otimes L_1(0) & -N_2(0)\otimes L_2(0)\end{bmatrix}+Y_{\textnormal{ori}}(z),\ \ \ 0<|z|<\epsilon,
\end{equation}
with $Y_{\textnormal{ori}}(z)\in\mathcal{I}(\mathcal{H}_2)$ analytic at $z=0$ and where we have used the identity
\begin{align*}
	& \begin{bmatrix} N_1(w)\otimes K_1(w) & -N_1(w)\otimes K_2(w)\smallskip\\
	 N_2(w)\otimes K_1(w) & -N_2(w)\otimes K_2(w)\end{bmatrix}\begin{bmatrix}M_1(\lambda)\otimes L_1(\lambda) & -M_1(\lambda)\otimes L_2(\lambda)\smallskip\\
	 M_2(\lambda)\otimes L_1(\lambda) & -M_2(\lambda)\otimes L_2(\lambda)\end{bmatrix}\\
	 &\ \ \ \ \ =(w-\lambda)K_{\sigma,s}(w,\lambda)\begin{bmatrix}N_1(w)\otimes L_1(\lambda) & -N_1(w)\otimes L_2(\lambda)\smallskip\\
	 N_2(w)\otimes L_1(\lambda) & -N_2(w)\otimes L_2(\lambda)\end{bmatrix},\ \ \ w,\lambda>0,
\end{align*}
to ensure that all logarithmic terms cancel out in our computation of \eqref{SH:36}. At this point we combine Proposition \ref{central} with the representations \eqref{SH:35}, \eqref{SH:36} and deduce from Liouville's theorem that
\begin{equation}\label{SH:37}
	Y(z)=zA_1+A_0+\frac{1}{z}A_{-1},\ \ \ \ z\in\mathbb{C}\setminus\{0\},
\end{equation}
where $A_{-1}\in\mathcal{I}(\mathcal{H}_2)$ has kernel $A_{-1}(x,y)=\big[A_{-1}^{ij}(x,y)\big]_{i,j=1}^2$ with
\begin{equation*}
	A_{-1}^{ij}(x,y)=(-1)^j\big(N_i(0)\otimes L_j(0)\big)(x,y),\ \ \ (x,y)\in\mathbb{R}^2.
\end{equation*}
\begin{rem}\label{useful} The operator representation \eqref{SH:37} is the analogue of \cite[$(9.4.28)$]{Irev} in the derivation of the Tracy-Widom Painlev\'e-II formula \cite[$(1.17)$]{TW} through matrix-valued Riemann-Hilbert techniques. And, as in \cite[$(9.4.28)$]{Irev}, the operator $A_{-1}$ is nilpotent as a consequence of \eqref{SH:29},
\begin{equation*}
	(A_{-1})^2=0,\ \ \ \ \tr_{\mathcal{H}_2}A_{-1}=0.
\end{equation*}
Furthermore by \eqref{SH:25}, \eqref{SH:14}, for any $z>0$,
\begin{align*}
	L_2(z)=&\,\big(1-K_{\sigma,s}\upharpoonright_{L^2(0,\infty)}\big)^{-1}K_2(z)\stackrel{\eqref{opdef}}{=}\big(1-K_{\sigma,s}\upharpoonright_{L^2(0,\infty)}\big)^{-1}M_1(z)=N_1(z),\\
	N_2(z)=&\,\big(1-K_{\sigma,s}\upharpoonright_{L^2(0,\infty)}\big)^{-1}M_2(z)\stackrel{\eqref{opdef}}{=}\big(1-K_{\sigma,s}\upharpoonright_{L^2(0,\infty)}\big)^{-1}K_1(z)=L_1(z),
\end{align*}
and which should be compared with the parametrization of the nilpotent matrix in \cite[$(9.4.30)$]{Irev}.
\end{rem}
Before continuing, let us summarize our computations as they will form the first half of the anticipated Lax pair.
\begin{prop}\label{Laxp:1} Let $N(z),z\in(0,\infty)$ denote the integral operator on $\mathcal{H}_2$ as in \eqref{SH:26}. Then 
\begin{equation}\label{SH:38}
	\frac{\partial N}{\partial z}(z)=A(z)N(z),\ \ \ \ A(z)\equiv A(z;s)=zA_1+A_0+\frac{1}{z}A_{-1},\ \ z\in\mathbb{C}\setminus\{0\}
\end{equation}
where the integral operators $A_k=\big[A_k^{ij}\big]_{i,j=1}^2$ on $\mathcal{H}_2$ have the kernels $A_k^{ij}(x,y)$ stated above.
\end{prop}
The second half of the Lax pair
\begin{equation*}
	\frac{\partial N}{\partial s}(z)=B(z)N(z)
\end{equation*}
is computed in the same fashion, throughout replacing $z$-derivatives with $s$-derivatives. The main change occurs near $z=0$ since the analogue of \eqref{SH:300} with $s$-partial derivatives is analytic at $z=0$. We thus have the simpler result
\begin{prop}\label{Laxp:2} Let $N(z),z\in(0,\infty)$ denote the integral operator on $\mathcal{H}_2$ as in \eqref{SH:26}. Then 
\begin{equation}\label{SH:39}
	\frac{\partial N}{\partial s}(z)=B(z)N(z),\ \ \ \ B(z)\equiv B(z;s)=zB_1+B_0,\ \ z\in\mathbb{C}
\end{equation}
where the integral operators $B_k$ on $\mathcal{H}_2$ equal $B_1:=A_1$ and $B_0:=A_0$.
\end{prop}
\begin{rem} Our derivation of \eqref{SH:38} and \eqref{SH:39} is inspired by the non-rigorous computations in \cite[Section $5$]{KKS} and in \cite[Chapter XV.$6$]{KBI}.
\end{rem}
Propositions \ref{Laxp:1} and \ref{Laxp:2} complete our derivation of the operator-valued Lax pair. We now proceed and show that this pair indeed produces an integrable system for the Fredholm determinant $F_{\sigma}(s)$, exactly the integro-differential Painlev\'e-II equation of Amir-Corwin-Quastel.
\subsection{Derivation of the integro-differential equation} We first establish a connection between $F_{\sigma}(s)$ and Problem \ref{HP8}.
\begin{prop} For any $s\in\mathbb{R}$,
\begin{equation}\label{SH:40}
	\frac{\d}{\d s}\ln F_{\sigma}(s)=\tr_{\mathcal{H}_1}\int_0^{\infty}N_2(w)\otimes K_1(w)\,\d w=-\tr_{\mathcal{H}_1}A_0^{22}
\end{equation}
in terms of the integral operator $A_0$ back in \eqref{SH:35}.
\end{prop}
\begin{proof} We begin with the standard Jacobi identity
\begin{equation*}
	\frac{\d}{\d s}\ln F_{\sigma}(s)=-\tr_{L^2(0,\infty)}\bigg(\big(1-K_{\sigma,s}\upharpoonright_{L^2(0,\infty)}\big)^{-1}\frac{\d K_{\sigma,s}}{\d s}\bigg),
\end{equation*}
and compute the kernel derivative
\begin{equation*}
	\frac{\d K_{\sigma,s}}{\d s}(x,y)\stackrel{\eqref{SH:4}}{=}-\int_{-\infty}^{\infty}\textnormal{Ai}(x+s+t)\textnormal{Ai}(y+s+t)\,\d\sigma(t)\stackrel{\eqref{opdef}}{=}-\int_{-\infty}^{\infty}m_2(x|t)k_1(y|t)\d\sigma(t).
\end{equation*}
Hence,
\begin{eqnarray*}
	\frac{\d}{\d s}\ln F_{\sigma}(s)&=&-\int_0^{\infty}\int_0^{\infty}\big(1-K_{\sigma,s}\upharpoonright_{L^2(0,\infty)}\big)^{-1}(x,y)\frac{\d K_{\sigma,s}}{\d s}(y,x)\,\d x\,\d y\\
	&\stackrel{\eqref{SH:14}}{=}&\int_{-\infty}^{\infty}\left[\int_0^{\infty}\big(N_2(x)\otimes K_1(x)\big)(t,t)\,\d x\right]\d\sigma(t)=\tr_{\mathcal{H}_1}\int_0^{\infty}N_2(x)\otimes K_1(x)\,\d x,
\end{eqnarray*}
and the second equality in \eqref{SH:40} follows from the explicit formul\ae\,for the entries of $A_0$.
\end{proof}
At this point we use the Lax system \eqref{SH:38}, \eqref{SH:39}, i.e.
\begin{equation*}
	\frac{\partial N}{\partial z}(z)=A(z)N(z),\ \ \ \ \ \frac{\partial N}{\partial s}(z)=B(z)N(z)
\end{equation*}
and write out its compatibility condition
\begin{equation*}
	AB-BA=\frac{\partial B}{\partial z}-\frac{\partial A}{\partial s},
\end{equation*}
keeping in mind that the entries of $A$ and $B$ are operators and thus do not commute in general. The result equals
\begin{equation*}
	(A_0)_s=A_1+A_1A_{-1}-A_{-1}A_1,\ \ \ \ \ \ \ \ \ (A_{-1})_s=A_0A_{-1}-A_{-1}A_0,
\end{equation*}
or more explicit, using partial information for the $A_0^{ij}$ operators ($A_1^{12}=A_0^{21}=\mathbb{I}_1,A_0^{22}=-A_0^{11}$),
\begin{equation}\label{SH:41}
	(A_0^{22})_s=-A_{-1}^{21},\ \ \ \ \ \ \ \ \ (A_0^{12})_s=\mathbb{I}_1+A_{-1}^{22}-A_{-1}^{11},
\end{equation}
which follows from the equation on $(A_0)_s$ and from the equation on $(A_{-1})_s$,
\begin{equation}\label{SH:42}
	\begin{cases}(A_{-1}^{11})_s=-A_0^{22}A_{-1}^{11}+A_{-1}^{11}A_0^{22}+A_0^{12}A_{-1}^{21}-A_{-1}^{12},&\smallskip\\
	 (A_{-1}^{21})_s=A_{-1}^{11}-A_{-1}^{22}+A_0^{22}A_{-1}^{21}+A_{-1}^{21}A_0^{22},&\smallskip\\
	(A_{-1}^{12})_s=-A_0^{22}A_{-1}^{12}-A_{-1}^{12}A_0^{22}+A_0^{12}A_{-1}^{22}-A_{-1}^{11}A_0^{12},&\smallskip\\
	(A_{-1}^{22})_s=A_0^{22}A_{-1}^{22}-A_{-1}^{22}A_0^{22}-A_{-1}^{21}A_0^{12}+A_{-1}^{12}.&\smallskip
	\end{cases}
\end{equation}
The coupled system \eqref{SH:41}, \eqref{SH:42} is more complicated than its analogue \cite[$(9.4.34)$]{Irev} because of the aforementioned non-commutativity, still both systems share certain similarities. For instance
\begin{prop} The integral operator
\begin{equation*}
	I:=A_0^{12}+A_{-1}^{21}+(A_0^{22})^2-A_2^{12}
\end{equation*}
on $\mathcal{H}_1$ is a conserved quantity for the system \eqref{SH:41}, \eqref{SH:42}. In fact, $I=0$ is the zero integral operator on $\mathcal{H}_1$

\end{prop}
\begin{proof} We have
\begin{equation*}
	I_s=(A_0^{12})_s+(A_{-1}^{21})_s+(A_0^{22})_sA_0^{22}+A_0^{22}(A_0^{22})_s-(A_2^{12})_s=0
\end{equation*}
by \eqref{SH:41}, \eqref{SH:42} and with $(A_2^{12})_s=\mathbb{I}_1$. Next, for any $(x,y)\in\mathbb{R}^2$,
\begin{equation*}
	I(x,y)=-n_2(0|x)\ell_1(0|y)+\int_{-\infty}^{\infty}\left[\int_0^{\infty}n_2(w|x)k_1(w|t)\,\d w\right]\left[\int_0^{\infty}n_2(\lambda|t)k_1(\lambda|y)\,\d\lambda\right]\d\sigma(t),
\end{equation*}
so letting $s\rightarrow+\infty$ with the help of \eqref{opdef} and \eqref{SH:14}, using a Neumann series argument for $n_i(\cdot|x)$ we find $\lim_{s\rightarrow+\infty}I(x,y)=0$ pointwise in $(x,y)\in\mathbb{R}^2$ by the dominated convergence theorem. But $I$ was already shown to be $s$-independent, thus it must be the zero operator as claimed.
\end{proof}
At this point we $s$-differentiate the equation for $A_{-1}^{21}$ in \eqref{SH:42} one more time and use \eqref{SH:41}, the equations for $(A_{-1}^{11})_s, (A_{-1}^{22})_s, (A_{-1}^{21})_s$ from \eqref{SH:42} as well as the above integral of motion in the ensuing simplification. The result equals
\begin{equation}\label{SH:43}
	(A_{-1}^{21})_{ss}=2\Big(A_{-1}^{11}A_0^{22}-A_0^{22}A_{-1}^{22}-A_{-1}^{12}-2\big(A_{-1}^{21}\big)^2+A_0^{22}A_{-1}^{21}A_0^{22}\Big)+A_2^{12}A_{-1}^{21}+A_{-1}^{21}A_2^{12},
\end{equation}
and we now introduce (recall our notational convention to suppress all $s$-dependence throughout)
\begin{equation}\label{SH:44}
	u(s|t):=n_2(0|t)\stackrel{\eqref{SH:14}}{=}\big((1-K_{\sigma,s}\upharpoonright_{L^2(0,\infty)})^{-1}m_2(\cdot|t)\big)(0),\ \ \ \ s,t\in\mathbb{R}
\end{equation}
which will ultimately solve the anticipated integro-differential Painlev\'e-II equation. In order to get there, evaluate all kernels in \eqref{SH:43} on their diagonal, i.e. on the left hand side of \eqref{SH:43} we obtain
\begin{eqnarray}\label{SH:45}
	(A_{-1}^{21})_{ss}(t,t)&=&-\big(N_2(0)\otimes L_1(0)\big)_{ss}(t,t)=-\big(N_2(0)\otimes N_2(0)\big)_{ss}(t,t)\nonumber\\
	&\stackrel{\eqref{SH:44}}{=}&-\big(u^2(s|t)\big)_{ss}=-2\big(u_s(s|t)\big)^2-2u(s|t)u_{ss}(s|t)
\end{eqnarray}
through the use of Remark \ref{useful}. On the right hand side of \eqref{SH:43} instead,
\begin{equation}\label{SH:46}
	\Big(A_2^{12}A_{-1}^{21}+A_{-1}^{21}A_2^{12}-4\big(A_{-1}^{21}\big)^2\Big)(t,t)=-2(s+t)u^2(s|t)-4u^2(s|t)\int_{-\infty}^{\infty}u^2(s|\lambda)\d\sigma(\lambda),
\end{equation}
and
\begin{align}\label{SH:47}
	\Big(A_{-1}^{11}A_0^{22}&\,-A_0^{22}A_{-1}^{22}-A_{-1}^{12}+A_0^{22}A_{-1}^{21}A_0^{22}\Big)(t,t)=\Big(2A_{-1}^{11}A_0^{22}-A_{-1}^{12}+A_0^{22}A_{-1}^{21}A_0^{22}\Big)(t,t)\nonumber\\
	=&\,-\left[n_1(0|t)+\int_{-\infty}^{\infty}n_2(0|u)A_0^{22}(u,t)\,\d\sigma(u)\right]^2
\end{align}
where we used the following fact in the first and second equality (as well as Remark \ref{useful}):
\begin{prop}\label{flip} We have for any $(x,y)\in\mathbb{R}^2$ that $A_0^{22}(x,y)=A_0^{22}(y,x)$.
\end{prop}
\begin{proof} Simply compute
\begin{align*}
	A_0^{22}(x,y)=&\,-\int_0^{\infty}n_2(w|x)k_1(w|y)\,\d w\stackrel{\eqref{SH:14}}{=}-\int_0^{\infty}\Big(\big(1-K_{\sigma,s}\upharpoonright_{L^2(0,\infty)}\big)^{-1}m_2(w)\Big)(x)k_1(w|y)\,\d w\\
	\stackrel{\eqref{opdef}}{=}&\,-\int_0^{\infty}\Big(\big(1-K_{\sigma,s}\upharpoonright_{L^2(0,\infty)}\big)^{-1}k_1(w)\Big)(x)m_2(w|y)\,\d w\stackrel{\eqref{SH:25}}{=}-\int_0^{\infty}\ell_1(w|x)m_2(w|y)\,\d w\\
	=&\,-\int_0^{\infty}\big(M_2(w)\otimes L_1(w)\big)(y,x)\,\d w=A_0^{22}(y,x),
\end{align*}
with Lemma \ref{symlem} in the last equality.
\end{proof}
However, returning to the equation for $(A_{-1}^{21})_s$ in \eqref{SH:42}, its kernel evaluated on the diagonal yields (we use again Proposition \ref{flip} and Remark \ref{useful}),
\begin{equation*}
	-2u(s|t)u_s(s|t)=(A_{-1}^{21})_s(t,t)=-2u(s|t)\left(n_1(0|t)+\int_{-\infty}^{\infty}n_2(0|u)A_0^{22}(u,t)\,\d\sigma(u)\right),
\end{equation*}
and since the zeros of the Airy function are discrete and the resolvent analytic in $s\in\mathbb{R}$, therefore
\begin{equation}\label{SH:48}
	u_s(s|t)=n_1(0|t)+\int_{-\infty}^{\infty}n_2(0|u)A_0^{22}(u,t)\,\d\sigma(u).
\end{equation}
At this point we are left to substitute \eqref{SH:45}, \eqref{SH:46}, \eqref{SH:47} and \eqref{SH:48} into \eqref{SH:43},
\begin{equation*}
	u_{ss}(s|t)=\left[s+t+2u(s|t)\int_{-\infty}^{\infty}u^2(s|\lambda)\,\d\sigma(\lambda)\right]u(s|t),\ \ \ s,t\in\mathbb{R}.
\end{equation*}
This is exactly the integro-differential Painlev\'e-II equation of Amir, Corwin, Quastel, cf. \cite[Proposition $1.2$]{ACQ}, subject to the boundary data
\begin{equation}\label{SH:49}
	u(s|t)\stackrel{\eqref{SH:44}}{=}\big((1-K_{\sigma,s}\upharpoonright_{L^2(0,\infty)})^{-1}m_2(\cdot|t)\big)(0)\sim m_2(0|t)\stackrel{\eqref{opdef}}{=}\textnormal{Ai}(s+t),\ \ s\rightarrow+\infty,
\end{equation}
since $K_{\sigma,s}\rightarrow 0$ in operator norm as $s\rightarrow\infty$. We are now left to derive the relevant Tracy-Widom formula for $F_{\sigma}(s)$. To this end $s$-differentiate \eqref{SH:40}
\begin{equation*}
	\frac{\d^2}{\d s^2}\ln F_{\sigma}(s)=-\tr_{\mathcal{H}_1}(A_0^{22})_s\stackrel{\eqref{SH:41}}{=}\tr_{\mathcal{H}_1}A_{-1}^{21}=-\int_{-\infty}^{\infty}\big(n_2(0|t)\big)^2\,\d\sigma(t)\stackrel{\eqref{SH:44}}{=}-\int_{-\infty}^{\infty}u^2(s|t)\,\d\sigma(t),
\end{equation*}
and now integrate twice while respecting the boundary condition \eqref{SH:49},
\begin{equation*}
	\frac{\d}{\d s}\ln F_{\sigma}(s)=\int_s^{\infty}\left[\int_{-\infty}^{\infty}u^2(x|t)\,\d\sigma(t)\right]\d x,\ \ \ \ln F_{\sigma}(s)=-\int_s^{\infty}\left(\int_y^{\infty}\left[\int_{-\infty}^{\infty}u^2(x|t)\,\d\sigma(t)\right]\d x\right)\d y.
\end{equation*}
Finally rearrange the outer two integrals and summarize,
\begin{theo}[{\cite[Proposition $5.2$]{ACQ}}]\label{tb:2} Let $\d\sigma$ be a positive Borel probability measure on $\mathbb{R}$ with finite moments to all orders and which is absolutely continuous with respect to the Lebesgue measure. Then the Fredholm determinant $F_{\sigma}(s)$ of the operator $K_{\sigma}$ with kernel \eqref{SH:4} on $L^2(s,\infty)$ equals
\begin{equation*}
	F_{\sigma}(s)=\exp\left[-\int_s^{\infty}(x-s)\left(\int_{-\infty}^{\infty}u^2(x|t)\d\sigma(t)\right)\d x\right],\ \ s\in\mathbb{R},
\end{equation*}
where $u=u(s|t):\mathbb{R}^2\rightarrow\mathbb{R}$ solves the integro-differential Painlev\'e-II equation
\begin{equation}\label{SH:50}
	u_{ss}(s|t)=\left[s+t+2u(s|t)\int_{-\infty}^{\infty}u^2(s|\lambda)\,\d\sigma(\lambda)\right]u(s|t),
\end{equation}
subject to the boundary condition $u(s|t)\sim\textnormal{Ai}(s+t)$ as $s\rightarrow+\infty$ and $t\in\mathbb{R}$ fixed.
\end{theo}
We emphasize again that Theorem \ref{tb:2} was first proven in \cite{ACQ} through an application of the Tracy-Widom method \cite{TW}. Here we obtained the same result from the operator-valued Lax pair \eqref{SH:38}, \eqref{SH:39} and the underlying Hilbert boundary value problem \ref{HP8}. Note that our method, en route, established the following novel features of $u(s|t)$ - which should be compared with the ones for the Ablowitz-Segur transcendent $u(x)=u(x|(-\im,0,\im))$ in Subsection \ref{PIIconn}.
\begin{enumerate}
	\item[(1)] the above boundary value problem for \eqref{SH:50} is uniquely solvable, as $u(s|t)=n_2(0|t)$ is naturally constructed from \eqref{SH:23}, i.e. from the solution of the uniquely solvable Hilbert boundary value problem \ref{HP8}.
	\item[(2)] the map $\mathbb{R}\ni s\mapsto u(s|t)$ is smooth for any $t\in\mathbb{R}$ in light of \eqref{SH:44} and the analyticity properties of the resolvent. Equivalently because of \eqref{SH:23} and the analyticity properties of $X(z)=X(z;s)$, the solution of Problem \ref{HP8}.
\end{enumerate}

\begin{rem} A final remark is in order concerning the possibility to derive asymptotic expansions of $F_{\sigma}(s)$ as $s\rightarrow\pm\infty$ directly from Problem \ref{HP8}: while it is now fair to say that the operator-valued problem is well amenable to the rigorous derivation of dynamical systems, the development of a robust nonlinear steepest descent approach is work in progress by  the author. As of now, the only rigorous operator-valued asymptotic techniques available are in \cite{IKo} and \cite{IS}, for two particular integral kernels. In that sense, operator-valued Hilbert boundary value problems are not yet on the same level as their matrix-valued counterparts. However, if one accepts other techniques, then certain asymptotic information for $F_{\sigma}(s)$ has been derived recently, albeit for very specific measures $\d\sigma$, see \cite{CoG,Tsai,CC}.

\end{rem}

\section{Summary and further reading}\label{read}
This article is meant to provide a short survey of the theory of Riemann-Hilbert problems and their original appearances in mathematics. Particular emphasis is placed on Plemelj's 1908 work \cite{P0} on Hilbert's 21st problem and his use of Hilbert boundary value problems. Although Plemelj did not solve the problem for Fuchsian systems on $\mathbb{CP}^1$, his use of singular integral equations in the derivation of dynamical systems foreshadowed one of the key aspects of modern Riemann-Hilbert theory. Indeed, a Hilbert boundary value problem in modern terms is a Swiss army knife: once we know how a problem's observables can be characterized through a Hilbert boundary value problem then it presents us with a systematic tool to derive dynamical systems for the observables and with an efficient way to analyze them asymptotically - overall in complete analogy to classical contour integral methods. Still, the identification of a Riemann-Hilbert approach to a problem is non-systematic and we have displayed six different applications in the second half of the text. Much more than that has been done over the past decades and we would now like to give some guidance to the relevant Riemann-Hilbert literature - naturally based on the author's personal preferences: The monograph \cite{FIKN} is the standard reference concerning the Riemann-Hilbert approach to Painlev\'e equations and as a first taste to this subject we recommend the review articles \cite{Irev,Irev0} by Its. For a first introduction to orthogonal polynomial Riemann-Hilbert methods in random matrix theory we single out Deift's notes \cite{DRMT}, the notes by Bleher \cite{Ble}, by Its \cite{ICRM} and by Kuijlaars \cite{Kui}. For continued reading on orthogonal polynomial methods in more general ensembles than \eqref{e:36} we recommend Kuijlaars' survey \cite{Kui2}, the review article on universality \cite{Kui3} by the same author and the chapter on multi-matrix models by Bertola \cite{Ber}. The standard reference to integrable operators is \cite{D0}, the paper \cite{DIZ} and the somewhat non-rigorous monograph \cite{KBI}. Regarding Ulam's problem and other combinatorial problems amenable to Riemann-Hilbert methods, the recent monograph \cite{BDS} is likely going to become standard. Finally, and this is a topic which we did not touch upon, Hilbert boundary value problems have been identified as central tools in Gromov-Witten theories, i.e. in pure algebraic geometry. These theories aim to describe certain geometrical aspects of projective varieties in terms of invariants, suitably defined as intersection numbers on moduli spaces of curves. The invariants themselves are encoded in their generating functions whose convergence domains carry rich geometrical structures, nowadays known as Dubrovin-Frobenius manifolds. It was Dubrovin \cite{D1}, see also \cite{D2}, who first recognized the possibility to reconstruct such a manifold via a Hilbert boundary value problem. Further refinements of his procedure are given in the recent works of Cotti, Dubrovin, Guzzetti \cite{CDG1,CDG2}, sadly some of the last papers coauthored by Boris Dubrovin.

\end{document}